\documentclass[3p]{elsarticle}
\usepackage[algoruled]{algorithm2e}
\usepackage[colorlinks]{hyperref}
\usepackage[usenames, dvipsnames]{xcolor}
\usepackage[utf8]{inputenc}

\usepackage{amsfonts}
\usepackage{tikz}
\usetikzlibrary{matrix,arrows.meta}

\usepackage{amsfonts}
\usepackage{graphicx}
\usepackage{epstopdf}
\usepackage{algorithmic}
\usepackage{tikz}
\usepackage{amsmath}
\usepackage{amssymb}

\usetikzlibrary{shapes,arrows}

\ifpdf
  \DeclareGraphicsExtensions{.eps,.pdf,.png,.jpg}
\else
  \DeclareGraphicsExtensions{.eps}
\fi

\usepackage{enumitem}
\setlist[enumerate]{leftmargin=.5in}
\setlist[itemize]{leftmargin=.5in}

\usepackage{amsopn}
\DeclareMathOperator{\diag}{diag}

\foreach \x in {a,...,z}{
    \expandafter\xdef\csname bf\x \endcsname{\noexpand\ensuremath{\noexpand\mathbf{\x}}}
    \expandafter\xdef\csname bm\x \endcsname{\noexpand\ensuremath{\noexpand\boldsymbol{\x}}}
  }
\foreach \x in {A,...,Z}{
    \expandafter\xdef\csname bs\x \endcsname{\noexpand\ensuremath{\noexpand\boldsymbol{\x}}}
    \expandafter\xdef\csname bf\x \endcsname{\noexpand\ensuremath{\noexpand\mathbf{\x}}}
    \expandafter\xdef\csname bb\x \endcsname{\noexpand\ensuremath{\noexpand\mathbb{\x}}}
    \expandafter\xdef\csname ds\x \endcsname{\noexpand\ensuremath{\noexpand\mathds{\x}}}
    \expandafter\xdef\csname cal\x \endcsname{\noexpand\ensuremath{\noexpand\mathcal{\x}}}
  }

\DeclareMathOperator{\trace}{Tr}

\DeclareMathOperator{\vectorize}{vec}

\DeclareMathOperator{\rank}{rank}

\def\transp{\top}
\def\herm{\mathsf{H}}

\DeclareMathOperator*{\argmin}{arg\,min}

\DeclareMathOperator*{\real}{real}
\DeclareMathOperator*{\imag}{imag}
\DeclareMathOperator*{\prox}{\textbf{prox}}

\newcommand{\conj}[1]{\overline{#1}}
\renewcommand{\bmj}{\jmath}

\usepackage{amsthm}
\newtheorem{theorem}[]{Theorem}
\newtheorem{remark}[]{Remark}

\newtheorem{proposition}[]{Proposition}
\newtheorem{example}[]{Example}
\newtheorem{lemma}[]{Lemma}

\newcommand{\fourier}[1]{\mathfrak{#1}}
\newcommand{\fourierv}[1]{\boldsymbol{\mathfrak{#1}}}
\newcommand{\vect}[1]{{#1}_{\mathrm{vec}}}
\newcommand{\vectx}{\boldsymbol{\xi}}
\newcommand{\vectX}{\boldsymbol{\Xi}}

\newcommand{\vectz}{\boldsymbol{\psi}}
\newcommand{\vectZ}{\boldsymbol{\Psi}}

\def\multmat#1#2{\mathbf{M}_{#2}(\mathbf{#1})}
\def\pco#1#2{#1_{#2}}
\def\sylv#1{\mathcal{S}_{#1}}
\def\hankel#1{\mathcal{H}_{#1}}

\DeclareMathOperator{\colspan}{colspan}

\newcommand{\FourierGamma}{\boldsymbol{\mathfrak{F}}}
\newcommand{\ie}{\emph{i.e.}~}

\usepackage[capitalize]{cleveref}

\begin{document}

\begin{frontmatter}

  \title{Polarimetric phase retrieval:\\ uniqueness and algorithms\tnoteref{funding}}

  \author[cranaddress]{Julien Flamant\corref{mycorrespondingauthor}}
  \cortext[mycorrespondingauthor]{Corresponding author}
  \ead{julien.flamant@cnrs.fr}
  \author[cranaddress]{Konstantin Usevich}
  \ead{konstantin.usevich@univ-lorraine.fr}
  \author[iecladdress]{Marianne Clausel}
  \ead{marianne.clausel@univ-lorraine.fr}
  \author[cranaddress]{David Brie}
  \ead{david.brie@univ-lorraine.fr}
  \address[cranaddress]{CNRS, Université de Lorraine, CRAN, F-54000 Nancy France}
  \address[iecladdress]{CNRS, Université de Lorraine, Institut Elie Cartan de Lorraine, F-54000 Nancy France}
  \tnotetext[funding]{This work was funded by CNRS and GdR ISIS under the 2019-2021 OPENING exploratory research project grant.}

  \begin{abstract}
    This work introduces a novel Fourier phase retrieval model, called \emph{polarimetric phase retrieval} that enables a systematic use of polarization information in Fourier phase retrieval problems.
    We provide a complete characterization of uniqueness properties of this new model by unraveling equivalencies with a peculiar polynomial factorization problem.
    We introduce two different but complementary categories of reconstruction methods.
    The first one is algebraic and relies on the use of approximate greatest common divisor computations using Sylvester matrices.
    The second one carefully adapts existing algorithms for Fourier phase retrieval, namely semidefinite positive relaxation and Wirtinger-Flow, to solve the polarimetric phase retrieval problem.
    Finally, a set of numerical experiments permits a detailed assessment of the numerical behavior and relative performances of each proposed reconstruction strategy.
    We further highlight a reconstruction strategy that combines both approaches for scalable, computationally efficient and asymptotically MSE optimal performance.
  \end{abstract}


\end{frontmatter}

\section{Introduction}
\label{sec:introduction}

The problem of Fourier phase retrieval, \ie the recovery of a signal given the magnitude of its Fourier transform, has a long and rich history dating back from the 1950s \cite{sayre_implications_1952}.
The Fourier phase retrieval problem has been -- and continues to be -- of tremendous importance for many applications areas involving optics, such as crystallography \cite{elser2003phase,elser2018benchmark,millane1990phase}, astronomy \cite{fienup_reconstruction_1978,Fienup:93}, coherent diffraction imaging (also known as lensless imaging) \cite{miao1999extending,maiden2009improved,shechtman_phase_2015}, among others.
Such problem arises in optics since \emph{phase information} of light cannot be measured directly due to the high oscillating frequency of the electromagnetic field: indeed there is no conventional detector that can sample at a rate of $\sim 10^{12}$ Hz (infrared) up to $\sim 10^{18}$ Hz (hard x-rays).
This means that in such imaging applications, only intensity measurements can be performed, and that the phase should be recovered numerically afterwards.
Moreover, during the last decade, the phase retrieval problem has gained a lot of interest in the signal processing and applied mathematics community \cite{balan_signal_2006,candes_phase_2011,candes_phase_2013,bandeira2014phase}.
However, it is important to note that most of works in this community focus on generalized phase retrieval problems, where Fourier measurements are replaced or combined with random projections.
While this allows the derivation of several important results using probabilistic considerations, e.g. uniqueness or stability guarantees, these results are not directly applicable to the original  (deterministic)  Fourier phase retrieval problem.
Indeed, it is well known that one-dimensional univariate Fourier phase retrieval does not admit a unique solution in general \cite{beinert2015ambiguities}.
We refer the reader to \cite{boche_fourier_2017} for a recent review of proposed (deterministic) strategies to recover uniqueness of Fourier phase retrieval as well as associated algorithms.

Just like color (wavelength), \emph{polarization} is a fundamental property of light.
It encodes the geometry of oscillations of the electromagnetic field, which describes an ellipse in the 2D plane perpendicular to the propagation direction for vacuum-like media \cite{chipman_polarized_2018}.
As polarized light propagates in media, its polarization can change, thus revealing key properties, such as medium anisotropy or architectural order that are inaccessible to conventional, non-polarized light \cite{perez_polarized_nodate}.
As a result, polarized light imaging has found many applications such as in material characterization \cite{Gordon2000}, remote sensing \cite{tyo2006review} or bio-imaging \cite{guo_revealing_2020}.
Despite the important practical interests of polarization, only a few authors have considered leveraging this fundamental attribute of light in phase retrieval problems.
The authors in \cite{smirnova_attosecond_2009,raz_vectorial_2011} pioneered the use of polarization in Fourier phase retrieval in the context of ultrashort (e.g. attosecond, $\sim 10^{-18}s$) laser pulse characterization.
Figure \ref{fig:ultrashortExp} depicts a simplified experimental setup of such an experiment.
The goal here is to recover the polarized pulse after the medium, given \emph{polarimetric} projections of its Fourier transform -- recorded by the spectrometer.
Comparing the reconstructed output pulse with the input laser pulse, one is able to recover key anisotropic properties of the studied medium.
More recently, authors have developed \emph{vectorial ptychography} \cite{ferrand_ptychography_2015,ferrand_quantitative_2018}, a promising lensless imaging technique that simultaneously uses polarization and tilted measurements.
This allows quantitative imaging of complex anisotropic media, such as biominerals \cite{baroni_extending_2020,baroni_reference-free_2020}.

This work introduces a novel Fourier phase retrieval model, called \emph{polarimetric phase retrieval} that enables a systematic use of polarization information in Fourier phase retrieval problems.
The rationale is the following: we consider the \emph{polarimetric phase retrieval} problem as the problem of recovering the 1D bivariate signal representing polarized light from scalar quadratic Fourier magnitude measurements.
Importantly, this new model leverages physical acquisition schemes relevant to polarization measurement.
Notably, it encompasses as special cases previous models proposed in the literature \cite{jaganathan_reconstruction_2019,raz_vectorial_2013} originally developed to account for vectorial-like diversity in phase retrieval.
Our contributions can be stated as follows.
We first provide a complete characterization of uniqueness properties of the polarimetric phase retrieval model, by unraveling equivalencies with a peculiar polynomial factorization problem.
Notably, we show that unlike standard 1D Fourier phase retrieval, almost all bivariate signals can be uniquely recovered from polarimetric Fourier measurements.
We also introduce two different but complementary categories of reconstruction methods.
The first one is algebraic and relies on the use of approximate greatest common divisor computations using Sylvester matrices.
The second one carefully adapts existing algorithms for Fourier phase retrieval, namely semidefinite positive relaxation and Wirtinger-Flow, to solve the polarimetric phase retrieval problem.
Finally, a set of numerical experiments permits a thorough assessment of the numerical behavior and relative performances of each proposed reconstruction strategy.
We further highlight a reconstruction strategy that combines both approaches for scalable, computationally efficient and asymptotically MSE optimal performance.

This paper is organized as follows.
Section \ref{sec:model} introduces the polarimetric phase retrieval model and discusses its equivalent formulations as well as trivial ambiguities.
Section \ref{sec:uniqueness} provides a complete study of the uniqueness properties of the polarimetric phase retrieval model, by leveraging a polynomial factorization representation of the problem.
Section \ref{sec:solvingPPR_algebraic} exploits uniqueness results to propose two algebraic reconstruction methods based on approximate greatest common divisor computations.
Section \ref{sec:solvingPPR_iterative} takes a complementary path, by developing two iterative algorithms to solve the polarimetric phase retrieval problem.
Section \ref{sec:numericalExperiments} details several numerical experiments to illustrate and assess the practical performances of the proposed reconstruction algorithms.
Section \ref{sec:conclusion} collects concluding remarks and Appendices gather technical details and supplementary results.

\begin{figure}[t]
  \centering
  \includegraphics[width=.8\textwidth]{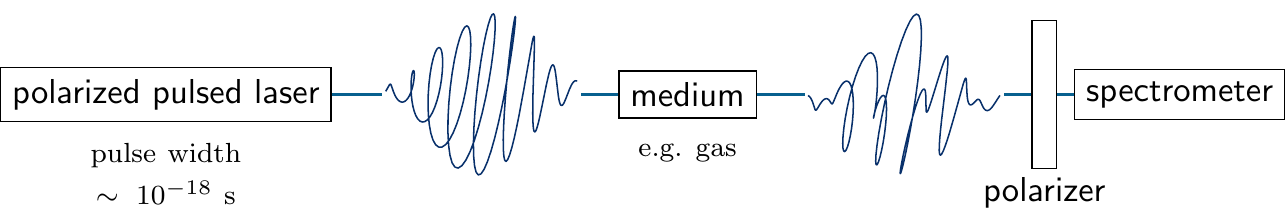}
  \caption{Simplified experimental setup for measuring ultrashort (e.g. a few attoseconds $10^{-18}$ s) electromagnetic polarized pulses.
    Inspired by \cite{smirnova_attosecond_2009,raz_vectorial_2011}.}
  \label{fig:ultrashortExp}
\end{figure}

\section{Polarimetric phase retrieval model}
\label{sec:model}

\subsection{General formulation}
\label{sub:generalFormulationPPR}

Consider a discrete bivariate signal $\bfx[n] = (x_1[n], x_2[n])^\transp \in \bbC^2$ defined for $n=0, 1, \ldots N-1$.
Let $\bfX \in \bbC^{N\times 2}$ be its matrix representation obtained by stacking samples row-wise, \ie such that
\begin{equation}
  \bfX \triangleq \begin{bmatrix}
    x_1[0]   & x_2[0]   \\
    x_1[1]   & x_2[1]   \\
    \vdots   & \vdots   \\
    x_1[N-1] & x_2[N-1]
  \end{bmatrix}\:.
\end{equation}
We define the \emph{polarimetric phase retrieval} (PPR) problem as the problem of recovering the bivariate signal $\lbrace \bfx[n]\rbrace_{n=0, 1, \ldots, N-1}$ given scalar quadratic Fourier measurements.
Formally,
\begin{equation}
  \begin{split}
    \text{find } \bfX \in \bbC^{N\times 2} \text{ given measurements } y_{m, p} = \left\vert \bfa_m^\herm \bfX\bfb_p\right\vert^2\\
    m=0, 1, \ldots M-1, \quad  p = 0, 1, \ldots P-1
  \end{split}
  \tag{\sf{PPR}} \label{prob:PPR}
\end{equation}
where $\bfa_m \in \bbC^N$ is the discrete Fourier vector corresponding to frequency $\frac{2\pi m}{M}$, \ie $a_m[n] = \exp(\bmj\frac{2\pi m}{M}n)$ for $n=0, 1, \ldots N-1$, and the $\bfb_p \in \bbC^2$ denote $P$ arbitrary projection vectors which are supposed to be normalized such that $\Vert \bfb_p \Vert_2^2 = 1$.
\ref{prob:PPR} measurements encode the physics of the acquisition in coherent diffraction imaging, where only intensity measurements can be performed: Fourier vectors $\bfa_m$ model Fraunhofer diffraction, whereas the $\bfb_p$'s represent the different polarizers (or polarization analysers) required to measure polarization.
This agrees completely with the experimental setup described in Figure \ref{fig:ultrashortExp} in the context of ultra-short polarized electromagnetic pulse characterization.
Note also that while we focus here on \emph{physically realizable} measurement schemes, there is no obstacle from a mathematical viewpoint to extend the measurement scheme of \ref{prob:PPR} can be extended to arbitrary $\bfa_m$ sensing vectors.
This includes, for instance, random gaussian vectors as in \emph{generalized phase retrieval} problems, see \emph{e.g.} \cite{candes_phase_2011,balan_signal_2006,bandeira_saving_2014} to cite only a few.

\subsection{Relation with Fourier matrix measurements}
\label{sub:matrixBPR}

A closely related problem to \ref{prob:PPR} is the \emph{bivariate phase retrieval} (BPR) problem.
Let us introduce the discrete Fourier transform of the bivariate signal $\lbrace\bfx[n]\rbrace_{n=0, 1, \ldots, N-1}$ as
\begin{equation}\label{eq:defDFT}
  \fourierv{X}[m]  \triangleq \sum_{n=0}^{N-1}\bfx[n]\exp\left(-2\pi \bmj \frac{mn}{M}\right) = \begin{bmatrix}
    \fourier{X}_1[m] \\ \fourier{X}_2[m]
  \end{bmatrix} = (\bfa_m^\herm \bfX)^\transp \in \bbC^2
\end{equation}
for $m = 0, 1, \ldots M-1$.
Then let $\FourierGamma[m]$ denote the rank-one 2-by-2 complex spectral matrix at frequency indexed by $m$,
\begin{equation}\label{eq:Gamma_m}
  \FourierGamma[m] \triangleq \fourierv{X}[m]\fourierv{X}[m]^\herm = \begin{bmatrix}
    \vert \fourier{X}_1[m]\vert^2           & \fourier{X}_1[m]\conj{\fourier{X}_2[m]} \\
    \fourier{X}_2[m]\conj{\fourier{X}_1[m]} & \vert \fourier{X}_2[m]\vert^2
  \end{bmatrix} \in \bbC^{2\times 2}\:.
\end{equation}
For each $m$, the spectral matrix $\FourierGamma[m]$ collects the squared amplitude of Fourier transforms of the two components $x_1[n]$ and $x_2[n]$ of the bivariate signal $\bfx[n]$ as well as their relative Fourier phase.
BPR is then formulated as the problem of recovering the original bivariate signal from its spectral matrices, that is:
\begin{equation}
  \begin{split}
    \text{find } \bfX \in \bbC^{N\times 2} \text{ given measurements } \FourierGamma[m], \quad
    m=0, 1, \ldots M-1\:.
  \end{split}
  \tag{\sf{BPR}} \label{prob:BPR}\
\end{equation}

Proposition \ref{prop:equivalenceBPRPPR} below shows that \ref{prob:BPR} and \ref{prob:PPR} are equivalent in the noiseless setting under very general assumptions on the projection vectors $\bfb_p$.

\begin{proposition}
  \label{prop:equivalenceBPRPPR}
  Suppose that the collection of projection vectors $\bfb_0, \bfb_1, \ldots \bfb_{P-1} \in \bbC^2$ satisfies the condition
  \begin{equation}
    \mathrm{span}_{\bbR} \left\lbrace \bfb_p\bfb_p^\herm\right\rbrace_{p=0, 1, \ldots P-1}= \left\lbrace \bfM \in \bbC^{2\times 2} \mid \bfM^\herm = \bfM \right\rbrace \tag{\calH} \label{assumption:BPRPPR}
  \end{equation}
  \ie, the $P$ rank-one matrices $\bfb_p\bfb_p^\herm$ are a generating family over $\bbR$ of the space of 2-by-2 Hermitian matrices.
  Then, under assumption \eqref{assumption:BPRPPR}, the problem \ref{prob:PPR} is equivalent to \ref{prob:BPR} in the sense that $\bfX$ is a solution of the problem \ref{prob:PPR} if and only if it is solution of \ref{prob:BPR}.
\end{proposition}
\begin{proof}
  We first show that that under \eqref{assumption:BPRPPR}, measurements $y_{m,p}$ in \ref{prob:PPR} can be directly expressed in terms of spectral matrices $\FourierGamma[m]$, which shall imply that any solution of the problem \ref{prob:PPR} is a solution of the problem~\ref{prob:BPR}.
  Let us fix $0 \leq m \leq M-1$.
  Then one has for every $0 \leq p \leq P-1$
  $$y_{m,p} = \vert \bfa_m^\herm \bfX \bfb_p\vert^2 = \fourierv{X}[m]^\transp \bfb_p \bfb_p^\herm \conj{\fourierv{X}[m]} = \trace \conj{\bfb_p}\bfb_p^\transp \FourierGamma[m].$$
  Conversely, let us assume that
  $\lbrace \conj{\bfb_p}\bfb_p^\transp\rbrace_{p=0, 1, \ldots, P-1}$ (or equivalently, the set $\left\lbrace \bfb_p\bfb_p^\herm\right\rbrace_{p=0, 1, \ldots P-1}$) is a generating family of the space of 2-by-2 Hermitian matrices.
  Then there exist $\lambda_{m, p} \in \bbR$, $=0, 1, \ldots, M-1$, $p=0, 1, \ldots P-1$  such that
  $$ \FourierGamma[m] = \sum_{p=0}^{P-1}\lambda_{m,p} \conj{\bfb_p}\bfb_p^\transp,$$
  meaning that the spectral matrices $\FourierGamma[m]$ can be readily obtained from scalar polarimetric projections $y_{m,p}$.
  It implies that any solution of the problem \ref{prob:BPR} is a solution of the problem~\ref{prob:PPR}.
  Gathering the two parts of the proof yields Proposition~\ref{prop:equivalenceBPRPPR}.
\end{proof}
\begin{example}
  Let $P=4$ and consider the following projection vectors
  \begin{equation}
    \bfb_0 = \begin{bmatrix}
      1 \\ 0
    \end{bmatrix}, \:     \bfb_1 = \begin{bmatrix}
      0 \\ 1
    \end{bmatrix}, \:
    \bfb_2 = \frac{1}{\sqrt{2}}\begin{bmatrix}
      1 \\ 1
    \end{bmatrix},\:
    \bfb_3 = \frac{1}{\sqrt{2}}\begin{bmatrix}
      1 \\ \bmj
    \end{bmatrix}\:. \label{eq:projectionSchemeSimple}
  \end{equation}
  A direct check shows that rank-one matrices $\bfb_0\bfb_0^\herm$, $\bfb_1\bfb_1^\herm$, $\bfb_2\bfb_2^\herm$, $\bfb_3\bfb_3^\herm$ form a basis over the real vector space of 2-by-2 Hermitian matrices, and as a result, they are a generating family of such matrices.
  Polarimetric measurements in \ref{prob:PPR} read explicitly
  \begin{equation}
    \begin{split}
      &y_{m, 0} = \left\vert \fourier{X}_1[m]\right\vert^2,\quad y_{m, 1} = \left\vert \fourier{X}_2[m]\right\vert^2\\
      &y_{m, 2} =  \frac{1}{2}\left\vert \fourier{X}_1[m] + \fourier{X}_2[m]\right\vert^2,\quad y_{m, 3} =\frac{1}{2} \left\vert \fourier{X}_1[m]+\bmj\fourier{X}_2[m]\right\vert^2
    \end{split}\label{eq:measurementSchemeSimple}
  \end{equation}
  These expressions give directly the diagonal terms of $\FourierGamma[m]$ as $y_{m, 0}$ and $y_{m, 1}$.
  The off-diagonals terms can be recovered easily using polarization identities in the complex case, such that
  \begin{align}
    \real\left(\fourier{X}_1[m]\conj{\fourier{X}_2[m]}\right) & = \frac{1}{2}\left(\left\vert \fourier{X}_1[m] + \fourier{X}_2[m]\right\vert^2-\left\vert \fourier{X}_1[m]\right\vert^2-\left\vert \fourier{X}_2[m]\right\vert^2\right)  = y_{m, 2}- \frac{1}{2}\left(y_{m, 0}+y_{m, 1}\right)\:,                                                                                                              \\
    \imag\left(\fourier{X}_1[m]\conj{\fourier{X}_2[m]}\right) & = \frac{1}{2}\left(\left\vert \fourier{X}_1[m] + \bmj\fourier{X}_2[m]\right\vert^2 -\left\vert \fourier{X}_1[m]\right\vert^2-\left\vert \fourier{X}_2[m]\right\vert^2\right) = y_{m, 3}- \frac{1}{2}\left(y_{m, 0}+y_{m, 1}\right)\:.
  \end{align}
  Remark that the measurement scheme \eqref{eq:projectionSchemeSimple} yield the same quadratic measurements \eqref{eq:measurementSchemeSimple} as proposed by the authors in \cite{raz_vectorial_2013,jaganathan_reconstruction_2019}.
  This shows that \ref{prob:PPR} encompasses existing measurements strategies as a special case, while bringing extra flexibility in the experimental design of measurements.
\end{example}

\subsection{Trivial ambiguities}
\label{sub:trivialAmbiguities}
Thanks to Proposition~\ref{prop:equivalenceBPRPPR}, we can now give a characterization of trivial ambiguities of \ref{prob:PPR} model by leveraging the equivalent \ref{prob:BPR} problem.
Indeed, one can investigate in a rather simple way the trivial ambiguities that characterize \ref{prob:BPR}.
For ease of presentation, let us extend to all $n \in \bbZ$ the bivariate signal $\lbrace \bfx[n]=(x_1[n], x_2[n])^\transp\rbrace_{n=0,\ldots, N-1}$ by zero padding for $n < 0$ and $n\geq N$.
Consider the BPR measurement matrix $\FourierGamma[m]$ defined in  \eqref{eq:Gamma_m} for an arbitrary frequency indexed by $m\in\mathbb{Z}$.
Our goal in this section consists in identifying trivial operations on $\lbrace \bfx[n]\rbrace_{n\in \bbZ}$ that leave the measurements $\lbrace\FourierGamma[m]\rbrace_{m\in \bbZ}$ unchanged.

\paragraph{Global phase ambiguity} Let $\alpha \in \bbR$ and consider the bivariate signal $\bfx'$ such that $\bfx'[n] = \exp(\bmj \alpha) \bfx[n]$ for every $n$.
Then, one has, for any $m \in \bbZ$, $\FourierGamma'[m] = \FourierGamma[m]$ since $\fourier{X}_i'[m] = \exp(\bmj \alpha) \fourier{X}_i[m] $ for $i=1, 2$.

\paragraph{Time shift} Consider the time shifted signal $\bfx'$ such that $x'_1[n] = x_1[n-n_{1}]$ and $x'_2[n] = x'_2[n-n_{2}]$. Then $\FourierGamma'[m] = \FourierGamma[m]$ iff time-shifts are equal $n_{1}= n_{2}$.

\paragraph{Conjugate reflection}
Consider now $\bfx'$ such that $x'_1[n] = \overline{x_1[N-n]}$ and $x'_2[n] = \overline{x_2[N-n]}$.
Then one has for every $m$
\begin{equation}
  \FourierGamma'[m] = \begin{bmatrix}\vert \fourier{X}_1[m]\vert^2               & \fourier{X}_2[m]\overline{\fourier{X}_1}[m] \\
               \fourier{X}_1[m]\overline{\fourier{X}_2}[m] & \vert \fourier{X}_2[m]\vert^2
  \end{bmatrix} = \FourierGamma[m]^\transp.
\end{equation}
This shows that conjugate reflection is not, in general, a trivial ambiguity for complex bivariate phase retrieval.
This contrasts with standard univariate phase retrieval, see \cite{boche_fourier_2017,beinert2015ambiguities}.

Conjugate reflection can still be a trivial ambiguity provided that the measurement matrix is symmetric for every $m$, that is $\FourierGamma[m] = \FourierGamma[m]^\transp$.
Equivalently, $\FourierGamma[m]$ is symmetric iff $\fourier{X}_1[m]\overline{\fourier{X}_2}[m] = \fourier{X}_2[m]\overline{\fourier{X}_1}[m]$.
This means that $\imag\left(\fourier{X}_1[m]\overline{\fourier{X}_2}[m] \right)=0$, \ie components $\fourier{X}_1[m]$, $\fourier{X}_2[m]$ are in phase at every frequency.
Interestingly, this condition is interpreted in physical terms as: conjugate reflection is a trivial ambiguity for bivariate phase retrieval iff $\bfx$ is linearly polarized at all frequencies.

\subsection{1D equivalent model for \ref{prob:PPR}}
\label{sub:1DequivalentModelphysicsModel}

Back to the original problem \ref{prob:PPR}, we see that it defines a new measurement model that perform quadratic scalar projections of the matrix representation $\bfX \in \bbC^{N\times 2}$ of the bivariate signal of interest.
This \emph{matrix representation} of the underlying signal $\lbrace \bfx[n]\rbrace_{n=0, 1, \ldots N-1}$ can be confusing at first: indeed, the bivariate signal is intrinsically one-dimensional, in the sense that it is a function of a single index $n$ -- which can represent time or 1D spatial coordinates, for instance.
Thus, a natural question is the following:
can the \ref{prob:PPR} problem be equivalently rewritten as a one-dimensional phase retrieval problem? If so, what is the physical interpretation of such problem?

Let us denote by $\vectx \triangleq \vectorize \bfX \in \bbC^{2N}$ the long vector obtained by stacking the two columns of $\bfX$.
Using standard properties of matrix products vectorization, one can rewrite \ref{prob:PPR} measurements as
\begin{equation}
  y_{m, p} = \vert \bfa_m^\herm \bfX \bfb_p\vert^2 = \vert (\bfb_p^\transp \otimes \bfa_m^\herm) \vectx\vert^2 = \vert (\conj{\bfb_p} \otimes \bfa_m)^\herm \vectx\vert^2
  \label{eq:PPRmeasurements1Dequivalence}
\end{equation}
for $m=0, 1, \ldots M-1$,  $p = 0, 1, \ldots P-1$ and where $\bfa \otimes \bfb$ stands for the Kronecker product of vectors $\bfa$ and $\bfb$.
Letting $\bfc_{m, p} \triangleq  \conj{\bfb_p} \otimes \bfa_m \in \bbC^{2N}$, the \ref{prob:PPR} problem is equivalent to
\begin{equation}
  \begin{split}
    \text{find } \vectx \in \bbC^{ 2N} \text{ given measurements } y_{m, p} = \left\vert \bfc_{m, p}^\herm \vectx\right\vert^2\\
    m=0, 1, \ldots M-1, \quad  p = 0, 1, \ldots P-1
  \end{split}
  \tag{\sf{PPR-1D}} \label{prob:PPR-1D}
\end{equation}
This shows that \ref{prob:PPR} can be rewritten as a specific instance of 1D phase retrieval with structured measurements vectors $\bfc_{m, p} \in \bbC^{2N}$ called \ref{prob:PPR-1D}.
While being mathematically sound, this equivalent 1D problem brings almost no insights about the bivariate nature of the signal to be recovered.
Moreover, \ref{prob:PPR-1D} cannot be interpreted as a Fourier phase retrieval problem with masks \cite{bandeira2014phase, jaganathan_phase_20152}, since measurements vectors $\bfc_{m,p}$ intertwine Fourier measurements $\bfa_m$ and polarimetric projections $\bfb_p$ using Kronecker product.
Thus, the study of the theoretical properties of \ref{prob:PPR} (and \ref{prob:BPR}) cannot be inferred from standard phase retrieval properties applied to \ref{prob:PPR-1D}.
This requires a dedicated study, which is described in detail in Section \ref{sec:uniqueness}.
Nonetheless, as we shall see in Section \ref{sec:solvingPPR_iterative}, the equivalent formulation \ref{prob:PPR-1D} can still be particularly useful for designing algorithms to solve the original \ref{prob:PPR} problem.

\section{Uniqueness and polynomial formulation}
\label{sec:uniqueness}
This section studies the uniqueness properties of noiseless \ref{prob:PPR} under the set of assumptions \eqref{assumption:BPRPPR} defined in Section~\ref{sub:matrixBPR}.
Thanks to Proposition~\ref{prop:equivalenceBPRPPR}, we see that any solution of the problem \ref{prob:PPR} is a solution of the problem \ref{prob:BPR}, and vice-versa.
This formal equivalence permits to study uniqueness properties of the original \ref{prob:PPR} problem by leveraging its bivariate equivalent \ref{prob:BPR}.
Pushing this idea further, we reformulate \ref{prob:BPR} using a polynomial formalism.
Thus, the problem \ref{prob:BPR} appears as a crucial bridge that enables us to study the uniqueness conditions as uniqueness of a certain factorization of polynomials.
This idea is classic  \cite{boche_fourier_2017,beinert2015ambiguities}, but unlike the previous works, we allow for having roots at infinity, which enables us to establish the one-to-one correspondence between the two formulations, as well as a complete characterization of the uniqueness properties of  \ref{prob:BPR} (and by equivalence, to those of \ref{prob:PPR}) for any signals of bounded support.
In particular, this gives another view on the ambiguities in the phase retrieval problem.

\subsection{Polynomials with the roots at infinity, and operations with them}\label{sub:polynomials}
In this paper, we work with polynomials which may have possibly roots at $\infty$.
We briefly review properties of such polynomials, leaving more details in \ref{app:polynomials} (we also refer an interested reader to \cite[\S I.0]{heinig1984algebraic}, where such polynomials are used to build the algebraic theory of Hankel matrices).
Formally, we define $\bbC_{\le D}[z]$ to be the space of polynomials of degree at most $D$
\[
  A(z) = \sum\limits_{n=0}^{D} a[n] z^n.
\]
defined by a vector of coefficients $\bfa = \begin{bmatrix}a[0] & a[1] & \cdots & a[D] \end{bmatrix}^{\top} \in \bbC^{D+1}$.
We will say that the polynomial has root at $\infty$ (with multiplicity $\mu_k$) if its leading coefficient vanishes (\ie if $a[D]=\cdots =a[D-\mu_k+1]=0$).
With such a convention, the following extended version of the fundamental theorem of algebra holds true:
any nonzero polynomial $A \in \bbC_{\le D}[z]$ can be uniquely (up to permutation of roots) factorized  as
\begin{equation}\label{eq:polyFactorizationFTA}
  A(z) = \lambda \prod_{i=1}^{m}  (z-\alpha_{i})^{\mu_i}
\end{equation}
where $\lambda \in \bbC$, $\alpha_i \in \bbC \cup \{ \infty \}$ are distinct roots  and $\mu_i$ are the multiplicities of $\alpha_i$,
so that their sum  is
\[
  \mu_1 + \cdots +\mu_m =D;
\]
in addition, the multiplication by $(z - \infty)^{d}$ formally means that $d$ leading zero coefficients are appended (see also \ref{app:polynomials} for a formal definition).
\begin{example}\label{ex:simple_poly}
Consider the following polynomial from  $\bbC_{\le 5}[z]$:
\begin{equation}
  \label{eq:ex_poly}
  A(z) = 0 \cdot z^5 + 0 \cdot z^4 + \frac{1}{2} z^3 + \frac{1}{2} z^2 - z  \in \bbC_{\le 5}[z].
\end{equation}
This polynomial has roots $\{\infty, -2, 1, 0\}$, where  the root $\infty$ has multiplicity $2$.
Hence it has the following factorization
\[
A(z) = \frac{1}{2} (z-\infty)^2 (z - 1) (z+2) z.
\]
\end{example}
We will also use the operation of conjugate reflection of $A(z) \in \bbC_{\le D}[z]$,  defined as
\begin{equation}\label{eq:conjReverse}
  \widetilde{A}(z) = z^{D} \overline{A(\overline{z}^{-1})}  =  \sum\limits_{n=0}^{D} \overline{a[D-n]} z^n.
\end{equation}
Then it is easy to see that the conjugate reflection of the polynomial \eqref{eq:polyFactorizationFTA} admits a factorization
\[
\widetilde{A}(z) = \widetilde{\lambda} \prod_{i=1}^{m}  \left(z-\overline{\alpha^{-1}_{i}}\right)^{\mu_i}, \quad \text{where } \widetilde{\lambda} \triangleq \overline{\lambda}  \prod_{\substack{i =1\\ \alpha_i \neq \infty}}^m (-\overline{\alpha_{i}})^{\mu_i},
\]
\ie the roots $\alpha_i$ are mapped to $\overline{\alpha^{-1}_i}$, where $0$ is formally assumed to be the inverse of $\infty$ and vice versa.
\begin{example}
For \Cref{ex:simple_poly}, the conjugate reflection $\widetilde{A}(z)  \in \bbC_{\le 5}[z]$, as well as its factorization  becomes:
\[
\widetilde{A}(z)  = 0 \cdot z^5 -  z^4 + \frac{1}{2} z^3 + \frac{1}{2} z^2  = (-1) (z - \infty) (z+\frac{1}{2})(z-1)z^2.
\]
which  has roots $\{\infty, 1, -\frac{1}{2}, 0\}$, where the root $0$ has multiplicity $2$.
\end{example}
Graphically,  the conjugate reflection of the roots has a nice interpretation in terms of the Riemann sphere: the mapping of the root under conjugate reflection becomes simply a reflection with respect to the plane passing through the equator, see Fig.~\ref{fig:RiemannSphere}.

\begin{figure}[ht!]
  \includegraphics[width=\textwidth]{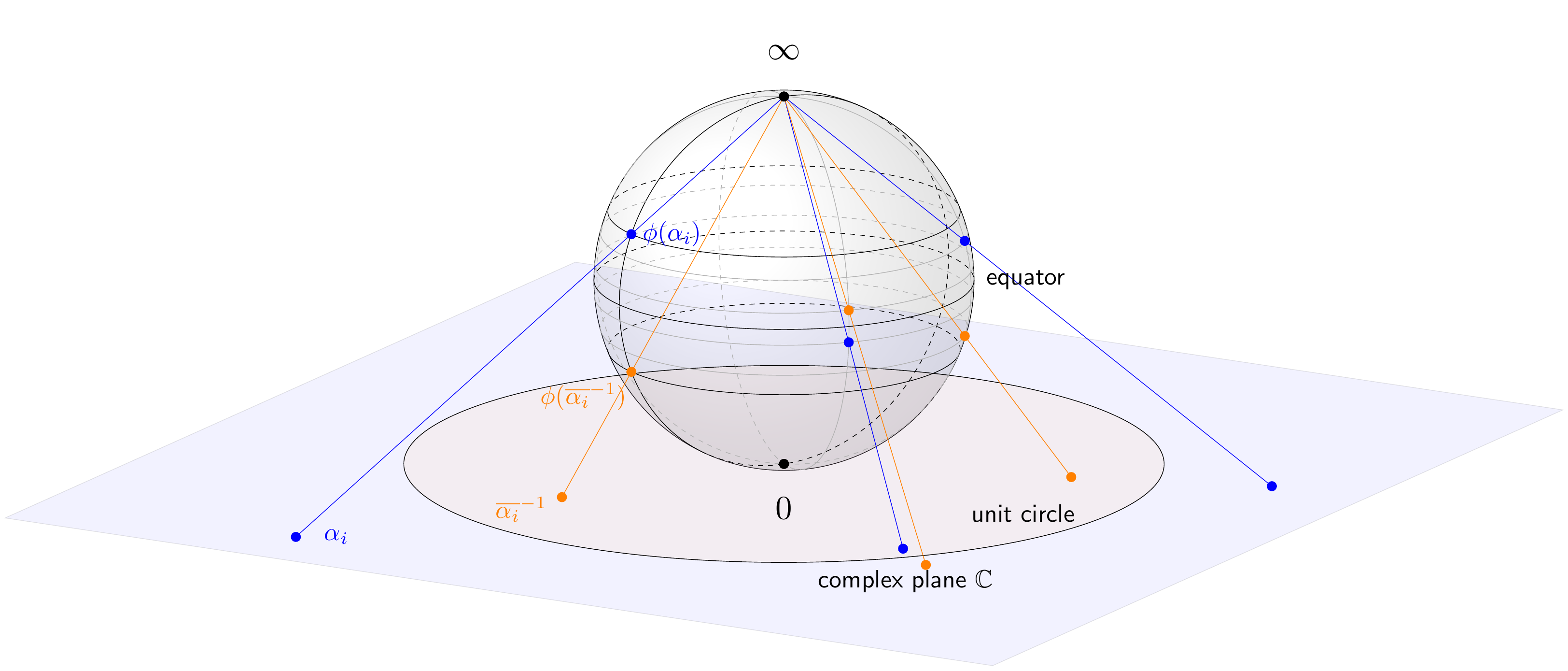}
  \caption{Complex plane and the Riemann sphere (the preimage under the stereographic projection). The conjugate inversion corresponds to reflection with respect to the equator on the Riemann sphere. Here $\phi: \bbC \to \calS^2$ denote the inverse stereographic mapping onto the sphere $\calS^2$.}\label{fig:RiemannSphere}
\end{figure}

Finally, we need to be careful when speaking about multiplication and divisors in such polynomial spaces, see also \ref{app:polynomials} for more details.
The multiplication of polynomials $A \in \bbC_{\le D}[z]$ and $B \in \bbC_{\le D'}[z]$ is the polynomial $C(z) = A(z)B(z)\in\bbC_{\le(D+ D')}[z]$. Conversely, the polynomial $C(z)$ has a divisor $A(z)$ if it can be represented as $C(z) = A(z)B(z)$ in this sense.
Similarly to the standard case, the greatest common divisor exists and is unique (up to a multiplicative constant) if at least one polynomial is non-zero.
Note that divisibility takes into account the roots at $\infty$ and their multiplicities.
\begin{example}
The polynomial $0 \cdot z^3 + 0 \cdot z^2 + z + 2 = (z-\infty)^2 (z+2)  \in \bbC_{\le 3}[z]$ is a divisor of the polynomial $A(z)$ in from \Cref{ex:simple_poly}, but the polynomial $0 \cdot z^4 + 0 \cdot z^3 + 0 \cdot z^2 + z + 2 = (z-\infty)^3(z+2) \in \bbC_{\le 4}[z]$ is not,
because there are not enough infinite roots in the expansion of $A(z)$.
\end{example}

\subsection{Phase retrieval as a polynomial factorization problem}
\label{sub:PAFdefinitionEquivalence}
In this subsection, we are going to provide a polynomial reformulation of \ref{prob:BPR}.
First, we define the following four polynomials as generating polynomials of the components of the bivariate signal $\bfx[n] = (x_1[n], x_2[n])^\transp \in \bbC^2$, $n=0, 1, \ldots N-1$ and their conjugate reflections, all belonging to $\bbC_{\le N-1}[z]$
\begin{align*}
   & X_1(z) =  \sum_{n=0}^{N-1}x_1[n]z^{n}, \quad \widetilde{X}_2(z)  = \sum_{n=0}^{N-1}\overline{x_1}[N-n+1]z^{n}, \\
   & X_2(z)  =  \sum_{n=0}^{N-1}x_2[n]z^{n}, \quad
  \widetilde{X}_2(z)  = \sum_{n=0}^{N-1}\overline{x_2}[N-n+1]z^{n}.                                                 \\
\end{align*}
Then we define the following matrix polynomial
\begin{equation}\label{eq:Gamma_z}
  \boldsymbol{\Gamma}(z) =
  \begin{bmatrix}
    \Gamma_{11}(z) & \Gamma_{12}(z) \\
    \Gamma_{21}(z) & \Gamma_{22}(z)
  \end{bmatrix} =
  \begin{bmatrix}
    X_1(z) \widetilde{X}_1(z) & X_1(z) \widetilde{X}_2(z) \\
    X_2(z)
    \widetilde{X}_1(z)        &
    X_2(z)
    \widetilde{X}_2(z)
  \end{bmatrix}
  =
  \begin{bmatrix}
    X_1(z) \\
    X_2(z)
  \end{bmatrix}
  \begin{bmatrix}
    \widetilde{X}_1(z) &
    \widetilde{X}_2(z)
  \end{bmatrix},
\end{equation}
where each element is a polynomial $\Gamma_{ij} \in \bbC_{\le 2N-2}[z]$, and we will use the notation $\boldsymbol{\Gamma} \in \bbC_{\le 2N-2}^{2 \times 2}[z]$.

The coefficients of these polynomials are nothing but the covariance functions (auto-correlation and cross-correlation) of the signals $X_1$ and $X_2$; in addition, the spectral matrices $\FourierGamma[m]$ appearing in \ref{prob:BPR} are linked to the evaluations of the polynomial $\boldsymbol{\Gamma}(z)$.

\begin{lemma}\label{lem:Fourier2polynomials}
  The coefficients $\Gamma_{ij}$ of the matrix polynomial can be expressed as
  \[
    \Gamma_{ij}(z)=\sum\limits_{n=0}^{2N-2}\gamma_{ij}[n - N+1]z^n  \mbox{ with } \gamma_{ij}[n] = \sum_{k\in \bbZ} x_i[k + n]\overline{x_j[k]},
  \]
  where $x_i[n] = 0$ for $n < 0$ and $n \ge N$ by convention, and $\gamma_{ij}[n]$ are defined for $n \in \{-N+1, \ldots, N-1\}$.
  Moreover, the spectral matrices $\FourierGamma[m]$ appearing in \ref{prob:BPR} can be expressed as
  \begin{equation}\label{eq:Gamma_evaluations}
    \FourierGamma[m] = e^{\bmj 2\pi \frac{m(N-1)}{M}}\boldsymbol{\Gamma}(e^{-\bmj 2\pi \frac{m}{M}})\:.
  \end{equation}
\end{lemma}

This result is well-known and follows from the spectral (Fourier) representation of the signals, but we give a formal proof in  \ref{app:Fourier2polynomials}, because we consider finite support signals and extended polynomials.

Therefore, we will refer to $\Gamma_{ij}(z)$ as \emph{measurement polynomials}.
Note that the coefficients of the measurement polynomials  $\Gamma_{ij}(z)$ can be uniquely identified from the measurement polynomials from the Fourier measurements $\lbrace\FourierGamma[m]\rbrace_{m=0, \ldots, M-1}$, if the number of measurements $M$  exceeds the degree of these polynomials, \ie $2N-2$, by at least one:
\begin{align}
  M & \geq 2N-1
\end{align}
which is the well-known oversampling condition in standard univariate Fourier phase retrieval, see e.g. \cite{boche_fourier_2017}.

Therefore, the problem \ref{prob:BPR} is equivalent to the following recovery problem, which we refer to as Polynomial Autocorrelation Factorization \eqref{prob:PAF} as to emphasize that we factorize the autocovariance measurements in the polynomial form.
The polynomial reformulation of the problem is very helpful for establishing the uniqueness conditions for \ref{prob:BPR}.
Notably it enables a complete characterization of its uniqueness properties in terms of algebraic properties of complex polynomials.
\begin{theorem}\label{thm:BPR_eq_PAF}
  For $M \ge 2N-1$, \ref{prob:BPR} is equivalent to the following problem
  \begin{equation}
    \begin{split}
      \text{\emph{find  the polynomials }}  X_{1}, X_{2} \in \bbC_{\le N-1}[z] \text{ \emph{given measurement polynomials} } \Gamma_{ij}(z) \text{ defined as \eqref{eq:Gamma_z} }
    \end{split}
    \tag{\sf{PAF}} \label{prob:PAF},
  \end{equation}
  \ie there is a one-to-one correspondence  between the data ($\boldsymbol{\Gamma}(z)$ and $\FourierGamma[m]$) as well as the sets of solutions of the problems (polynomials $X_1,X_2$ and bivariate signal components $x_1[n], x_2[n]$).
\end{theorem}
The proof can be summarized in Figure~\ref{fig:equivalenceBPRandPAF}, and the formal proof can be found in \ref{app:Fourier2polynomials}.

\begin{figure}[hbt!]
  \begin{center}
    \begin{tikzpicture}
      \matrix (m)
      [
        matrix of math nodes,
        row sep    = 5em,
        column sep = 8em
      ]
      {
        \begin{array}{c}
          \text{signals} \\
          \bfx_1,\bfx_2 \in \bbC^{N}
        \end{array} &
        \begin{array}{c}
          \text{Fourier measurements, see \eqref{eq:Gamma_m}} \\
          \{\FourierGamma[m] \}^{M-1}_{m=0} \in (\bbC^{2\times 2})^{M}
        \end{array}   \\
        \begin{array}{c}
          \text{polynomials} \\
          X_1(z),X_2(z) \in \bbC_{\le N-1}[z]
        \end{array} &
        \begin{array}{c}
          \text{measurement polynomials, see \eqref{eq:Gamma_z}} \\
          \boldsymbol{\Gamma} (z) \in \bbC^{2\times 2}_{\le 2N-2}[z]
        \end{array}
        \\
      };
      \path
      (m-1-1) edge [<->] node [left] {\text{one-to-one}} (m-2-1)
      (m-1-2) edge [<->] node [right] {$\begin{array}{c}\text{one-to-one} \\ \text{if } M \ge 2N-1  \end{array}$} (m-2-2)
      (m-1-1.east |- m-1-2)
      edge [|->] node [above] {\ } (m-1-2)
      (m-2-1.east |- m-2-2)
      edge [|->] node [above] {\ } (m-2-2);
    \end{tikzpicture}
  \end{center}
  \caption{Equivalences of data and solutions in problems \ref{prob:BPR} and \ref{prob:PAF}.}
  \label{fig:equivalenceBPRandPAF}
\end{figure}
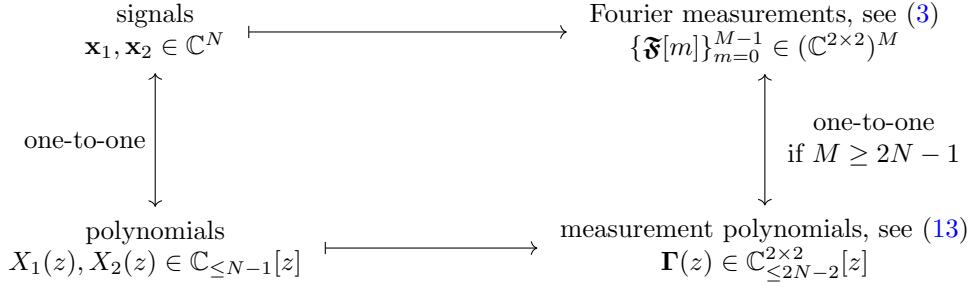

\subsection{General uniqueness result}
\label{sub:uniquenessGeneral}

We now derive a full characterization of the uniqueness properties of the polynomial factorization problem \ref{prob:PAF}.
It is important to keep in mind that by the set of equivalences given in Figure~\ref{fig:equivalenceBPRandPAF}, these results also provide a complete characterization of uniqueness properties of \ref{prob:PPR} and \ref{prob:BPR}.
The following fundamental lemma establishes that the uniqueness properties of the \ref{prob:PAF} problem essentialy boil down to a specific spectral factorization problem.

\begin{lemma}\label{lemma:PAFspectralFactorisation}
 Let $Q(z) \triangleq \gcd(X_1, X_2)$  where $Q \in \bbC_{\le D}[z]$  and $R_1, R_2 \in \bbC_{\le N-D-1}[z]$  be the corresponding quotients, \ie $X_1(z) = Q(z)R_1(z)$ and $X_2(z) = Q(z)R_2(z)$ with $\gcd(R_1,R_2) = 1$.
Then we have that
\begin{enumerate}
\item[1)] the GCD of the autocovariance polynomials $H(z) \triangleq \gcd(\Gamma_{11}, \Gamma_{12}, \Gamma_{21}, \Gamma_{22})$ must have the form
\begin{equation}\label{eq:H_and_Q}
H(z)  = c Q(z)\tilde{Q}(z), \quad c\neq 0;
\end{equation}

\item[2)] given the  quotients $R_{ij}(z)$ (\ie $\Gamma_{ij}(z) = H(z) R_{ij}(z)$, $\gcd(R_{11}, R_{12}, R_{21}, R_{22}) = 1$), the quotients of $X_1$, $X_2$ are determined up to a multiplicative constant as
  \begin{equation}\label{eq:determinationR1R2_H}
    R_1(z)  = \gcd\lbrace  R_{11}(z),  R_{12}(z) \rbrace,\qquad
    R_2(z)  = \gcd\lbrace  R_{21}(z),  R_{22}(z) \rbrace.
  \end{equation}
\end{enumerate}
\end{lemma}
\begin{proof}
Start by 1).
Direct calculations show that, for $i,j = 1, 2$,  $\Gamma_{ij}(z)  = X_i(z)\tilde{X}_j(z)= R_i(z)\tilde{R}_j(z)Q(z)\tilde{Q}(z)$.
Then the GCD of autocovariance polynomials can be explicitly computed as
  \begin{align*}
    H(z) & = \gcd(\Gamma_{11}, \Gamma_{12}, \Gamma_{21}, \Gamma_{22})\\
     & = \gcd\left(\gcd\left(\Gamma_{11}, \Gamma_{12}\right),\gcd\left(\Gamma_{21}, \Gamma_{22}\right)\right) &                                                   \\
                                                             & = \gcd\left(R_1(z)Q(z)\tilde{Q}(z), R_2(z)Q(z)\tilde{Q}(z)\right)                                      & \text{ since } \gcd(\tilde{R_1}, \tilde{R}_2) = 1 \\
                                                             & = cQ(z)\tilde{Q}(z)                                                                                     & \text{ since } \gcd({R_1}, {R}_2) = 1
  \end{align*}
Proof of 2).  From the previous point, we have that $R_{ij}(z) =  c^{-1} R_i(z)\tilde{R}_j(z)$. The determination \eqref{eq:determinationR1R2_H} of $R_1$ and $R_2$ is  then straightforward using that $\gcd({R_1}, {R}_2) =  \gcd(\tilde{R_1}, \tilde{R}_2) = 1$ by assumption.

\end{proof}
\cref{lemma:PAFspectralFactorisation} shows that the study of the uniqueness properties of \ref{prob:PAF} is directly related to uniqueness of the spectral factorization \eqref{eq:H_and_Q}, \ie the recovery of $Q(z)$ given $H(z)$.
 Indeed, if $Q(z)$ can be uniquely recovered from $H(z)$, then the polynomials $X_1(z),X_2(z)$ can be found by mutiplying $Q(z)$ with the quotients $R_1(z)$ and $R_2(z)$ obtained in \eqref{eq:determinationR1R2_H}.
Before giving the sufficient and necessary uniqueness condition, we make a remark about the roots of the product $Q(z)\widetilde{Q}(z)$ which are key to understanding uniqueness.

\begin{remark}\label{rem:roots_H}
Let $Q(z) = \lambda \prod_{i=1}^D (z-\alpha_i)$ (with possibly repeating $\alpha_i$).
Then $H(z) = c Q(z)\tilde{Q}(z)$ has the following factorization
\begin{equation}
H(z) =c\lambda \widetilde{\lambda} \prod_{i=1}^D (z-\alpha_i) (z- \overline{\alpha_i^{-1}}).
\end{equation}
Furthermore, if $|\alpha|=1$, then $\alpha = \overline{\alpha^{-1}}$.
Therefore, a unit-modulus $\alpha$ is a root of $Q(z)$ of multiplicity $\mu$ if and only if it is a root of $H(z)$ of multiplicity $2\mu$.
\end{remark}
From \Cref{rem:roots_H}, we see that the unit-modulus roots of $H(z)$ do not contribute to ambiguity of \ref{prob:PAF}. Indeed,  all unit-modulus roots of $Q(z)$ can be uniquely retrieved from $H(z)$.
This helps us  to establish a necessary and sufficient condition for uniqueness.

\begin{theorem}[Uniqueness of \ref{prob:PAF}]\label{corr:uniquenessPAF}
  The following equivalences are true:
  \begin{align}
    \text{\ref{prob:PAF} admits a unique solution} & \Leftrightarrow X_1(z) \text{ and } X_2(z) \text{ have no common roots outside the unit circle.}                        \\
                                                   & \Leftrightarrow H(z) =  \gcd(\Gamma_{11}, \Gamma_{12}, \Gamma_{21}, \Gamma_{22}) \text{ has no roots outside the unit circle}.
  \end{align}
\end{theorem}
\begin{proof}
The last equivalence being trivial by \Cref{lemma:PAFspectralFactorisation}, we prove only the first equivalence.
\begin{itemize}
\item  $\boxed{\Rightarrow}$ Suppose that the solution of \ref{prob:PAF} is essentially unique, but the polynomial $Q(z)$ has a root $\alpha$ outside the unit circle.
Then easy calculations show that polynomial $Q'(z) = \frac{Q(z)(z-\overline{\alpha^{-1}})}{(z-\alpha)}$ satisfies
\[
Q'(z)\widetilde{Q'}(z) = Q(z) \widetilde{Q}(z).
\]
On the other hand $Q'(z)$ is not proportional to $Q(z)$, and therefore polynomials $X'_1(z) \triangleq Q'(z) R_1(z)$, $X'_2(z) \triangleq Q'(z) R_2(z)$ are not proportional to $X_1(z)$ and $X_2(z)$, but give an alternative pair of polynomials such that $\Gamma'_{ij}(z) = \Gamma_{ij}(z)$ (a contradiction).

\item $\boxed{\Leftarrow}$ Assume that $H(z)$ has only unit-modulus roots. Then there is a unique monic polynomial $Q(z)$ such that $H(z) = c Q(z)\widetilde{Q}(z) = c(Q(z))^2$.
Therefore, by  \cref{lemma:PAFspectralFactorisation}, we can find polynomials $R'_1(z)$ and $R'_{2}(z)$ such that
\[
X_1(z) = c_1 R'_1(z) Q(z), \quad X_2 = c_2 R'_2(z) Q(z), \quad c_1,c_2 \in \mathbb{C}.
\]
The pair $(c_1,c_2)$ can be determined from the relations \eqref{eq:Gamma_z} up to a common unit-modulus factor accounting for the global rotation ambiguity, see Section \ref{sub:trivialAmbiguities}.  \qedhere
\end{itemize}
\end{proof}

\begin{remark}
  Note that the uniqueness condition given in \cref{corr:uniquenessPAF} clarifies previous statements made in the literature \cite{raz_vectorial_2013,jaganathan_reconstruction_2019}.
  In particular, in \cite[Theorem 1]{raz_vectorial_2013} it was claimed that a necessary and sufficient for uniqueness of the solution of \ref{prob:BPR}  is the coprimeness of the polynomials $X_1(z)$ and $X_2(z)$.
Our  \cref{corr:uniquenessPAF} shows that it was just a sufficient condition, because unimodular roots do not affect uniqueness.
This agrees with a similar behavior observed for standard univariate 1D phase retrieval, see \cite{beinert2015ambiguities}.
\end{remark}

\begin{remark}[Almost everywhere uniqueness of \ref{prob:PAF}]
  Whereas the analysis of non-uniqueness properties of \ref{prob:PAF} follows closely that of the standard phase retrieval problem, a distinctive feature of the bivariate/polarimetric phase retrieval problem is that it is almost everywhere unique.
  This can be seen by observing that the set of polynomials $X_1, X_2 \in \bbC_{\le N-1}[z]$ with at least one common root is  an algebraic variety of dimension smaller than $2N-1$; , hence it is of measure zero.
  Put it differently, this shows that \ref{prob:PAF} has the appealing property that almost all polynomials $X_1, X_2  \in \bbC_{\le N-1}[z]$ can be uniquely recovered by the set of measurement polynomials $\Gamma_{ij}(z)$.
\end{remark}

\subsection{Ambiguities and counting the number of solutions}

In this section, we refine \cref{corr:uniquenessPAF} by providing the number of solutions and describing the set of solution of  \ref{prob:PAF}.
Note that, that \cref{lemma:PAFspectralFactorisation} implies that the uniqueness properties of \ref{prob:PAF} resumes in essence to that of a standard univariate phase retrieval problem (taking into account specificites related to the bivariate/polarimetric phase retrieval setting).
Indeed the uniqueness of the univariate phase retrieval is determined by the uniqueness \cite{beinert2015ambiguities,boche_fourier_2017} of the factorization $H(z) = Q(z)\tilde{Q}(z)$.
In principle, we could invoke the existing results from  \cite{beinert2015ambiguities}, however, we prefer to give a complete characterization that relies upon the formalism with $0$ and $\infty$ roots used in this paper.
In particular, this formalism allows us to treat the time shift ambiguities in a unified manner.

\begin{theorem}[Number of solutions of \ref{prob:PAF}]
  \label{theorem:PAFuniqueness}
Let $H(z) = Q(z)\tilde{Q}(z)$ and   $\mu_1, \ldots, \mu_{N_D}$ be the respective multiplicities of the $N_D$ non-unit-modulus roots pairs $(\delta_1, \conj{\delta}_1^{-1}), \ldots, (\delta_{N_D}, \conj{\delta}_{N_D}^{-1})$ of $H(z)$.
  Then the problem \ref{prob:PAF} admits exactly
  \begin{equation}\label{eq:numberSolutionsPAF}
    \prod_{i=1, \vert \delta_i\vert \neq 1}^{N_D}(\mu_i + 1)
  \end{equation}
  different solutions, where only non-unimodular common roots of $X_1$ and $X_2$ contribute to the total number of solutions.
  In particular, when common roots are all simple and outside the unit circle, there is exactly $2^{N_D}$ different solutions.
\end{theorem}
\begin{proof}
  Lemma \ref{lemma:PAFspectralFactorisation} shows that the number of solutions of \ref{prob:PAF} is exactly the number of different (up to multiplication by a scalar) polynomials $Q(z)$ such that $H(z) = Q(z)\tilde{Q}(z)$.
  This \emph{spectral factorization} problem is equivalent to selecting the roots of $Q(z)$ amongst the root pairs $(\delta_i, \conj{\delta}_i^{-1})$ of $H(z)$.
Since  $\mu_1, \ldots, \mu_{N_D}$ are the multiplicities of the root pairs of $H(z)$, then for each root pair $(\delta_i, \conj{\delta}_i^{-1})$ one has to select exactly $\mu_i$ roots among those pairs, leading to a polynomial $Q(z)$ of degree $D = \mu_1 + \ldots+ \mu_{N_D}$.

Consider a non-unit root pair $(\delta_i, \conj{\delta}_i^{-1})$ with multiplicity $\mu_i$; then the number of different combinations of $\mu_i$ roots is that of a random draw of $k = \mu_i$ items with replacement in a set of $n = 2$ elements, \ie $(n+k-1)!/(k!(n-1)!) = \mu_i +1$.
  Repeating the same process for each root pair gives the total number \eqref{eq:numberSolutionsPAF} of solutions to the \ref{prob:PAF} problem.
\end{proof}

\begin{remark}[Counting multiplicities]
The number of solutions \eqref{eq:numberSolutionsPAF} depends in fact on the multiplicities of pairs of roots of $Q(z)$.
This means in particular that if $\delta$ and $\conj{\delta}^{-1}$ are roots of $Q(z)$ with multiplicities $\mu_1$ and $\mu_2$, then the multiplicity of the pair $(\delta, \conj{\delta}^{-1})$ of $H(z)$ is equal to  $\mu_1 + \mu_2$.
  The same  applies  to $0$ and $\infty$ roots.
\end{remark}

\begin{example}
Consider $Q(z) =A(z)$ that is the polynomial from \cref{ex:simple_poly} having double $\infty$ root and simple roots $\{-2,1,0\}$. Then the polynomial $H(z) = A(z) \tilde{A}(z)$, where is
\[
H(z) = Q(z) \tilde{Q}(z) = -\frac{1}{2} (z - \infty)^{3} \left(z+\frac{1}{2}\right)(z+2) (z-1)^2 z^3.
\]
The multiplicity of the root pair $(0,\infty)$ is $3$, while the root pair $(-2,-\frac{1}{2})$ have multiplicity $1$.
This yields a total of $4\cdot 2= 8$ solutions, where the other factorizations are given by permuting $0$ and $\infty$ roots or/and replacing root $-2$ with $-\frac{1}{2}$,
For example, some of possible alternative factorisations are given by $Q(z) = \frac{1}{2} (z+2) (z-1) z^3$ or $Q(z) = (z+\frac{1}{2})(z-1)z^3$.
\end{example}

\begin{remark}[time shifts and common roots at $0$ or $\infty$] One striking benefit of the use of polynomials with roots at the infinity is that it permits to recover the trivial shift ambiguity as part of the set of solutions to \ref{prob:PAF}.
  To illustrate this, let us consider the particular example of a bivariate signal $\lbrace \bfx[n]\rbrace_{n=0, \ldots N-1}$ where $\bfx[0] = \bfx[N-1] = \mathbf{0}$. We further assume that the polynomials associated to the sub-signal $\lbrace \bfx[n]\rbrace_{n=1, \ldots N-2}$ have no roots in common.
  Thus, the polynomials $X_1(z)$ and $X_2(z)$ share only two roots: $0$ and $\infty$, which turn out to be "conjugate reflected" from one another.
  The polynomial $H(z) = Q(z)\tilde{Q}(z)$ has then a double root pair $(0, \infty)$ meaning that by \eqref{eq:numberSolutionsPAF} the \ref{prob:PAF} problem has exactly three different solutions -- which are simply trivially shifted versions from one another.
  Figure \ref{fig:shiftAmbiguity} depicts these three solutions obtained by selecting either the root $0$ or $\infty$ for each single root pair $(0, \infty)$ of $H(z)$.
  \begin{remark}[On conjugate reflection ambiguity]
    Assume that $\bfx[0], \bfx[N-1] \neq \mathbf{0}$, so that time-shifts are not part of the total number of solutions \eqref{eq:numberSolutionsPAF}.
    Then, in the non-unique case, the number of non-trivially different solutions \eqref{eq:numberSolutionsPAF} for \ref{prob:PAF} (and by equivalence, for \ref{prob:PPR} and \ref{prob:BPR}) is twice that of the standard univariate 1D phase retrieval problem \cite{boche_fourier_2017, beinert2015ambiguities}.
    This can be explained by the fact that conjugate reflection is not, in general, a trivial ambiguity for the bivariate case (see Section \ref{sub:trivialAmbiguities}).
    More precisely, this means that unlike the univariate case, exchanging common roots $(\delta_1, \ldots, \delta_{N_D})$ with their conjugated reflected versions $(\overline{\delta_1}^{-1}, \ldots, \overline{\delta_{N_D}}^{-1})$ do not yield to a trivial ambiguity.
  \end{remark}

  \begin{figure}
    \includegraphics[width=\textwidth]{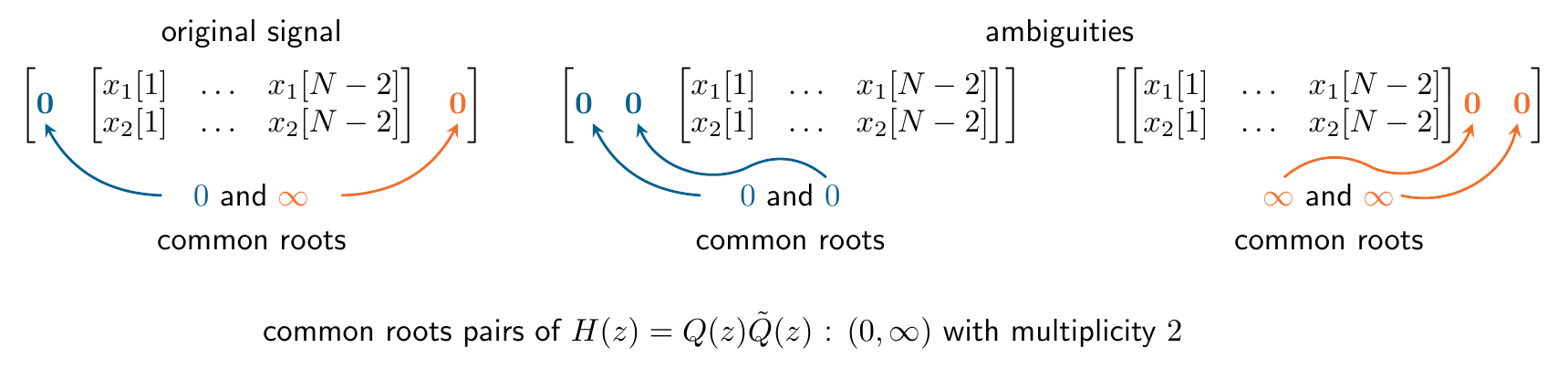}
    \caption{The three solutions of \ref{prob:PAF} obtained by selecting either the root $0$ or $\infty$ for each single root pair $(0, \infty)$ of $H(z)$.
      These three solutions are simple time-shifted versions of one-another, showing that Theorem \ref{theorem:PAFuniqueness} encompasses the time-shift trivial ambiguity as a special case.}\label{fig:shiftAmbiguity}
  \end{figure}
\end{remark}

We conclude the study of uniqueness properties of \ref{prob:PAF} in \cref{theorem:explicitFormPAF} below.
It provides an explicit expression of \ref{prob:PAF} solutions in the simplified case where there are no $0$ or $\infty$ roots in common, meaning that $\bfx[0] \neq \mathbf{0}$ and $\bfx[N-1] \neq \mathbf{0}$.

\begin{theorem}[Expression of \ref{prob:PAF} solutions]\label{theorem:explicitFormPAF}
  Suppose that $H(z) = \gcd(\Gamma_{11}, \Gamma_{12}, \Gamma_{21}, \Gamma_{2}) =Q(z)\tilde{Q}(z)$ with $D$ pair roots $(\delta_i, \conj{\delta}_i^{-1})$.
  Let $X_1(z) = Q(z)R_1(z)$ and $X_2(z) = Q(z)R_2(z)$.
  Denote by $\alpha_{1i}$ and $\alpha_{2i}$ the roots of $R_1(z)$ and $R_2(z)$, respectively,
  Then all solutions $X_1'(z)$ and $X_2'(z)$ to the \ref{prob:PAF} problem can be expressed as
  \begin{align}
    X_1'(z) & = e^{\jmath \theta} \lambda_1 \prod_{i=1}^{D}(z-\beta_{i})\prod_{i=1}^{N-D-1}(z-\alpha_{1i}) \\
    X_2'(z) & = e^{\jmath \theta} \lambda_2\prod_{i=1}^{D}(z-\beta_{i}) \prod_{i=1}^{N-D-1}(z-\alpha_{2i})
  \end{align}
  where each $\beta_i$ is chosen amongst zeros pairs $(\delta_i, \overline{\delta}_i^{-1})$ of $H(z)$; the angle $\theta \in (-\pi, \pi)$ accounts for the global phase trivial ambiguity.
  The constants $\lambda_1, \lambda_2 \in \bbC$ are given by
  \begin{align}
    \lambda_1 & =\sqrt{ \left \vert \gamma_{11}[N-1]\right\vert \prod_{i=1}^{D} | \beta_{i}| ^{-1}\prod_{i=1}^{N-D-1} | \alpha_{1i}|^{-1}}                \\
    \lambda_2 & = e^{\jmath\Delta}\sqrt{ \left\vert \gamma_{22}[N-1]\right\vert \prod_{i=1}^{D} | \beta_{i}| ^{-1}\prod_{i=1}^{N-D-1} | \alpha_{2i}|^{-1}} \:,
  \end{align}
  where $\Delta$ reads
  \begin{equation}
    \Delta =  \pi(N-1)+ \arg  \gamma_{12}[N-1] + \sum_{i=1}^D\arg \beta_i + \sum_{i=1}^{N-D-1}\arg \alpha_{2i}\:.
  \end{equation}
\end{theorem}
\begin{proof}
  See \ref{app:proof:theoremUniqueness}.
\end{proof}

\subsection{Uniqueness in practice}
\label{sub:uniquenessInterpretation}
We provide in this section a small numerical study which illustrates how uniqueness of \ref{prob:PAF} can quickly occur in practice.
Let us consider the case of a bivariate signal with constant polarization state, which is one of the simplest models for bivariate signals.
Formally, such signal $\lbrace \bfx[n]\rbrace_{n=0, 1, \ldots, N-1}$ can be written as
\begin{equation}
  \bfx[n] = \begin{bmatrix}a\\b\end{bmatrix}s[n], \quad \bfs \in \mathbb{C}^N, a, b \in \mathbb{C}\:. \label{eq:constantPolarizedSignal}
\end{equation}
In \eqref{eq:constantPolarizedSignal}, the complex-valued signal $\bfs$ defines the time or 1D-spatial dynamics of the bivariate signal $\bfx[n]$, whereas complex constants $a$ and $b$ define its polarized state.
For instance, when $a \neq 0, b=0$, the signal $\bfx[n]$ is said to be linearly horizontally polarized; similarly for $a = 0, b\neq0$, it is said to be linearly horizontally polarized.
Finally, when e.g. $a = 1$ and $b = \pm \bmj$, \eqref{eq:constantPolarizedSignal} is that of a circularly polarized signal, since the two entries of $\bfx[n]$ are in quadrature with one another.

Regarding the uniqueness of the bivariate signal defined by \eqref{eq:constantPolarizedSignal}, we observe that the polynomials $X_1(z)$ and $X_2(z)$ associated with entries of $\bfx[n]$ share the same $N-1$ roots\footnote{we assume here for simplicity that nor $a=0$ nor $b=0$, so that polynomials are properly defined.}, those of the complex polynomial $S(z) \triangleq \sum_{n=0}^{N-1}s[n]z^n$.
Thus, according to \cref{theorem:PAFuniqueness} there are up to $2^{N-1}$ different solutions for the bivariate phase retrieval problem.
Notably, if $S(z)$ has no roots on the unit circle, and if they are all distinct from one another, then there are exactly $2^{N-1}$ different solutions for the bivariate phase retrieval problem.
We now consider two perturbations models of the constant polarized signal \eqref{eq:constantPolarizedSignal}:
\begin{itemize}
  \item \emph{single entry perturbation}: we consider a perturbated signal $\lbrace\bfx_{n_0, \boldsymbol{\epsilon}}[n]\rbrace_{n=0, 1, \ldots, N-1}$, where $n_0$ is a randomly selected index from the uniform discrete distribution on $[\![0, N-1]\!]$ and
        \begin{equation}
          \bfx_{n_0, \boldsymbol{\epsilon}}[n_0] = \bfx[n] +\boldsymbol{\epsilon}, \boldsymbol{\epsilon} \sim \calN_\bbC(0, \sigma^2\bfI_2),\quad \bfx_{n_0, \boldsymbol{\epsilon}}[n] = \bfx[n]\quad \text{for } n\neq n_0 \label{eq:pertuModel1}
        \end{equation}
        with $\calN_\bbC(0, \sigma^2\bfI_2)$ the complex normal distribution of variance $\sigma^2$.
  \item \emph{full signal perturbation}: this time the perturbation is applied to all entries of $\lbrace \bfx[n]\rbrace_{n=0, 1, \ldots, N-1}$ such that
        \begin{equation}
          \bfx_{\boldsymbol{\epsilon}}[n] = \bfx[n] +\boldsymbol{\epsilon}[n], \qquad \boldsymbol{\epsilon}[n] \sim \calN_\bbC\left(0, \frac{\sigma^2}{N}\bfI_2\right)\label{eq:pertuModel2}
        \end{equation}
        for every $n \in [\![0, N-1]\!]$. Note that we normalized by $N$ the variance of the perturbation to ensure proper comparison between the two models.
\end{itemize}
To assess numerically whether these perturbated signals remains non-uniquely recoverable by \ref{prob:PAF}, we use two different but complementary figures of merits.
The first one is the rank deficiency of the standard Sylvester matrix of polynomials $X_1(z)$ and $X_2(z)$: it is upper bounded by $N / 2$ when $\exists c \in \bbC$ such that $X_2(z) = cX_1(z)$, and lower bounded by 0 when $\gcd \left(X_1(z), X_2(z)\right) = 1$, \ie the Sylvester matrix is full-rank.
See also details in Section \ref{sub:SylvesterGCDmatrices} further below.
The second figure of merit is \emph{root separation}, \ie the minimal Euclidean distance between roots of $X_1(z)$ and $X_2(z)$.

\begin{figure}[t]
  \centering
  \includegraphics[width=\textwidth]{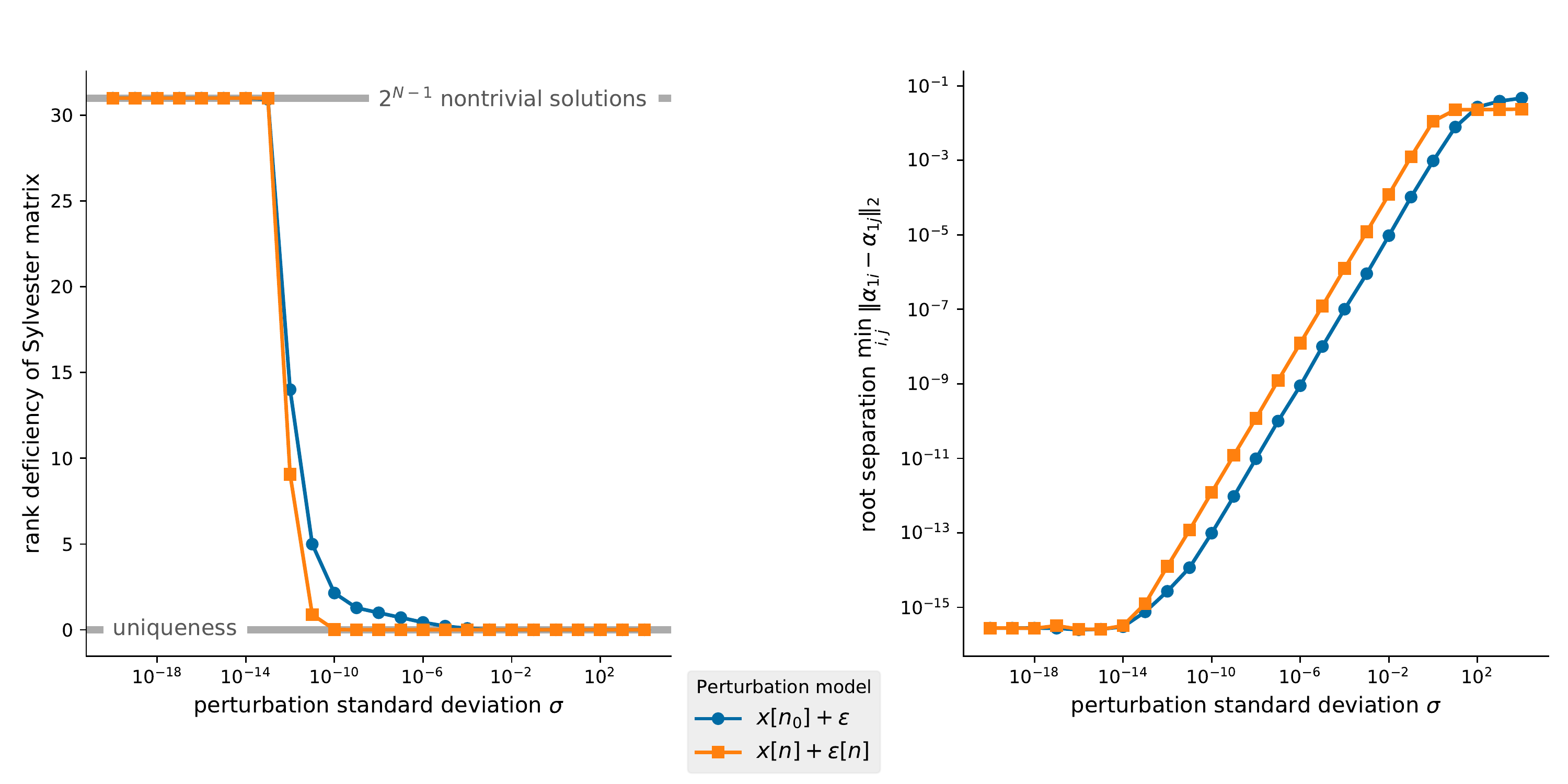}
  \caption{Illustration of the uniqueness properties of \ref{prob:PAF} by studying the evolution of two figures of merit as the standard deviation $\sigma$ of the perturbation increases for models \eqref{eq:pertuModel1} and \eqref{eq:pertuModel2}.}\label{fig:BPRUniqueness}
\end{figure}

Figure \ref{fig:BPRUniqueness} displays the evolution of two non-uniqueness metrics with respect to the perturbation standard deviation $\sigma$ arising in models \eqref{eq:pertuModel1} and \eqref{eq:pertuModel2}.
Results have been obtained by averaging 1000 independent realizations for each value of $\sigma$.
We used an arbitrary randomly generated bivariate signal with constant polarization \eqref{eq:constantPolarizedSignal} for $N = 32$ and unit norm.
Regarding the rank deficiency criteria, both perturbation models lead to maximum rank-deficiency (non-unique case where roots of $X_1(z)$ and $X_2(z)$ all coincide) for very small values of $\sigma < 10^{-13}$.
For $\sigma > 10^{-13}$ we observe a rapid phase transition to zero-rank deficiency (uniqueness) which appears from $\sigma > 10^{-10}$ for the full signal perturbation model \eqref{eq:pertuModel2}.
A similar but slower phenomenon arises for the single-entry perturbation model \eqref{eq:pertuModel1}, where we observe uniqueness for $\sigma \geq 10^{-5}$.
It is interesting to note that the start of the transition coincides with that of root separation criteria, which starts to increase for $\sigma > 10^{-14}$.
Moreover, it appears that an average root separation of the order of $10^{-11}$ is sufficient to ensure uniqueness with model \eqref{eq:pertuModel2}, whereas it should be of the order of $10^{-8}$ with model \eqref{eq:pertuModel1}.

In summary, we see that very small perturbations (say, roughly $\sigma \approx 10^{-6}$) can already transform a non-uniquely recoverable signal from \ref{prob:PAF} into a uniquely recoverable one.
This provides a numerical illustration of the fact that \ref{prob:PAF} is almost everywhere unique, since the set of complex polynomials with at least one common roots is of measure zero.
Compared to standard phase retrieval, where additional measurements are needed to ensure uniqueness (see \cite{boche_fourier_2017} and references therein), \ref{prob:BPR} and its equivalent \ref{prob:PPR} inherently ship with nice uniqueness properties.

\section{Solving PPR: algebraic methods}
\label{sec:solvingPPR_algebraic}
We propose two different but complementary strategies for solving \ref{prob:PPR} in practice.
This first section describes algebraic approaches which exploits the peculiarities of \ref{prob:PPR} through its equivalent \ref{prob:BPR} formulation.
In particular, it leverages the polynomial representation \ref{prob:PAF} which underpins uniqueness results of Section \ref{sec:uniqueness}.
Iterative algorithms for solving \ref{prob:PPR} are described later on in Section \ref{sec:solvingPPR_iterative}.

In the sequel, we assume that measurements are corrupted by additive i.i.d. Gaussian noise such that for $m=0, 1, \ldots M-1$ and $p = 0, 1, \ldots P-1$
\begin{equation}
  y_{m,p} =\vert \bfa_m^\herm \bfX\bfb_p\vert^2+ n_{m,p}, \quad  n_{m, p} \sim \calN(0, \sigma^2)
  \label{eq:noisyMeasuremntsPPR}
\end{equation}
where $\sigma^2$ is the Gaussian noise variance.
The signal-to-noise ratio (SNR) is then defined as
\begin{equation}
  \text{SNR} = \frac{\sum_{m=0}^{M-1}\sum_{p=0}^{P-1}\vert \bfa_m^\herm \bfX\bfb_p\vert^4}{MP\sigma^2}\label{eq:SNRdefAWGNmodel}
\end{equation}

Algebraic approaches exploits two key properties of \ref{prob:PPR} in the noiseless case: \emph{(i)} it is equivalent as \ref{prob:BPR} under the nonrestrictive hypothesis set \eqref{assumption:BPRPPR} and \emph{(ii)} \ref{prob:BPR} itself can be equivalently formulated as a polynomial factorization problem \ref{prob:PAF}, see Section \ref{sub:PAFdefinitionEquivalence} for details.
Notably, a key result from Section \ref{sec:uniqueness} is that polynomials $X_1(z)$ and $X_2(z)$ can be uniquely recovered (up to trivial ambiguities) as greatest common divisors when \ref{prob:BPR} is unique.

In practice, such an approach can be performed in two stages.
One first needs to reconstruct the measurement polynomials $\Gamma_{ij}(z)$, $i=1,2$ given scalar, possibly noisy \ref{prob:PPR} measurements $y_{m,p}$, $m=0,1, \ldots, M-1$, $p=0, 1, \ldots, P-1$.
Second, using techniques from approximate GCD computations \cite{usevich_variable_2017}, we recover $X_1(z)$ and $X_2(z)$ from the measurements polynomials $\Gamma_{ij}(z)$.

\subsection{Reconstruction of measurement polynomials}
Recall that by Lemma \ref{lem:Fourier2polynomials} measurement polynomials $\Gamma_{ij}(z)$ can be readily expressed in  terms of correlation functions $\gamma_{11}[n], \gamma_{12}[m], \gamma_{21}[m], \gamma_{22}[m]$.
Thus, recovery of polynomials $\Gamma_{ij}(z)$'s is identical to the recovery of $\lbrace\gamma_{ij}[n]\rbrace_{n\in \bbZ}$ for $i,j=1,2$.
Equivalently, by discrete Fourier transformation, one must retrieve the spectral matrix $\FourierGamma[m]$ for $m=0, 1, \ldots, M-1$ from \ref{prob:PPR} measurements.

Consider noisy measurements given by \eqref{eq:noisyMeasuremntsPPR}.
Since $\vert \bfa_m^\herm \bfX\bfb_p\vert^2 = \trace \conj{\bfb_p}\bfb_p^\transp \FourierGamma[m]$, an estimate $\hat{\boldsymbol{\Gamma}}[m]$ of $\FourierGamma[m]$ is found for every $m$ by minimizing the following quadratic-loss
\begin{equation}
  \hat{\boldsymbol{\Gamma}}[m] = \argmin_{\substack{\boldsymbol{\tilde{\Gamma}}[m] = \boldsymbol{\tilde{\Gamma}}[m]^\herm \\\rank \boldsymbol{\tilde{\Gamma}}[m]=1}}\sum_{p=0}^{P-1}\left(y_{m,p} - \trace \conj{\bfb_p}\bfb_p^\transp \boldsymbol{\tilde{\Gamma}}[m] \right)^2\:,
  \label{eq:optimizationProblemGamma}
\end{equation}
where the Hermitian and rank-one constraint ensures the estimated spectral matrix $\hat{\boldsymbol{\Gamma}}[m]$ to have the right structure for future polynomial gcd computations.

To solve \eqref{eq:optimizationProblemGamma}, we adopt a heuristic but simple strategy similar to practical polarimetric reconstruction techniques used in optics \cite{,schaefer_measuring_2015,goldstein2017polarized}.
First, we exploit the \emph{Stokes parameters} representation of 2-by-2 Hermitian matrices, which read for an arbitrary Hermitian matrix $\bfM\in \bbC^{2\times 2}$
\begin{equation}
  \bfM = \frac{1}{2}\begin{bmatrix}
    S_{0} +S_{1}     & S_{2} + \bmj S_{3} \\
    S_{2}-\bmj S_{3} & S_{0} - S_{1}
  \end{bmatrix}\qquad  S_0, S_1, S_2, S_3 \in \bbR\:.\label{eq:Stokes2matrix}
\end{equation}
This set of four real-valued parameters are widely used in optics to describe the different polarization states of light.
Formally, Stokes parameters define a bijective map $\calS: \lbrace \bfM \in \bbC^{2\times 2}\vert \bfM = \bfM^\herm\rbrace \rightarrow \bbR^4$ such that $\calS(\bfM) = (S_0,S_1, S_2, S_3)^\transp$.
This allows to express the noiseless measurements as a simple scalar product between Stokes vectors, \ie
\begin{equation}
  \trace \conj{\bfb_p}\bfb_p^\transp \boldsymbol{\tilde{\Gamma}}[m] = \left[\calS\left(\conj{\bfb_p}\bfb_p^\transp\right)\right]^\transp \calS\left(\boldsymbol{\tilde{\Gamma}}[m]\right)\:.
\end{equation}
Thus, for $m$ fixed, we can set $\bfy_{m, :} \triangleq (y_{m, 0}, y_{m, 1}, \ldots, y_{m, P-1})^\transp \in \bbR_+^P$ and the polarization measurement matrix $\bfD \in \bbR^{P\times 4}$ such that its $p$-th row reads $\bfD_p = \left[\calS\left(\conj{\bfb_p}\bfb_p^\transp\right)\right]^\transp$, leading us to rewrite problem \eqref{eq:optimizationProblemGamma} as
\begin{equation}
  \hat{\boldsymbol{\Gamma}}[m]= \argmin_{\substack{\boldsymbol{\tilde{\Gamma}}[m] = \boldsymbol{\tilde{\Gamma}}[m]^\herm \\\rank \boldsymbol{\tilde{\Gamma}}[m]=1}}\left \Vert \bfy_{m,:} -\bfD\calS\left(\boldsymbol{\tilde{\Gamma}}[m]\right)\right\Vert^2_2
  \label{eq:optimizationProblemGammaStokes}\:.
\end{equation}
A possibly sub-optimal yet very simple solution to \eqref{eq:optimizationProblemGammaStokes} consists in finding the best rank-one approximation of the classical least square estimator of Stokes parameters, \ie
\begin{equation}
  \hat{\boldsymbol{\Gamma}}[m]= \mathrm{rank1}\left\lbrace \calS^{-1}\left(\bfD^\dagger \bfy_{m, :}\right)\right\rbrace
\end{equation}
where $\bfD^\dagger$ denotes the Moore-Penrose pseudo-inverse of $\bfD$, $\calS^{-1}$ is the inverse Stokes mapping defined by \eqref{eq:Stokes2matrix}.
The operator $\mathrm{rank1}$ finds the best rank-one approximation of a given matrix with respect to the Frobenius norm.
More precisely, consider the SVD of a 2-by-2 Hermitian matrix $\bfM = \bfU\diag(\sigma_0, \sigma_1)\bfU^\herm$, one has $\mathrm{rank1}(\bfM) = \sigma_0\bfu_0\bfu_0^\herm$, where $\sigma_0$ and $\bfu_0$ are respectively the largest singular value and its corresponding singular vector.

\subsection{Sylvester matrices and GCD}
\label{sub:SylvesterGCDmatrices}
After estimating the matrices $\hat{\boldsymbol{\Gamma}}[m]$, we will build the (estimated) matrix polynomial $\hat{\boldsymbol{\Gamma}}(z)$, and then our goal  is to solve (approximately) the problem \ref{prob:PAF}.
Thanks to Lemma~\ref{lemma:PAFspectralFactorisation} and Theorem~\ref{theorem:explicitFormPAF}, the solution can be found through the GCD (or an approximate common divisors) of the polynomials.
The GCD (or approximate GCD) can be found thanks to the correspondence between (greatest) common divisors and low-rank Sylvester matrices, reviewed below
(see also \cite{usevich_variable_2017} for more details).

For simplicity, we assume that the polynomials have the same (extended) degree, \ie $A,B \in \bbC_{\le L}[z]$.
Then we define the Sylvester-like matrices, parameterized by an integer $D \le L$ (possibly negative)
\begin{equation}\label{eq:sylvesterMatrix}
  \sylv{D}(A,B) \triangleq
  \left[
    \begin{array}{ccc|ccc}
      \pco{a}{0} &        &            & \pco{b}{0} &        &            \\
      \vdots     & \ddots &            & \vdots     & \ddots &            \\
      \pco{a}{L} &        & \pco{a}{0} & \pco{b}{L} &        & \pco{b}{0} \\
                 & \ddots & \vdots     &            & \ddots & \vdots     \\
                 &        & \pco{a}{L} &            &        & \pco{b}{L}
    \end{array}
    \right] \in \bbC^{(2L-D+1) \times 2(L-D+1)}.
\end{equation}
When $D=1$ (\ie the matrix is square $2L \times 2L$), the matrix is the well-known Sylvester matrix.
There are, however, two important extensions of the classic case:
\begin{itemize}
  \item When $1 \le D \le L$, the matrix is tall (the number of columns does not exceed the number of rows), and it is called the \emph{Sylvester subresultant} matrix.
  \item If $D \le 1$ (in general, chosen to be negative), the matrix is fat (the number of rows does not exceed the  number of columns), and such a matrix is called \emph{extended Sylvester} matrix.
\end{itemize}
For an overview of such matrices and the corresponding literature, we refer to \cite{usevich_variable_2017} (note that unlike \cite{usevich_variable_2017} we use the same notation for subresultant and extended Sylvester matrices).
The following theorem is classic.
\begin{theorem}[Sylvester]\label{thm:sylvester}
  Two polynomials  $A, B \in \bbC_{\le L}[z]$ have a non-trivial common divisor if and only if $\sylv{1}(A,B)$ is rank deficient.
  Moreover the (extended) degree $K$ of $\gcd(A(z),B(z))$  is equal to the rank defect of  $\sylv{1}(A,B)$, \ie
  \[
    K = 2L - \rank{\sylv{1}(A,B)}
  \]
  and $\gcd(A(z),B(z)) \in \bbC_{\le K}[z]$.
\end{theorem}
\begin{remark}
  Note that we use the term ``extended degree'' of $\gcd(A(z),B(z))$ to highlight the fact that the polynomials may have common  $\infty$ roots. (And therefore the degree in the usual sense may be lower, see also remarks at the end of \ref{app:polynomials}.)
\end{remark}
The GCD itself can be retrieved from the left or right kernel of the Sylvester matrix $\sylv{D}(A,B)$, as summarized in the following propositions (which can be viewed as extensions of Theorem~\ref{thm:sylvester}).
In what follows, we assume that the GCD has (extended) degree $K$ and note $H(z) = \gcd(A(z), B(z)) \in \bbC_{\le K}[z]$.
Moreover, we define
\[
  F(z) = \frac{A(z)}{H(z)}, \quad G(z) = \frac{B(z)}{H(z)}
\]
the corresponding quotient polynomials.
We begin with the result on the right kernel of Sylverster subresultant matrices.

\begin{proposition}[Right kernel, see e.g. {\cite[Lemma 4.6]{usevich_variable_2017}}]\label{prop:sylvesterRightKernel}
  The rank of the Sylvester subresultant matrix $\sylv{K}(A,B)$ is equal to $2(L-K+1)-1$ (\ie it has rank defect equal to $1$). Moreover, for the (unique up to scalar factor) nonzero vector in the right kernel
  \[
    \sylv{K}  (A,B)\begin{bmatrix}\bfu \\ \bfv\end{bmatrix} = 0;
  \]
  with $\bfu,\bfv \in \bbC^{L-K+1}$, the corresponding polynomials are multiples of the quotient polynomials:
  \[
    U(z) = -cG(z), \quad V(z) = c F(z),
  \]
  where $c \in \bbC$ is some constant.
\end{proposition}

For the case of extended Sylvester matrices ($D \le 1$), the result on the left kernel matrices is less known in the form that we are using here.
This is the reason why we also provide a short proof in \ref{app:sylvesterAGCD}.
\begin{proposition}[Left kernel]\label{prop:sylvesterLeftKernel}
  Let $D \le  1$ (\ie $\sylv{D}  (A,B)$ is fat with $2L-D+1$ rows). Then the rank  of $\sylv{D}  (A,B)$ is equal to
  \[
    \rank(\sylv{D}  (A,B))  = 2L-D+1-K;
  \]
  therefore the dimension of the left kernel (\ie the rank defect) is equal to $K$ (the extended degree of the GCD).
  Moreover, a vector $\bfu \in \bbC^{2L-D+1}$ is in the left kernel ($\bfu^{\transp} \sylv{D} (A,B) = \mathbf{0}$) if and only if
  the vector of coefficients $\bfh \in \bbC^{K+1}$ of the GCD satisfies
  \[
    \bfh^{\transp} \underbrace{\begin{bmatrix}
        u[0]   & u[1]   & \cdots & u[2L-D-K]   \\
        u[1]   & u[2]   & \cdots & u[2L-D-K+1] \\
        \vdots & \vdots &        & \vdots      \\
        u[K]   & u[K+1] & \cdots & u[2L -D]    \\
      \end{bmatrix}}_{\hankel{K+1} (\bfu)} = 0,
  \]
  \ie  $\bfh$ is in the (left) kernel of the Hankel matrix built from $\bfu$.
\end{proposition}

The next section exploits these properties of the kernel of Sylvester matrices to obtain algebraic reconstruction techniques for the \ref{prob:PAF} problem.

\subsection{Algebraic algorithms for \ref{prob:PAF}}

In this section, we propose two algorithms for solving \ref{prob:PAF}
using the results of Section \ref{sub:SylvesterGCDmatrices}.
The intuition behind the algorithms is that generic polynomials (informally speaking, with probability $1$ if drawn from an absolutely continuous probability distribution) $X_1(z)$ and $X_2(z)$ are coprime, and therefore
$\gcd (\Gamma_{11}(z),\Gamma_{12}(z)) =X_1(z)$ and $\gcd (\Gamma_{21}(z),\Gamma_{22}(z)) =X_2(z)$.

As a consequence, we now assume without loss of generality that the problem \ref{prob:PAF} (or equivalently, \ref{prob:PPR} or \ref{prob:BPR}) admits a unique solution up to trivial ambiguities.
In all the algorithms, we use the singular value decomposition (SVD) as to find the approximate kernels of the matrix.
Thus the proposed reconstuction methods may appear as suboptimal since the Sylvester structure is not taken account when computing the (low-rank) kernels.
This limitation could be overcome with structured low-rank approximations \cite{markovsky2008structured}, to be specifically tailored for the \ref{prob:PAF} problem. Such a study would fall outside the scope of the present work.
Still, as demonstrated by the numerical experiments presented in Section \ref{sec:numericalExperiments}, the SVD already provides excellent reconstruction performance in many scenarios, while maintaining a reasonable computational burden.

\subsubsection{Right kernel Sylvester}
\begin{algorithm}[t]
  \SetAlgoLined
  \KwIn{Matrix polynomial $\boldsymbol{\Gamma}(z)  \in \bbC^{2\times 2}_{\le 2(N-1)}$.}
  Build the matrix $\bfS = \sylv{N-1}(\Gamma_{11}, \Gamma_{21}) \in \bbC^{(3N-2) \times 2N}$\;
  Take $\bfv =  \bfv_{2N} \in \bbC^{2N}$ to be the $2N$-th right singular vector of $\bfS$ (corresponding to the last nontrivial singular value)\;
  Partition $\bfv$ as $\bfv = (-\bfv_2, \bfv_1)$, where $\bfv_1 = c \widehat{\bfx}_1$ and $\bfv_2 = c \widehat{\bfx}_2$ with $c \in \bbC$\;
  Determine $\vert c \vert$ by proper norm scaling as
  $$\vert c \vert = \left(\frac{\Vert \bfv_1 \Vert_2^2 + \Vert \bfv_2 \Vert_2^2}{\gamma_{11}[0] + \gamma_{22}[0]}\right)^{\frac{1}{2}}$$
  Set $\widehat{\bfx}_1 = \bfv_1 / \vert c \vert$ and $\widehat{\bfx}_2 = \bfv_2 / \vert c \vert$\;
  \KwResult{Estimates $\widehat{\bfx}_1$ and $\widehat{\bfx}_2$}
  \caption{Sylvester right kernel for \protect\ref{prob:PAF}}
  \label{algo:rightSylvesterPAF}
\end{algorithm}

The first algorithm is based on the properties of the right kernel of Sylvester matrices described in \cref{prop:sylvesterRightKernel}.
It uses the fact that $X_1(z)$ and $X_2(z)$ are (without noise) quotient polynomials of
\[
  \Gamma_{11}(z) = X_1(z) \widetilde{X}_1(z) \text{ and } \Gamma_{21}(z) = X_2(z) \widetilde{X}_1(z).
\]
Note that $X_1(z)$ and $X_2(z)$ are also quotient polynomials of $\Gamma_{12}(z) = X_1(z) \widetilde{X}_2(z) \text{ and } \Gamma_{22}(z) = X_1(z) \widetilde{X}_2(z)$, which adds some freedom in the choice of measurement polynomials.
For the sake of simplicity, we will work with estimated polynomials $\Gamma_{11}(z)$ and $\Gamma_{21}(z)$ in the following.

The complete right kernel Sylvester approach for solving \ref{prob:PAF} is summarized in \cref{algo:rightSylvesterPAF}.
It estimates the (one-dimensional) right kernel by computing the last nontrivial singular value of the Sylvester matrix $\sylv{N-1}(\Gamma_{11}, \Gamma_{21})$.
According to Proposition \ref{prop:sylvesterRightKernel}, this directly gives the vectors of coefficients $\bfx_1, \bfx_2$ of the polynomials $X_1(z)$ and $X_2(z)$ (or simply, the columns of the matrix $\bfX$ in \ref{prob:PPR} and \ref{prob:BPR}) up to one complex multiplicative constant.
This constant is then computed (up to one unit-modulus factor due to the trivial rotation ambiguity) by scaling the 2-norm of $\bfx_1$ and $\bfx_2$ thanks to the value at the origin of estimated autocovariance functions $\gamma_{11}[0] $ and $\gamma_{22}[0]$.

One of the key advantages of this algorithm lies in its simplicity.
Indeed, it only requires a single SVD of a $(3N-2)\times 2N$  matrix and thus has computational complexity $\calO(N^3)$ (with a naive implementation by computing the full SVD).
Potentially, iterative algorithms for truncated SVD can bring this complexity down to $\calO(N^2)$, or even lower (by using techniques based of the FFT).
For the sake of simplicity, we only considered the naive SVD implementation in our experiments.

\subsubsection{Left kernel Sylvester}

The second algorithm exploits the properties of the left kernel of extended (fat) Sylvester matrices (\ie $\sylv{D}$ for $D \leq 1$) detailed in \cref{prop:sylvesterLeftKernel}.
For simplicity and to reduce the size of the involved matrices we set $D = 1$ in what follows.
Nonetheless, the proposed approach can be adapted to any value of $D \leq 1$ if needed.

Algorithm~\ref{algo:leftSylvesterPAF} summarizes the complete procedure.
Compared to the right kernel Sylvester approach, the polynomials $X_1(z)$ and $X_2(z)$ are obtained by two separate GCD computations, that is
$$X_1(z) = c_1 \gcd (\Gamma_{11}, \Gamma_{12}) \text{ and } X_2(z) = c_2 \gcd (\Gamma_{21}, \Gamma_{22}) $$
where $c_1, c_2$ are complex constants.
The computation of each GCD requires three steps: a first SVD to determine the $N-1$ last left singular vectors of $\sylv{1}$; the construction of a fat, horizontally stacked Hankel matrix $\bfH$ with $N$ rows from these $N-1$ singular vectors; a second SVD to obtain the $N$ coefficients of the GCD as the last left singular vector of $\bfH$.
We refer the reader to \cite{usevich_variable_2017} for further details on this procedure.
Once GCDs have been obtained, determination of constants $c_1$ and $c_2$  (up to a common global phase factor) is carried out by properly scaling the norms of estimated coefficients $\bfx_1$ and $\bfx_2$ (using $\gamma_{11}[0]$ and $\gamma_{22}[0]$) and adjusting the phase factor $\arg c_1 \conj{c_2}$ thanks to the value at origin of the estimated cross-covariance function $\gamma_{12}[0]$.

Importantly, the complexity of the left kernel Sylvester method described in Algorithm~\ref{algo:leftSylvesterPAF} is higher for two main reasons.
First, as explained above, it requires the computations of 2 SVD for each of the two GCD determinations.
Moreover, while the first SVD has a cost of $\calO(N^3)$, the second SVD is performed on a large fat stacked Hankel matrix $\bfH$, with complexity $\calO(N^4)$ for a naive implementation.
This can potentially reduced by computing the SVD of  $\bfH\bfH^{\herm}$ instead, eventhough we do not consider such refinement here.

\begin{algorithm}[t]
  \SetAlgoLined
  \KwIn{Matrix polynomial $\boldsymbol{\Gamma}(z)  \in \bbC^{2\times 2}_{\le 2(N-1)}$}
  \For{j = 1,2}{
    Build the matrix $\bfS = \sylv{1}(\Gamma_{j1}, \Gamma_{j2}) \in \bbC^{(4N-4) \times (4N-4)}$\;
    Take the last $N-1$ left singular vectors of $\bfS$, \ie
    \[
      \bfu_{3N-2}, \ldots, \bfu_{4N-4}.
    \]
    Stack the  Hankel matrices with $N$ rows in the following matrix
    \[
      \bfH = \begin{bmatrix}\bfH_{N} (\bfu_{3N-2}) & \cdots & \bfH_{N} (\bfu_{4N-4})\end{bmatrix} \in \bbC^{N \times (N-1)(3N-3)}
    \]
    Retrieve $\bfw_j = c_j \widehat{\bfx}_j$, $c_j \in \bbC$ as the last left singular vector of  $\bfH$.
  }
  Determine constants $c_1, c_2$ as
  $$ c_1 = \frac{\Vert \bfw_1 \Vert_2}{\sqrt{\gamma_{11}[0]}} \text{ and }  c_2 = \frac{\Vert \bfw_2 \Vert_2}{\sqrt{\gamma_{22}[0]}}\exp\left[\jmath(\arg \gamma_{12}[0] - \arg \bfw_2^\herm \bfw_1)\right] $$
  Set $\widehat{\bfx}_1 = \bfw_1 / c_1 $ and $\widehat{\bfx}_2 = \bfw_2 / c_2$\;
  \KwResult{Estimates $\widehat{\bfx}_1$ and $\widehat{\bfx}_2$}
  \caption{Sylvester left kernel for \protect\ref{prob:PAF}}
  \label{algo:leftSylvesterPAF}
\end{algorithm}

\section{Solving PPR: iterative algorithms}
\label{sec:solvingPPR_iterative}
We now address the design of iterative algorithms to solve the noisy \ref{prob:PPR} problem.
Section \ref{sub:SDPrelaxation} and Section \ref{sub:Wirtingerflow} exploit the \ref{prob:PPR-1D} representation of the original problem to provide a semidefinite programming (SDP) relaxation and Wirtinger flow algorithm, respectively.

\subsection{SDP relaxation}
\label{sub:SDPrelaxation}
Semidefinite programming (SDP) approaches for phase retrieval have been increasingly popular for over a decade \cite{candes_phase_2013,candes_phase_2011}.
In the classical 1D phase retrieval case, SDP approaches exploit that eventhough measurements are quadratic in the unknown signal $\bfx \in \bbC^N$, they are linear in the rank-one matrix $\bfx\bfx^\herm$.
For \ref{prob:PPR}, the 1D equivalent representation \ref{prob:PPR-1D} enables to formulate a SDP relaxation of the original problem, by observing that
\begin{equation}
  \vert \bfc_{m, p}^\herm \vectx\vert^2 = \trace \bfc_{m, p}\bfc_{m,p}^\herm \vectx\vectx^\herm \triangleq \trace \bfC_{m,p}\vectX
\end{equation}
\ie, noiseless measurements can be rewritten as a linear functions of the lifted positive semidefinite rank-one matrix $\vectX \triangleq \vectx\vectx^\herm \in \bbC^{2N\times 2N}$.
Following the classical PhaseLift methodology \cite{candes_phase_2011,candes_phase_2013}, the original nonconvex \ref{prob:PPR} problem can be relaxed into a SDP convex program as
\begin{equation}
  \begin{split}
    \text{minimize}\quad & \frac{1}{2}\sum_{m=0}^{M-1}\sum_{p=0}^{P-1}\left(y_{m,p} - \trace \bfC_{m,p}\vectX\right)^2 + \lambda\Vert \vectX\Vert_\star\\
    \text{subject to}\quad & \vectX \succeq  0\label{eq:SDPRelaxationPPRNoisy}
  \end{split}
\end{equation}
where $\lambda \geq 0$ is an hyperparameter that allows to control the trade-off between the likelihood of observations and the nuclear norm regularization $\Vert \cdot \Vert_\star$.
Note that since $\vectX$ is constrained to be positive semidefinite, the nuclear norm regularization is equivalent to the trace-norm regularization used in \cite{candes_phase_2013} since $\Vert \vectX \Vert_\star = \trace \vectX$ in this case.

The SDP program \eqref{eq:SDPRelaxationPPRNoisy} takes a standard form: thus it can be solved in many ways, including interior point methods \cite{vandenberghe1996semidefinite}, first-order methods \cite{monteiro2003first} or using disciplined convex programming solvers such as \texttt{CVXPY}\footnote{\url{https://www.cvxpy.org/}}.
For completeness, we provide below an explicit algorithm to solve \eqref{eq:SDPRelaxationPPRNoisy} using a proximal gradient approach \cite[Chapter 10]{beck2017first}.
It closely follows the approach described in \cite{candes_phase_2013,goldstein2014field}.

The objective function in \eqref{eq:SDPRelaxationPPRNoisy} can be rewritten as the sum $f(\vectX) + g(\vectX)$, where
\begin{equation}
  f(\vectX) = \frac{1}{2}\sum_{m=0}^{M-1}\sum_{p=0}^{P-1}\left(y_{m,p} - \trace \bfC_{m,p}\vectX\right)^2, \quad g(\vectX) = \lambda\Vert \vectX\Vert_\star + \iota_{\succeq 0}(\vectX)\label{eq:proximalGradientSplitting}
\end{equation}
where $\iota_{\succeq 0}(\cdot)$ denote the indicator function on the positive semidefinite cone.
This ensures the formal equivalence between \eqref{eq:SDPRelaxationPPRNoisy} and the unconstrained minimization problem
\begin{equation}
  \min_{\vectX \in \bbC^{2N\times 2N}} f(\vectX) + g(\vectX)\:.\label{eq:proximalGradientSplitting2}
\end{equation}
The convex optimization problem \eqref{eq:proximalGradientSplitting2} can be efficiently solved by proximal gradient methods, which take advantage of the splitting between $f$ and $g$ of the objective function.
More precisely, we use the fast proximal gradient method which consist, at iteration $k$:
\begin{align}
  \vectX^{(k+1)} & = \prox_{t_k g}\left(\vectZ^{(k)} - t_k \nabla f(\vectZ^{(k)})\right)    \label{eq:ProxGradientStepFISTA} \\
  \eta_{k+1}     & = \frac{1+\sqrt{1+4\eta_k^2}}{2}                                                                          \\
  \vectZ^{(k+1)} & =  \vectX^{(k+1)}+ \left(\frac{\eta_k-1}{\eta_{k+1}}\right)\left( \vectX^{(k+1)} -  \vectX^{(k)}\right)
\end{align}
where $t_k$ is a step-size which is chosen such that the proximal gradient step \eqref{eq:ProxGradientStepFISTA} obey some sufficient decrease condition; see e.g. \cite[p. 271]{beck2017first} for details.
Our choice for the function $g$ in \eqref{eq:proximalGradientSplitting2} enables a simple expression for the associated proximal operator (see \cite{goldstein2014field}):
\begin{equation}
  \begin{split}
    \prox_{\tau g}(\bfX) &\triangleq \min_{\bfZ \succeq 0} \tau \lambda \Vert \bfZ\Vert_\star + \Vert \bfZ- \bfX\Vert_2^2\\
    &= \bfU\mathrm{\textbf{shrink}}(\Sigma, \tau \lambda)\bfU^\herm
    \label{eq:proximalOperatorSDP}
  \end{split}
\end{equation}
where in the last equation, $\bfU \Sigma \bfU^\herm$ is the eigenvalue decomposition of $\bfX$ and the shrink operator is defined entry wise by $\mathrm{shrink}(\sigma_i, \tau\lambda) = \mathrm{sign}(\sigma_i)\max\lbrace \vert \sigma_i\vert-\tau \lambda, 0\rbrace$.

\paragraph{Choice of regularization parameter $\lambda$} In this work, we fix the value of the regularization parameter to $\lambda = 1/\mathrm{SNR}$: we found empirically that this choice provides good results in most scenarios, as it provides a reasonable tradeoff between likelihood of observations and the nuclear norm regularization in the objective function of \eqref{eq:SDPRelaxationPPRNoisy}.

\paragraph{Convergence}  Obviously, as our algorithm is a convex SDP program, the precision towards the optimal cost value can become arbitrarily good as one increases the number of iterations.
In practice, one needs to stop the algorithm when a prescribed tolerance $\varepsilon$ is reached. To this aim we implemented stopping criteria that carefully monitor a normalized residual, see \cite{goldstein2014field} for details.
Moreover, it may happen that the estimated lifted matrix $\hat{\boldsymbol{\Xi}}$ generated by the sequence of $\vectX^{(k)}$ is not rank one: in this case, one first computes the rank-one approximation of $\hat{\boldsymbol{\Xi}}$ (e.g. using SVD) to obtain the estimated signal $\hat{\boldsymbol{\xi}}$.

\paragraph{Complexity}
The computational cost of the proposed algorithm concentrates on the proximal gradient step \eqref{eq:ProxGradientStepFISTA}, where the evaluation of the proximal operator and the computation $\nabla f$ share the computational burden.
More precisely, the eigenvalue decomposition of a $2N \times 2N$ matrix together with the shrink operator leads to $\calO(N^3)$ calculations.
The computation of the gradient leads to $MP$ trace evaluations of order $\calO(N^2)$ flops, meaning that the number of flops per iteration is of order $\calO(MPN^2 + N^3)$.
\bigskip

The full procedure is summarized in Algorithm \ref{algo:PolarPhaseLift}.

\begin{algorithm}[t]
  \SetAlgoLined
  \KwIn{measurements $\bfy \in \bbR^{MP}$, lifted measurement matrices $\bfC_{m,p}\in \bbC^{2N\times 2N}$, regularization parameter $\lambda \geq 0$.}
  set arbitrary $\vectX^{(0)}$\;
  $\vectZ^{(0)} \leftarrow \vectX^{(0)}$\;
  $k \leftarrow 0$\;
  \While{$\text{stopping criterion is not satisfied}$}{
    $\vectX^{(k+1)}  = \prox_{t_k g}\left(\vectZ^{(k)} - t_k \nabla f(\vectZ^{(k)})\right)$ where the proximal operator is given by \protect\eqref{eq:proximalOperatorSDP}\;
    $\eta_{k+1}         = \frac{1+\sqrt{1+4\eta_k^2}}{2} $\;
    $\vectZ^{(k+1)} =  \vectX^{(k+1)}+ \left(\frac{\eta_k-1}{\eta_{k+1}}\right)\left( \vectX^{(k+1)} -  \vectX^{(k)}\right)$\;
    $k\leftarrow  k+1$\;
  }
  $\hat{\boldsymbol{\xi}} \leftarrow \mathrm{rank1}\left(\vectX^{(k)}\right)$\;
  \KwResult{estimate $\hat{\boldsymbol{\xi}} $ }
  \caption{SDP relaxation for \protect\ref{prob:PPR}}
  \label{algo:PolarPhaseLift}
\end{algorithm}

\subsection{Wirtinger flow for PPR}
\label{sub:Wirtingerflow}

Exploiting further the 1D equivalent representation \ref{prob:PPR-1D} of the \ref{prob:PPR} problem, another approach consists in minimizing directly the following nonconvex quadratic objective
\begin{equation}
  \min_{\vectx \in \bbC^{2N}} F(\vectx) \triangleq \frac{1}{2}\Vert \bfy - \vert \bfC \vectx\vert^2 \Vert_2^2\label{eq:WFnonconvexProblem}
\end{equation}
where $\bfy \in \bbR^{MP}$ gather \ref{prob:PPR} measurements and where the rows of $\bfC \in \bbC^{MP\times 2N}$ are given by $\bfc_{m, p}^\herm$, see Section \ref{sub:1DequivalentModelphysicsModel}.
Provided that one can find a initial point $\vectx^{(0)}$ close enough from the global minimizer of \eqref{eq:WFnonconvexProblem}, a simple strategy based on gradient descent can be used to solve \ref{prob:PPR}.
However, such an approach requires special care since the optimization variable $\vectx$ is complex-valued.
In fact, the objective function in \eqref{eq:WFnonconvexProblem} is real-valued, and thus it is not differentiable with respect to complex analysis.
Instead, one needs to resort to the so-called $\bbC\bbR$ or \emph{Wirtinger}-calculus \cite{kreutz-delgado_complex_2009} to provide a meaningful extension of gradient-descent-type algorithms to the complex case.
This is precisely the approach proposed in \cite{candes_phase_2015Wirt} to solve standard phase retrieval, where the complex gradient descent is called \emph{Wirtinger flow}.

Leveraging the original Wirtinger flow approach, we propose below a complex-gradient descent algorithm which solves the nonconvex problem \eqref{eq:WFnonconvexProblem}.
Compared to the original paper \cite{candes_phase_2015Wirt}, we incorporate optimal step size selection \cite{jiang_wirtinger_2016} together with a proposed acceleration scheme \cite{xu_accelerated_2018}.
We further propose an efficient strategy for initialization based on the algebraic methods for \ref{prob:PPR} described in Section \ref{sec:solvingPPR_algebraic}.
The superiority of these initializations over standard ones (e.g. spectral initialization as in \cite{candes_phase_2015Wirt}) will be demonstrated in Section \ref{sub:initWFexp}.

The proposed PPR-WF algorithm is as follows. Starting from two initial points $\vectx^{(0)}$, $\vectx^{(1)}$, the $k$-th iteration reads
\begin{align}
  \beta_k        & = \frac{k+1}{k+3}                                                                 \\
  \vectz^{(k)}   & = \vectx^{(k)} + \beta_k \left(\vectx^{(k)} - \vectx^{(k-1)}\right)               \\
  \vectx^{(k+1)} & =  \vectz^{(k)} - \mu_k\nabla F\left(\vectz^{(k)}\right)\label{eq:WFGradientStep}
\end{align}
where $\beta_k$ is a sequence of accelerated parameters and $\mu_k$ is a carefully chosen stepsize, see further below.
Compared to the standard WF algorithm, PPR-WF takes advantage of the acceleration procedure first proposed in \cite{xu_accelerated_2018} in the context of ptychographic phase retrieval (but using a magnitude loss function instead of a square magnitude loss function as used here).
Note that the complex gradient of $F$ can be computed explicitly as
\begin{equation}
  \nabla F\left(\vectz\right) = \bfC^\herm \left(\vert \bfC \vectz\vert^2 - \bfy\right)\label{eq:complexGradientWFPPR}.
\end{equation}

\paragraph{Optimal step-size selection} We combine acceleration for WF with the optimal step-size selection proposed in \cite{jiang_wirtinger_2016} for the standard WF algorithm.
For completeness, we reproduce here the main ingredients underpinning optimal step size selection in \eqref{eq:WFGradientStep} and refer the reader to \cite{jiang_wirtinger_2016} for further details.
At iteration $k$, the optimal stepsize $\mu_k$ is defined by line search, \ie
\begin{equation}
  \mu_k \triangleq \argmin_\mu F\left(\vectx^{(k+1)}\right) \triangleq F\left(\vectz^{(k)} - \mu\nabla F\left(\vectz^{(k)}\right)\right)\label{eq:WFstepSizeLineSearch}
\end{equation}
The authors in \cite{jiang_wirtinger_2016} showed that the 1D optimization problem \eqref{eq:WFstepSizeLineSearch} boils down to finding the roots of a univariate cubic polynomial with real coefficients, the latter being completely determined by the knowledge of  $\vectz^{(k)}$,  $\nabla F\left(\vectz^{(k)}\right)$ and $\bfy$, see \cite[Eq. (17)]{jiang_wirtinger_2016}.
Roots can be determined in closed-form, and two cases can occur: \emph{(a)} there is only one real root, and thus it gives the optimal step-size $\mu_k$; \emph{(b)} there are three real roots, and in this case $\mu_k$ is set to the real root associated to the minimum objective value.
Note that optimal selection for WF is somewhat inexpensive, with computational cost dominated by the calculation of the cubic polynomial coefficients scaling as $\calO(MP)$.

\paragraph{Initialization}

Since PPR-WF attempts a minimizing a nonconvex quadratic objective \eqref{eq:WFnonconvexProblem}, the choice of initial points $\vectx^{(0)}$, $\vectx^{(1)}$ is crucial to hope that PPR-WF will be able to recover a global minimizer of the objective function.
For simplicity, in this work we set $\vectx^{(1)} = \vectx^{(0)}$, so that we only discuss the selection of $\vectx^{(0)}$.
In this work we consider four different initialization strategies for PPR-WF:
\begin{itemize}
  \item  \emph{spectral initialization} \cite{candes_phase_2015Wirt}: this standard approach consists in computing the eigenvector $\bfv$ corresponding to the largest eigenvalue of the matrix
        \begin{equation}
          \bfY \triangleq \frac{1}{MP} \sum_{r=0}^{MP - 1} y_r \bfc_r\bfc_r^\herm
        \end{equation}
        and to rescale it properly to set
        \begin{equation}
          \vectx^{(0)} = \frac{\bfv}{\lambda}, \quad \lambda = \left(N \dfrac{\sum_{r=0}^{MP-1}y_r}{\sum_{r=0}^{MP-1} \Vert \bfc_r \Vert^2}\right)^{1/2}
        \end{equation}

  \item \emph{random phase initialization}: we first generate a random measurement phase vector  $\boldsymbol{\phi} \in \bbR^{MP}$ with i.i.d. entries $\phi_r \sim \calU([0, 2\pi])$. Then, we set
        \begin{equation}
          \vectx^{(0)} = \bfC^\dagger \tilde{\bfy}, \quad \tilde{\bfy} \triangleq \bfy \odot \exp(\bmj \boldsymbol{\phi})
        \end{equation}
        where $\bfC^\dagger$ is the pseudo-inverse of $\bfC$ and $\odot$ denotes entrywise product between vectors.

  \item \emph{left kernel Sylvester initialization}: we simply set $\vectx^{(0)}$ as the result of the left-kernel Sylvester method.
  \item \emph{right kernel Sylvester initialization} we simply set $\vectx^{(0)}$ as the result of the right-kernel Sylvester method.
\end{itemize}

\paragraph{Convergence monitoring}
We monitor convergence of PPR-WF by computing at each iteration $k$, the normed residual $\Vert\vectx^{(k+1)} - \vectx^{(k)}\Vert_2 / \Vert\vectx^{(k)}\Vert_2$ and stop the algorithm when it goes below a prescribed tolerance $\varepsilon \ll 1$.

\paragraph{Complexity} The computational cost per iteration of PPR-WF is dominated by the evaluation of the complex gradient \eqref{eq:complexGradientWFPPR}, which scales as $\calO(MPN)$.
Note that the optimal step-size selection procedure scales as $\calO(MP)$, so meaning that the whole cost of PPR-WF remains $\calO(MPN)$ per iteration.
Algorithm \ref{algo:WFPPR} summarizes the proposed PPR-WF algorithm.

\begin{algorithm}[t]
  \SetAlgoLined
  \KwIn{measurements $\bfy \in \bbR^{MP}$, measurment matrix $\bfC \in \bbC^{MP\times 2N}$, tolerance $\varepsilon$}
  set $\vectx^{(0)}$ using the desired initialization method\;
  $\vectx^{(1)} \leftarrow\vectx^{(0)}$\;
  $k \leftarrow 1$\;
  \While{$\Vert \vectx^{(i+1)}-\vectx^{(i)}\Vert_2 > \varepsilon\Vert\vectx^{(i)}\Vert_2$}{
    $\beta_k \leftarrow \frac{k+1}{k+3}$\;                                                                         $\vectz^{(k)}  \leftarrow \vectx^{(k)} + \beta_k \left(\vectx^{(k)} - \vectx^{(k-1)}\right)$\;
    compute optimal step-size $\mu_k$ \protect\eqref{eq:WFstepSizeLineSearch}\;
    $\vectx^{(k+1)} \leftarrow   \vectz^{(k)} - \mu_k\nabla F\left(\vectz^{(k)}\right)$\;
    $k\leftarrow  k+1$\;
  }
  $\hat{\boldsymbol{\xi}} \leftarrow \vectx^{(k)}$\;

  \KwResult{estimate $\hat{\boldsymbol{\xi}}$}
  \caption{Wirtinger Flow for \protect\ref{prob:PPR}: PPR-WF}
  \label{algo:WFPPR}
\end{algorithm}

\section{Numerical experiments}
\label{sec:numericalExperiments}
We provide in this section several numerical experiments that address how \ref{prob:PPR} can be solved in practice using both algebraic and algorithmic  approaches described in Section \ref{sec:solvingPPR_algebraic} and Section \ref{sec:solvingPPR_iterative}, respectively.
Importantly, we demonstrate that the use of Wirtinger Flow together with a right-Sylvester initial point achieves the best performance in terms of mean-square error (MSE) with limited computational burden.
This combination of algorithmic and algebraic reconstruction methods provides a scalable, asymptotically MSE optimal, and parameter free inversion procedure for \ref{prob:PPR}.

Just like in standard phase retrieval, the global phase ambiguity in \ref{prob:PPR} requires to properly realign any estimated signal $\hat{\bfX}$ with the ground truth $\bfX$ in order to provide a meaningful MSE value.
Thus, the MSE is defined as
\begin{equation}
  \text{MSE} \triangleq \bfE \Vert \tilde{\hat{\bfX}} - \bfX \Vert_F^2\quad \text{ where }  \tilde{\hat{\bfX}} \triangleq e^{\bmj \Phi_0}\hat{\bfX} \text{ with } \Phi_0 \triangleq \argmin_{\phi \in [0, 2\pi)} \Vert e^{\bmj\phi}\hat{\bfX} - \bfX\Vert_F^2\:.
\end{equation}
Note that in practice, the minimization involved in the realignment procedure can simply be performed by evaluating the complex phase of the standard inner product between the vectors $\hat{\boldsymbol{\xi}}$ and $\vectx$ obtained from matrices $\hat{\bfX}$ and $\bfX$, respectively.

This section is organized as follows.
Section \ref{sub:noiselessReconstruction} presents the reconstruction of a realistic bivariate pulse from noiseless \ref{prob:PPR} measurements using the different approaches presented in the paper.
Section \ref{sub:initWFexp} then discusses the choice of initialization in PPR-WF.
Section \ref{sub:noisyPerf} benchmarks the robustness to noise of proposed reconstructions methods.
Finally, Section \ref{sub:nbMeasurementPerf} provides a first study of the impact of the number of \ref{prob:PPR} measurements on reconstruction performances.

\subsection{Reconstruction of bivariate pulse}
\label{sub:noiselessReconstruction}

As a first experiment, we consider the reconstruction of a bivariate pulse from noiseless \ref{prob:PPR} measurements.
The signal to be recovered defines a typical complex-valued bivariate analytic signal associated to the bivariate electromagnetic field to be estimated in ultra-short electromagnetic pulses experiments, see e.g. \cite{smirnova_attosecond_2009,walmsley_characterization_2009}.
It is defined for $N = 64$ points and we consider the simple noise-free measurement scheme \eqref{eq:measurementSchemeSimple} with $M= 2N-1$ and $K = 4$.
The bivariate pulse exhibits slow variations of the instantaneous polarization state, ensuring uniqueness of the \ref{prob:PPR} solution.
We investigate the capacity of the several methods introduced in Section \ref{sec:solvingPPR_algebraic} and Section \ref{sec:solvingPPR_iterative} to properly recover the bivariate signal of interest.
Note that for Wirtinger Flow, we consider two initialization strategies, one using spectral initialization and the other one based on the solution given by the right kernel Sylvester approach.

Figure \ref{fig:PulseReconstruction} depicts the different reconstructed bivariate signals obtained by each method along with the associated squared error $(\hat{\bfx}[n] - \bfx[n])^2$ for every time index $n$, where the estimated signal $\hat{\bfx}$ is realigned with the ground truth $\bfx$ beforehand.
Excepted Wirtinger Flow with spectral initialization, all methods successfully recover the original bivariate signal, where successful recovery in the noiseless context is decided whenever $\Vert \hat{\bfX} - \bfX \Vert_2^2 < 10^{-20}$.
Left and right-kernel Sylvester and Wirtinger Flow with right-Sylvester initialization provide similar reconstruction quality, with a slight advantage to left-kernel sylvester.
The SDP approach performs also well, yet three or four order of magnitude of MSE above the previous approaches.
Due to the very low error levels involved here, this has little consequence; however, compared to the aforementioned methods SDP exhibits both larger memory usage and overall computational cost, which makes it a less attractive option to solve this \ref{prob:PPR} problem in the noiseless scenario.
Strikingly, one can observe that the Wirtinger Flow approach relying on spectral initialization is not able to recover the ground truth signal.
Intuitively, it may be explained by the fact that spectral initialization provides an initial point too far from the global optimum, resulting in Wirtinger Flow to get stuck in a local minima instead.
This first experiment suggests that the performance of WF-based methods for \ref{prob:PPR} is tightly related to the quality of initial points, which we will  investigate in detail in the next section.

\begin{figure}[t]
  \includegraphics[width=\textwidth]{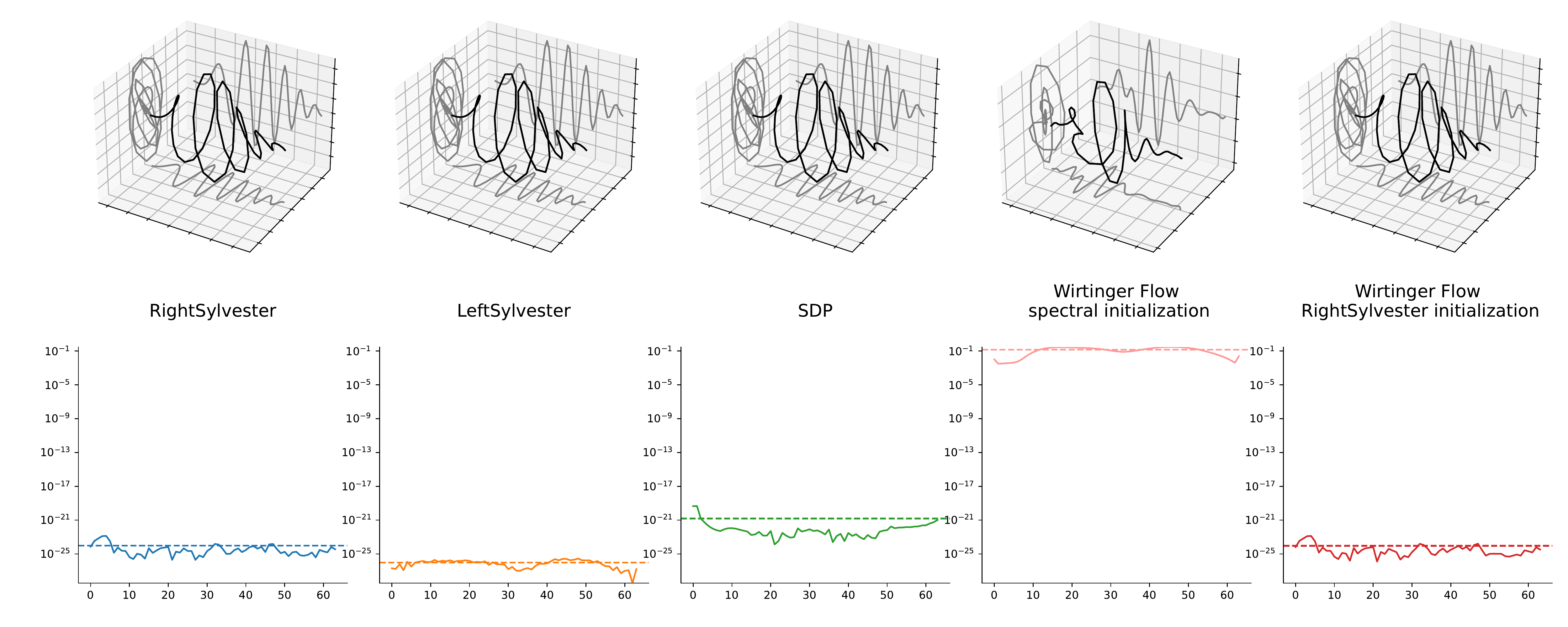}
  \caption{Reconstruction of a bivariate pulse ($N = 64$) from noiseless PPR measurements ($M=2N-1, P=4)$ using the different methods described in this paper. The reconstructed signal trace and squared error per time index $n$ are shown for each approach.}\label{fig:PulseReconstruction}
\end{figure}

\subsection{Comparison of initialization strategies for PPR-WF}
\label{sub:initWFexp}
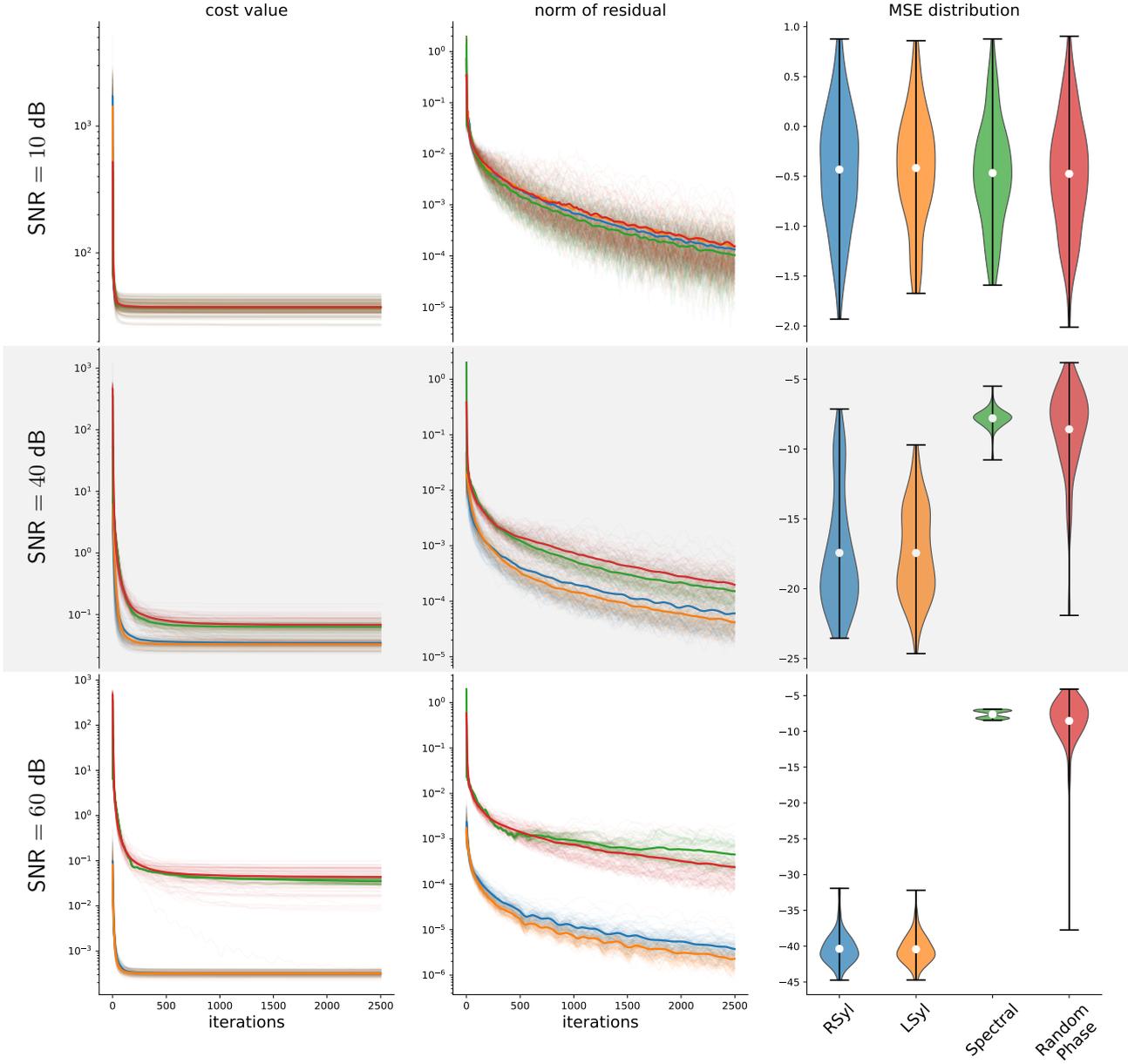
\begin{figure}[htbp]
  \hspace*{-1cm}
  \begin{tikzpicture}

    \node[right] at (0, 0) {\includegraphics[width=\textwidth]{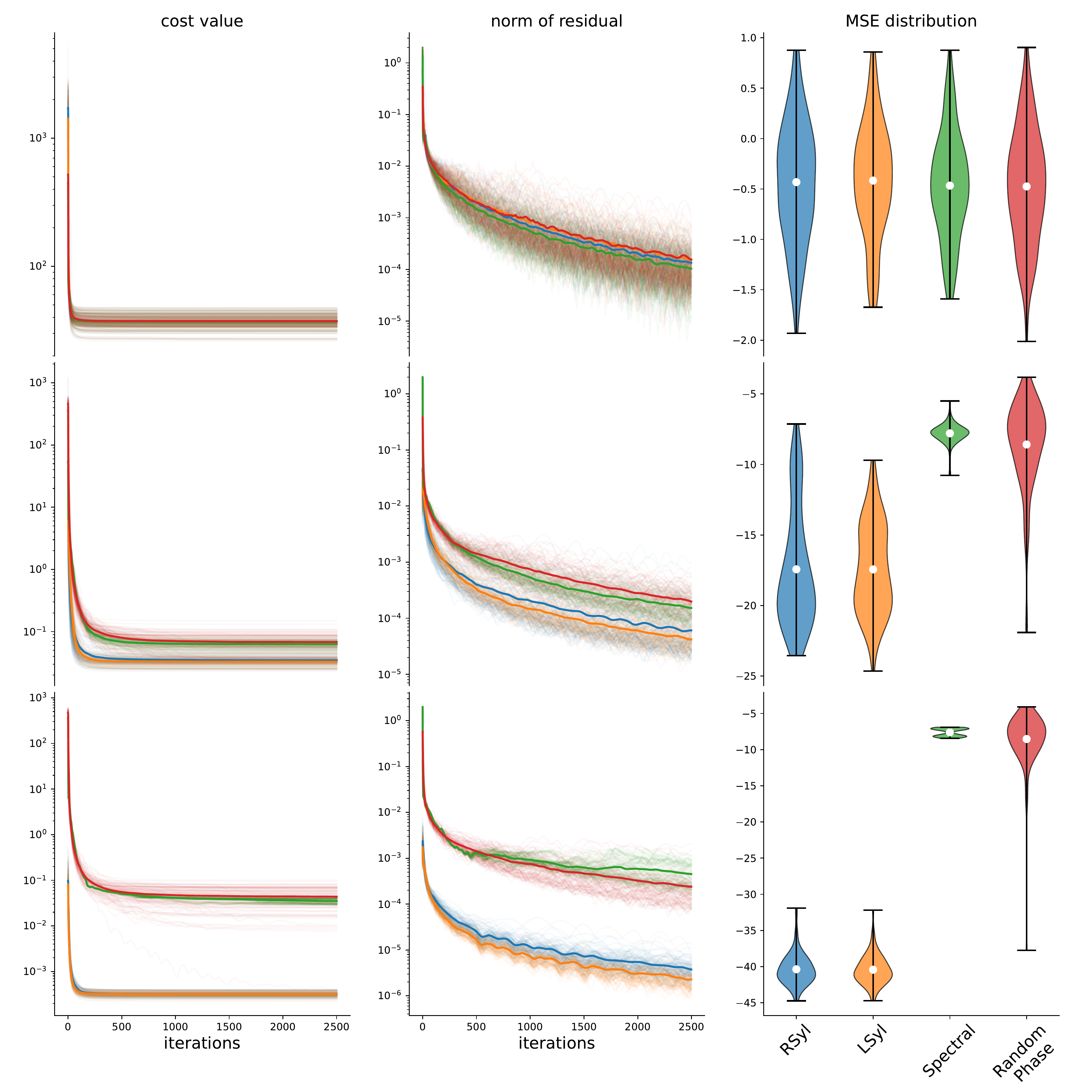}};
    \draw[fill = gray, fill opacity=0.1, draw=none] (-0.5, 2.8) rectangle  (16.5, -2.15);

    \node[rotate=90] at (0, 5.5) {\sf SNR $= 10$ dB};
    \node[rotate=90] at (0, 0.5) {\sf SNR $= 40$ dB};
    \node[rotate=90] at (0, -4.5) {\sf SNR $= 60$ dB};
  \end{tikzpicture}

  \caption{Comparison of initialization strategies for PPR-WF for the recovery of an arbitrary random bivariate signal of length $N=32$ with $M=2N-1$ and $P=4$ noisy measurements.
    We benchmark spectral initialization, random phase initialization, left and right-kernel Sylvester initialization strategies in terms of cost function evolution, normed residual decrease and MSE distribution. Rows corresponds to values of SNR of $10, 40$ and $60$ dB, respectively. For each SNR value, left and middle panels present the evolution of the cost function and residual value with iterations, respectively. For each initialization method, thin colored lines indicate trajectories for each one of the 100 independent trials, and thick colored lines display their average respective average.  The right panel provides violin plots representing a kernel density estimate of MSE distribution associated to each initialization strategy. White dots indicate mean-values and horizontal bars extreme values for each MSE distribution.}\label{fig:InitWFComparison}
\end{figure}

Choice of initial points in nonconvex problems is usually a difficult but crucial task, as it directly impacts whether or not the considered algorithm will be able to recover the global optimum of the problem.
The proposed PPR-WF algorithm does not avoid this key bottleneck, as already illustrated by the bivariate pulse recovery experiment depicted in Figure \ref{fig:PulseReconstruction}.
To assess the role played by initial points in PPR-WF, we carefully benchmark the four initialization methods described in Section \ref{sub:Wirtingerflow}, that is spectral initialization, random phase initialization, left and right-kernel Sylvester.
We generated a random Gaussian complex-valued signal $\bfX \in \bbC^{N\times 2}$ with i.i.d. entries of length $N = 32$ such that $\Vert \bfX \Vert_F = 1$ which was fixed for all experiments.
\ref{prob:PPR} noisy measurements \eqref{eq:noisyMeasuremntsPPR} were considered for the simple measurement scheme \eqref{eq:measurementSchemeSimple} with $M=2N-1$, $P=4$.
We investigated three values of SNR, of $10, 40$ and $60$ dB respectively.
For each SNR value, we generated $100$ independent noisy measurements and run the proposed PPR-WF algorithm using the four aforementioned initialization procedures.

Figure \ref{fig:InitWFComparison} depicts obtained reconstruction results for the three SNR scenarios, where we compare initialization methods in terms of cost function evolution $F(\vectx^{(k)})$ and normed residual $\Vert\vectx^{(k+1)} - \vectx^{(k)}\Vert_2 / \Vert\vectx^{(k)}\Vert_2$ decrease.
Note that we imposed a identical number of 2500 iterations of PPR-WF for each approach to ensure fair comparisons.
We also plot the empirical distribution of MSE values for each initialization for further comparison of the quality of the reconstructed signal (recall that MSE values are calculated after proper realignment of the estimated signal with the ground truth).
For SNR = $10$ dB (which is a very challenging scenario for PPR), there are no noticeable difference between initialization strategies: they provide similar results in terms of cost value decrease, residual evolution and MSE distribution.
For SNR = $40$ dB, one starts to observe significant differences between Sylvester-based approaches and spectral/random phase initializations.
On average, Sylvester-based initial points provides smaller optimal values, faster decrease of the residual and better reconstruction results in terms of MSE.
This behavior is accentuated for SNR $=60$ dB, where spectral and random phase initialization are unable to ensure convergence of PPR-WF to the global optimum.
This agrees with the observations made in Figure \ref{fig:PulseReconstruction} in the noiseless case for spectral initialization.

These results demonstrate the importance of the choice of the initial point in PPR-WF towards good convergence properties and recovery performance.
Overall, left and right-kernel Sylvester initializations systematically outperform spectral and random phase strategies.
While the left-kernel approach displays a slight advantage over the right-kernel approach in terms of residual decrease, it involves a much more important computational cost than its right-kernel counterpart.
This explains why we recommend to use right-kernel Sylvester initialization with PPR-WF for the best trade-off between algorithmic recovery performance and computational time.

\subsection{Recovery performance with noisy measurements}
\label{sub:noisyPerf}
\begin{figure}[t!]
  \centering
  \includegraphics[width=.75\textwidth]{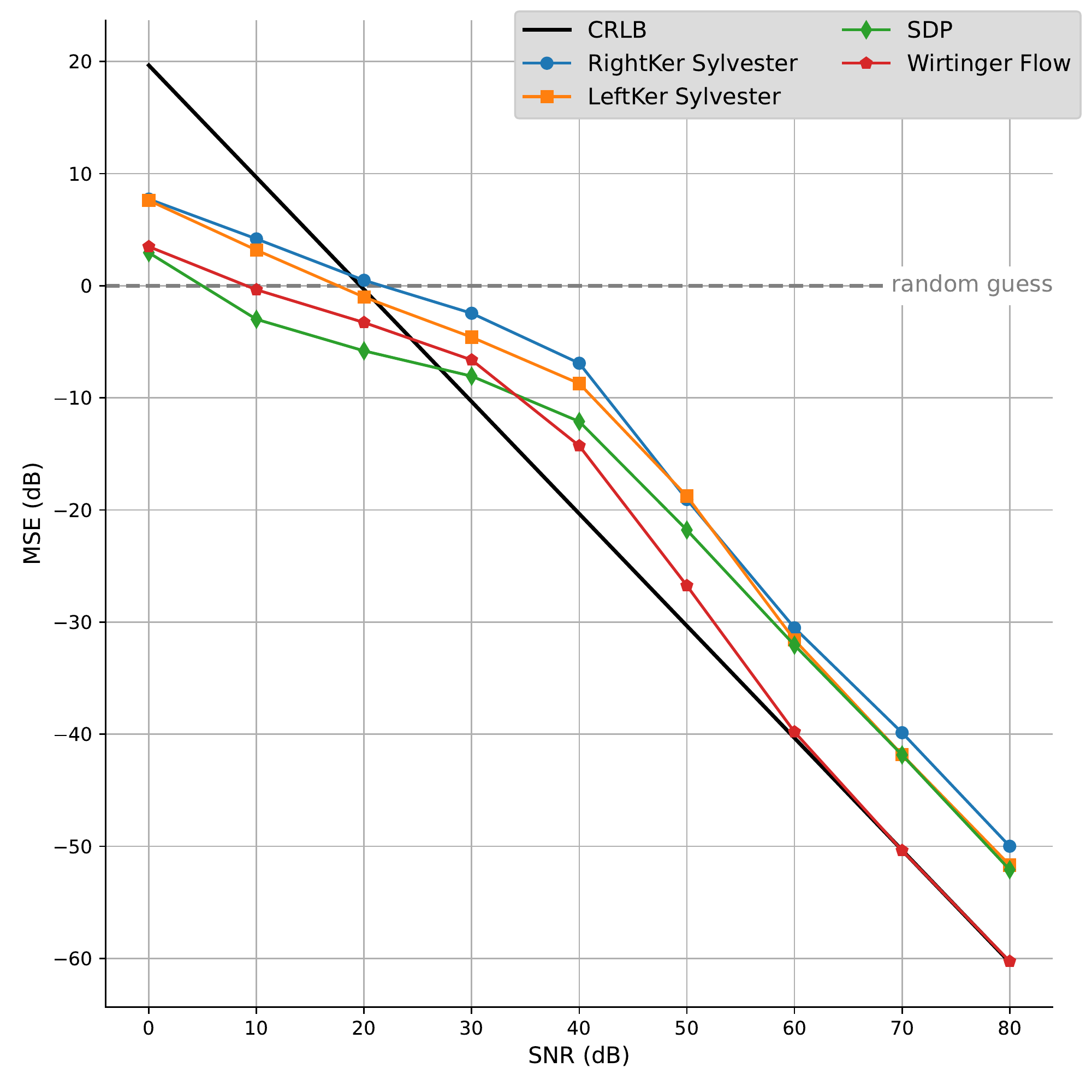}
  \caption{Evolution of the MSE with the SNR for the four PPR reconstruction methods proposed in this paper. Ground truth is randomly generated bivariate signal with $N=32$. Simple measurement scheme for $M=2N-1$ and $P=4$ was used.Thick black line indicate the corresponding Cramèr-Rao lower bound analytically derived in \ref{app:cramerRaoBoundPPR}.}
  \label{fig:noisyPerf}
\end{figure}

We now investigate the recovery performances of the different proposed algorithms for \ref{prob:PPR} when dealing with noisy measurements.
We consider an additive white Gaussian noise model \eqref{eq:noisyMeasuremntsPPR} for which the SNR is defined in \eqref{eq:SNRdefAWGNmodel}.
We generated a ground truth signal $\bfX \in \bbC^{N\times 2}$ with i.i.d. Gaussian entries of length $N = 32$ such that $\Vert \bfX \Vert_F = 1$ and selected the simple, $M=2N-1, P=4$ measurement scheme \eqref{eq:measurementSchemeSimple}.
For a given SNR value, the MSE associated with each one of the proposed methods to solve \ref{prob:PPR} was obtained by averaging of 100 independent reconstructions.
Note that we used our recommended right-kernel Sylvester initialization with PPR-WF, as explained in Section \ref{sub:initWFexp}.

Figure \ref{fig:noisyPerf} displays the evolution of MSE for values of SNR ranging from $0$ dB to $80$ dB.
As expected, the MSE decreases as the SNR increases, independently from the considered method.
Overall, algorithmic methods (PPR-WF and SDP) outperform algebraic (left and right-kernel Sylvester) ones in terms of MSE values.
More precisely, algebraic methods are not informative in the ``low-SNR'' regime (SNR $\leq 30$ dB) as they provide (relative) MSE values above $0$ dB, meaning that they do not provide a better reconstruction than a simple i.i.d. random guess scaled to the ground truth norm.
Furthermore we observe that SDP is more robust to noise than PPR-WF.
This agrees with the fact that SDP methods are known to be robust to noise in general.
Remarkably, the high-SNR regime ($\geq 60$ dB) highlights several distinctive behaviors.
First, we observe that beyond SNR $ = 40$ dB, PPR-WF outperforms all other methods, including SDP, by a few dB up to about 10 dB of relative MSE in the asymptotic regime.
Second, SDP do not longer outperforms left-kernel Sylvester, and only improves from the right-kernel Sylvester approach by a small margin.
This shows that, in this high-SNR regime, the computational burden associated to the SDP approach becomes prohibitive as 1) it provides no clear advantage over computationally cheaper algebraic methods and 2) it clearly underperforms PPR-WF.

For completeness, we also provide the Cramèr-Rao lower bound (CRLB) for the PRR measurement model \eqref{eq:noisyMeasuremntsPPR} to characterize a lower bound on the MSE of any unbiased estimator of the ground truth signal.
An analytical derivation of the resulting CRLB is given in \ref{app:cramerRaoBoundPPR}.
Figure \ref{fig:noisyPerf} displays the CRLB on top of MSE values obtained for each reconstruction method.
We observe that the CRLB is not informative below SNR $\leq 20$ dB as all methods provide smaller MSE values -- it simply means that the CRLB is particularly pessimistic in this regime.
On the contrary, the CRLB provides a meaningful lower bound in the high-SNR regime.
Importantly, it demonstrates that PPR-WF is an asymptotically optimal reconstruction method for \ref{prob:PPR} since it attains the CRLB for SNR $\geq 60$ dB.

\subsection{Influence of number of measurements}
\label{sub:nbMeasurementPerf}
\begin{figure}[t]
  \includegraphics[width=\textwidth]{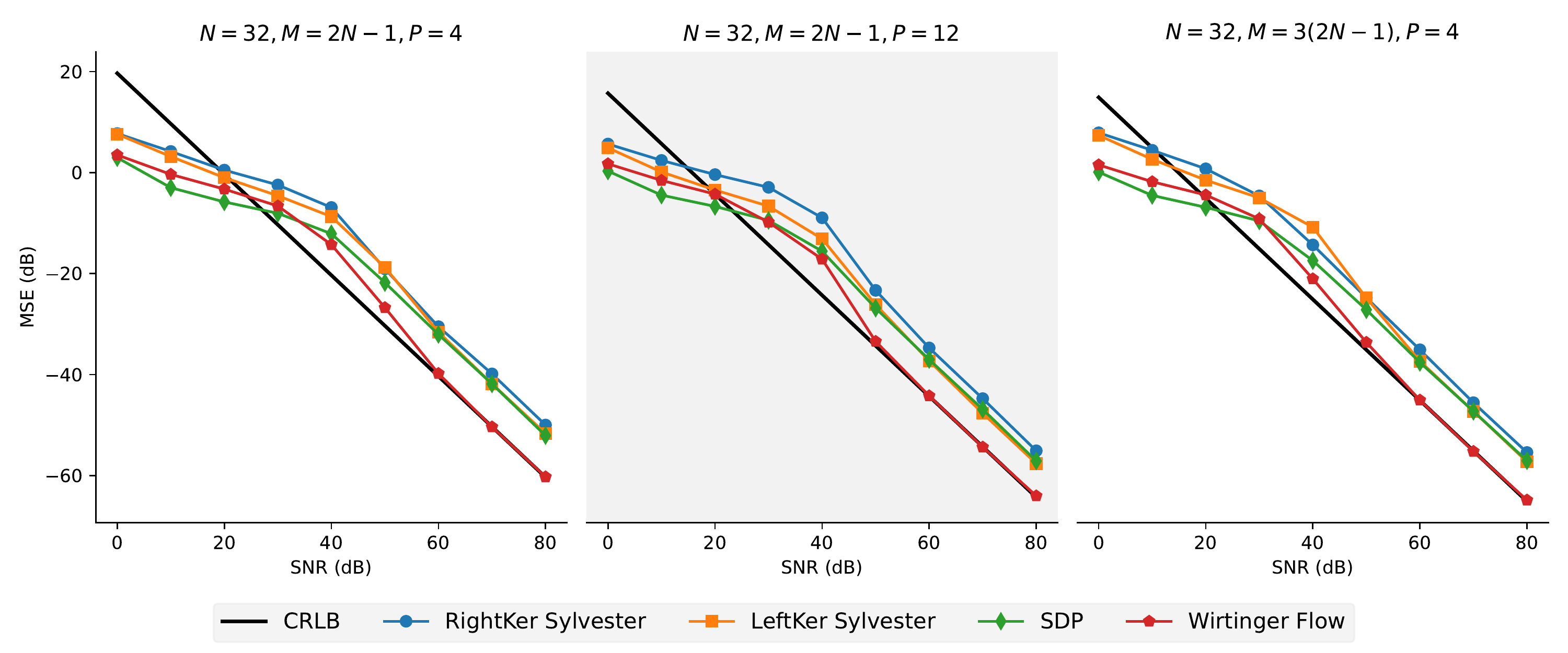}
  \caption{Comparison of the evolution of the MSE with respect to SNR for three measurements scheme $M=2N-1, P=4$ (left), $M = 2N-1, P=12$ (center) and $M=3(2N-1), P=4$ (right). Experiments follow the same protocol as described in Section \ref{sub:noisyPerf}.}\label{fig:comparisonMK}
\end{figure}

\begin{figure}[htbp]
  \includegraphics[width=\textwidth]{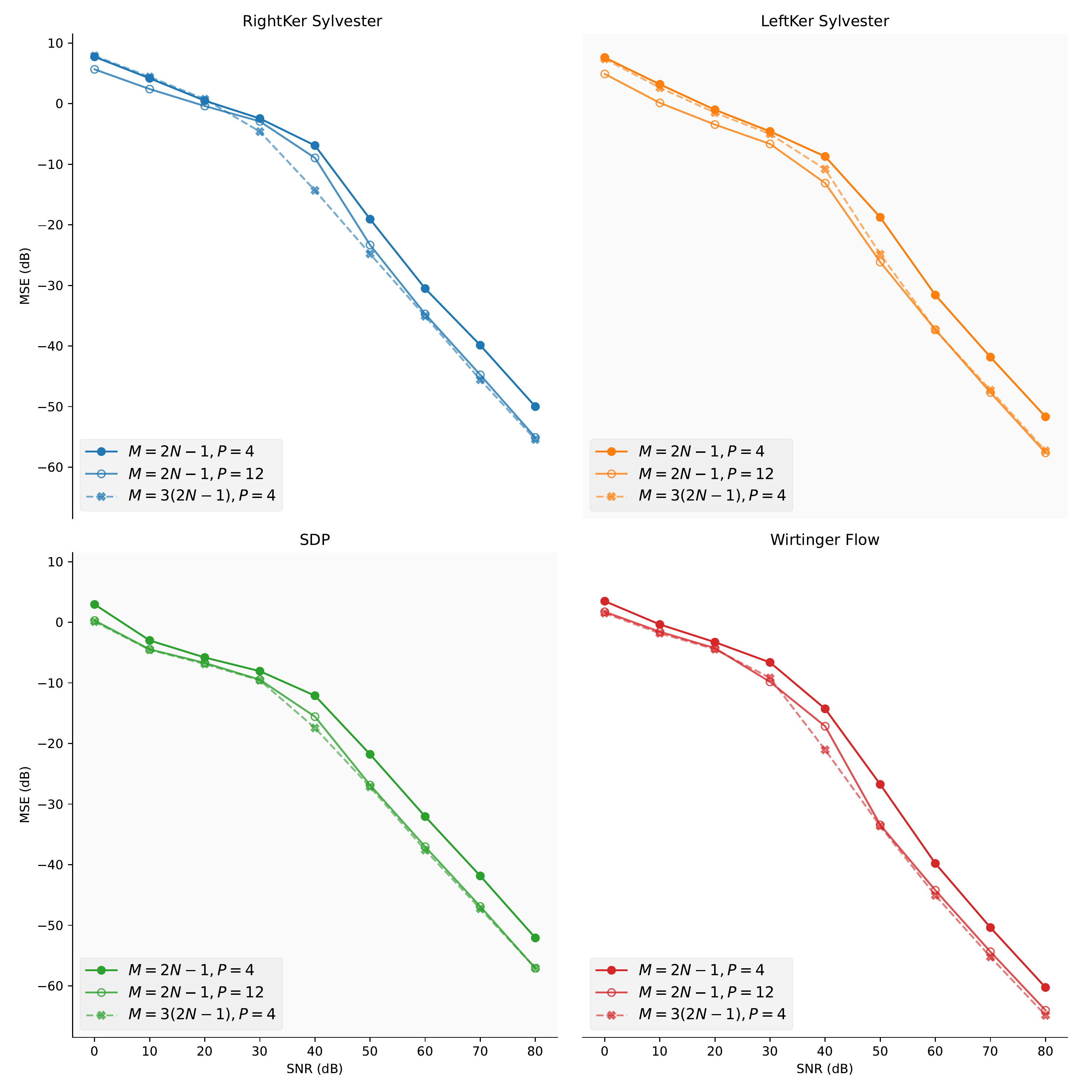}
  \caption{Side-by-side comparison of the behavior of each proposed reconstruction method for the three measurements scheme $M=2N-1, P=4$, $M = 2N-1, P=12$ and $M=3(2N-1), P=4$.}\label{fig:comparisonIndividualAlgoMK}
\end{figure}
One of the key advantages of the polarimetric measurement model in \ref{prob:PPR} is that one can easily increase the number of measurements $MP$ by performing more polarimetric projections, \ie by increasing $P$.
In fact, in practical experiments it may be oftentimes easier to set up a new polarizer state $\bfb_p$ than to change the actual detector, which would be required if one desires to increase the number of Fourier measurements $M$.
Therefore, a natural question is the following: if one desires to increase the total number of measurements $MP$, is it better -- in terms of MSE -- to increase the number of Fourier measurements $M$ or to increase the number of polarimetric projections $P$?
This is a vast topic related to the question of experimental design, which requires a specific treatment which is outside the scope of the present paper.
Nonetheless, we provide in the sequel a first study of the influence of the number of measurements in \ref{prob:PPR} for completeness.

Following the MSE performance analysis in Section \ref{sub:noisyPerf}, we use the same randomly generated ground truth signal $N=32$ and investigate the performances for two cases, \ie $M = 2N-1, P=12$ and $M=3(2N-1), P=4$, which lead to the same total number of measurements $MP$.
More precisely, the measurement scheme corresponding to each case is:
\begin{itemize}
  \item $M = 2N-1, P=12$ case: we use the correspondence between the 2-sphere and $\bbC^2$ to take advantage of optimal spherical tesselations such as HEALPix \cite{gorski2005healpix}.
        In physical terms, it can interpreted as finding one of the many possible Jones vector $\bfb_p$ corresponding to the Stokes parameters defining the rank-one matrix $\bfb_p\bfb_p^\herm$.
        Formally, given Cartesian coordinates $(s^x_p, s^y_p, s_p^z) \in \bbR^3$ of a point on the unit 2-sphere, we define the projection vector $\bfb_p$ as:
        \begin{equation}
          \bfb_p \triangleq \frac{1}{\sqrt{2}\sqrt{1+s_p^z}}\begin{bmatrix}
            \bmj s_p^x \\
            s_p^y+ (1+s^z_p)\bmj
          \end{bmatrix}\quad \text{if } s_p^z \neq -1, \quad \bfb_p \triangleq \begin{bmatrix}
            \bmj \\
            0
          \end{bmatrix}\quad \text{if } s_p^z = -1\:.
        \end{equation}
        Note that our choice of $P=12$ corresponds to the first level of HEALPix sphere discretization.

  \item $M=3(2N-1), P=4$ case: we keep the simple polarimetric measurement scheme \eqref{eq:projectionSchemeSimple} and increase the number $M$ of Fourier domain measurements.
\end{itemize}

Figure \ref{fig:comparisonMK} depicts MSE as a function of SNR for the two measurement setups described above, where results from the experiment in Section \ref{sub:noisyPerf} have been reproduced for better comparison.
As expected, increasing the total number of measurements $MP$ improves overall performance: this can be directly checked by remarking that the CRLB corresponding to $M = 2N-1, P=12$ and $M = 3(2N-1), P=4$ cases is lower that of the $M = 2N-1, P=4$ setup presented in Figure \ref{fig:noisyPerf}.
Moreover, the different proposed reconstructions method for \ref{prob:PPR} behave similarly with one another as in our description made in Section \ref{sub:noisyPerf}.
In particular, we note that PPR-WF also attains the CRLB in these two new setups, proving again that it establishes a versatile approach to solve \ref{prob:PPR}.

Figure \ref{fig:comparisonIndividualAlgoMK} provides a side-by-side comparison of these three measurement schemes for each reconstruction method.
First, remark that $M = 2N-1, P=12$ and $M = 3(2N-1), P=4$ scheme have similar CRLB MSE bounds, with a slight advantage to the $M = 3(2N-1), P=4$ case which can be observed on the PPR-WF panel.
Second, we note that for algorithmic approaches (SDP and PPR-WF), the difference concentrates in the mid-SNR regime, \ie between $30$ dB and $50$ dB, where oversampling in the Fourier domain offers slightly MSE improvement over increasing the number of polarimetric projections.
On the other hand, for algebraic approaches we observe that performing more polarimetric measurements usually improves the performance in the the low-SNR regime (SNR $\leq 30$ dB), eventhough algebraic approaches do not perform well in this scenario.
This performance improvement can be explained by the two-step nature of algebraic methods, which first need to reconstruct autocorrelation polynomials from polarimetric projections: in this case more polarimetric projections enable to reduce the reconstruction error in this first step.

\section{Conclusion}
\label{sec:conclusion}

In this paper, we have introduced a new model for Fourier phase retrieval called \ref{prob:PPR} that takes advantage of polarization measurements in applications involving polarized light.
Remarkably, we have shown that the problem of reconstructing a 1D bivariate signal from Fourier magnitude enjoys nice uniqueness properties which contrast with the non-uniqueness of standard 1D Fourier phase retrieval.
The theoretical study of \ref{prob:PPR} has been made possible by carefully drawing equivalences with two related problems, namely \ref{prob:BPR} and \ref{prob:PAF}.
Moreover, the derivation of uniqueness thanks to the polynomial factorization representation \ref{prob:PAF} motivated the development of two algebraic reconstructions methods for \ref{prob:PPR} using approximate greatest common divisor computations and Sylvester-like matrices.
We also carefully adapted SDP and Wirtinger-Flow methods to solve the \ref{prob:PPR} problem.
The extensive numerical experiments demonstrated the pros and cons of each approach.
It also allowed us to establish a scalable, computationaly efficient and robust to noise reconstruction strategy that combines both algebraic (right-kernel Sylvester initialization) and algorithmic (Wirtinger-Flow iterates) approaches.

Our conviction is that the new model \ref{prob:PPR} opens promising new avenues in for exploiting polarization in Fourier phase retrieval problems. For instance, an important challenge to be addressed lies in improving the performance of algebraic methods at low SNR, e.g. with more robust estimation of the measurement polynomials or adding some prior information about the signal to be recovered (e.g. smoothness).
These questions will be addressed in future work.

\appendix

\section{Polynomials and roots at infinity}\label{app:polynomials}
Let $\bbF = \bbR$ or $\bbC$ (or, in general, a field).
The vector space $\bbC_{\le D}[z]$ of polynomials of degree less than $D$ in subsection~\ref{sub:polynomials} is, in fact, isomorphic to the vector space $\bbC^{D+1}$ via the following one-to-one map:
\[
  \bfa =
  \begin{bmatrix}
    a[0] & a[1] & \cdots & a[D]
  \end{bmatrix}^{\transp} \mapsto
  A(z) = a[0] + z a[1] + \cdots + z^D a[D].
\]
In the literature (for example in algebraic geometry) a common way is to deal with bivariate homogeneous polynomials, that is polynomials of the form
\[
  y^D a[0] + y^{D-1} z a[1] + \cdots + z^D a[D],
\]
which belong to the space $\bbC_{=D}[z,y]$.
The roots of such polynomials belong to the projective space $\bbP^{1}$, which corresponds to the extended complex plane $\bbC \cup \{ \infty \}$.
However, for simplicity, we prefer to work with the polynomials in $\bbF_{\le D}[z]$ instead, and we refer the reader to  \cite[Ch. 8]{cox_ideals_1997} for details on homogeneous polynomials.

When working with the space $\bbF_{\le D}[z]$, the multiplication operation becomes the map from $\bbF_{\le D_1}[z] \times \bbF_{\le D_2}[z]$ to $\bbF_{\le (D_{1}+D_{2})}[z]$:
\begin{align*}
  (A(z), B(z)) & \mapsto C(z) = A(z)B(z),
\end{align*}
which in coordinates (i.e., for the vectors of coefficients $\bfa \in \bbF^{D_1+1}$ and $\bfb \in \bbF^{D_2+1}$) can be expressed as
\begin{equation}\label{eq:polyMultMatrixProduct}
  (\bfa, \bfb)  \mapsto  \bfc =  \multmat{b}{D_1} \bfa=  \multmat{a}{D_2} \bfb,
\end{equation}
where $\multmat{a}{L}$ is the multiplication matrix defined for any vector of coefficients
\[
  \bfa = \begin{bmatrix} a[0] & a[1] & \cdots & a[D]  \end{bmatrix},
\]
as the following matrix $(D+ L +1) \times (L+1)$
\begin{equation}\label{eq:multiplicationMatrix}
  \multmat{a}{L} \triangleq
  \underbrace{\begin{bmatrix}
      \pco{a}{0} &        &            \\
      \vdots     & \ddots &            \\
      \pco{a}{D} &        & \pco{a}{0} \\
                 & \ddots & \vdots     \\
                 &        & \pco{a}{D}
    \end{bmatrix}}_{L+1 \text{ columns}}.
\end{equation}
Armed with the definition of the multiplication, we can now give a formal justification to \eqref{eq:polyFactorizationFTA}.
Let us formally define
\[
  (z-\infty) \triangleq 0 \cdot z + 1  \in \bbC_{\le 1} [z].
\]
Then, thanks to the fundamental theorem of algebra, any nonzero polynomial $A \in \bbC_{\le D} [z]$ can be uniquely (up to permutation of roots) factorized as
\[
  A(z) = \lambda (z - \alpha_1) \cdots (z - \alpha_D),
\]
where for any $k$, $\alpha_k \in \bbC \cup \{\infty \}$.

Finally, we remark on the notion of the greatest common divisor, which, for two nonzero polynomials $A_1,A_2 \in \bbC_{\le D}[z]$ is a polynomial $H \in \bbC_{\le D'}[z]$  with highest possible $D'$, which is a divisor of both $A_1(z)$ and $A_2(z)$.
The GCD is defined uniquely up to a multiplication by a scalar in $\bbC \setminus \{ 0 \}$.
The same notion can be defined for several polynomials, see \cite[Section 2]{usevich_variable_2017} for more details.

\section{Relation between Fourier measurements and correlation polynomials}
\label{app:Fourier2polynomials}
\begin{proof}[{Proof of Lemma~\ref{lem:Fourier2polynomials}}]
  The first part of lemma follows from the correspondence between multiplication of polynomials and discrete convolution of two vectors, see
  \eqref{eq:polyMultMatrixProduct}.
  Next, recall that the discrete Fourier transform of $\lbrace \bfx[n]\rbrace_{n=0,1, \ldots, N-1}$ is denoted by $\fourierv{X}[m] = [\fourier{X}_1[m], \fourier{X}_2[m]]^\transp$ for $m=0, 1, \ldots, M-1$, see \eqref{eq:defDFT}.
  Then the
  Fourier entries can be related to polynomials $X_1(z)$ and $X_2(z)$ as follows:
  \begin{align*}
    \fourier{X}_1[m]  = X_1\left(e^{-\bmj 2\pi \frac{m}{M}}\right), \quad  \fourier{X}_2[m] = X_2\left(e^{-\bmj 2\pi \frac{m}{M}}\right),
  \end{align*}
  for any $m=0, 1, \ldots, M-1$. Similarly, thanks to \eqref{eq:conjReverse}, their conjugates can be expressed through the conjugate reflection polynomials $\widetilde{X}_1(z)$ and $\widetilde{X}_2(z)$
  \begin{align*}
    \conj{\fourier{X}_1[m]} & = \overline{X_1\left(e^{-\bmj 2\pi \frac{m}{M}}\right)} =
    \sum_{n=0}^{N-1} \conj{x_1[n]}e^{2\pi\bmj \frac{nm}{M}} = e^{\bmj 2\pi \frac{m(N-1)}{M}}\widetilde{X}_1\left(e^{-\bmj 2\pi \frac{m}{M}}\right), \\
    \conj{\fourier{X}_2[m]} & = \overline{X_2\left(e^{-\bmj 2\pi \frac{m}{M}}\right)} =
    \sum_{n=0}^{N-1} \conj{x_2[n]}e^{2\pi\bmj \frac{nm}{M}} = e^{\bmj 2\pi \frac{m(N-1)}{M}}\widetilde{X}_2\left(e^{-\bmj 2\pi \frac{m}{M}}\right).
  \end{align*}
  As a result, thanks to \eqref{eq:Gamma_m}, \ref{prob:BPR} measurements can be expressed in terms of measurement polynomials $\Gamma_{ij}(z)$ as follows:
  \begin{align*}
    \FourierGamma[m] & =
    \begin{bmatrix}
      \vert \fourier{X}_1[m]\vert^2           & \fourier{X}_1[m]\conj{\fourier{X}_2[m]} \\
      \fourier{X}_2[m]\conj{\fourier{X}_1[m]} & \vert \fourier{X}_2[m]\vert^2
    \end{bmatrix}                                                                                                                                                                                                                                                                                                                                                                                                                                   \\
                     & = e^{\bmj 2\pi \frac{m(N-1)}{M}} \begin{bmatrix}
                                                          X_1\left(e^{-\bmj 2\pi \frac{m}{M}}\right)  \widetilde{X}_1\left(e^{-\bmj 2\pi \frac{m}{M}}\right) & X_1\left(e^{-\bmj 2\pi \frac{m}{M}}\right)  \widetilde{X}_2\left(e^{-\bmj 2\pi \frac{m}{M}}\right) \\
                                                          X_2\left(e^{-\bmj 2\pi \frac{m}{M}}\right)  \widetilde{X}_1\left(e^{-\bmj 2\pi \frac{m}{M}}\right) & X_2\left(e^{-\bmj 2\pi \frac{m}{M}}\right)  \widetilde{X}_2\left(e^{-\bmj 2\pi \frac{m}{M}}\right)
                                                        \end{bmatrix}
    =
    e^{\bmj 2\pi \frac{m(N-1)}{M}} \boldsymbol{\Gamma}(e^{-\bmj 2\pi \frac{m}{M}}),
  \end{align*}
  which completes the proof.
\end{proof}

\begin{proof}[{Proof of Theorem~\ref{thm:BPR_eq_PAF}}]
  Here, we make use of the two one-to-one correspondences.
  Note that the mapping between $\bbC^{N+1}$ and $\bbC_{\le N}$ is a linear one-to-one map (and is an isomorphism), see \ref{app:polynomials}.
  Hence, the signals $X_1, X_2$ can be uniquely recovered from the polynomials and vice versa.

  Similarly, thanks to \eqref{eq:Gamma_evaluations}, the Fourier covariance measurements $\FourierGamma[m]$ are a linear transformation of the sequence
  \[
    \lbrace\boldsymbol{\Gamma}(e^{-\bmj 2\pi \frac{m}{M}})\rbrace_{m=0, \ldots, M-1}
  \]
  of evaluations of the matrix polynomial $\boldsymbol{\Gamma}(z)$ at a set of $M$ distinct points $\{e^{-\bmj 2\pi \frac{m}{M}}\}_{m=0, \ldots, M-1}$ on the complex plane.
  If $M \ge 2N-1$ (the degree of the polynomials + 1), then it is known that the coefficients of the polynomials can be uniquely recovered from the evaluations at $M$ distinct points, and therefore the following map is an injection
  \begin{align*}
    \bbC^{2\times 2}_{\le 2N-2} & \to  (\bbC^{2\times 2})^{M}                                \\
    \boldsymbol{\Gamma}(z)      & \mapsto \lbrace\FourierGamma[m]\rbrace_{m=0, \ldots, M-1},
  \end{align*}
  which completes the proof.

\end{proof}

\section{Proof of Theorem \ref{theorem:explicitFormPAF}}
\label{app:proof:theoremUniqueness}

Suppose that $H(z) = \gcd(\Gamma_{11}, \Gamma_{12}, \Gamma_{21}, \Gamma_{2}) =Q(z)\tilde{Q}(z)$ with $D$ pair roots $(\delta_i, \conj{\delta}_i^{-1})$.
We further assume that $(0, \infty)$ is not a pair root of $H(z)$.
Let $X_1(z) = Q(z)R_1(z)$ and $X_2(z) = Q(z)R_2(z)$.
By \cref{lemma:PAFspectralFactorisation}, the polynomials $R_1$ and $R_2$ can be determined up to one multiplicative constant by \eqref{eq:determinationR1R2_H}.
Let us denote by $\alpha_{1i}$ (resp. $\alpha_{2i}$) the $L-D-1$ roots of $R_1(z)$ (resp. $R_2(z)$), such that
\begin{equation*}
  R_1(z)  = \lambda_1 \prod_{i=1}^{N-D-1}(z-\alpha_{1i}), \qquad R_2(z) = \lambda_2 \prod_{i=1}^{N-D-1}(z-\alpha_{2i})
\end{equation*}
where $\lambda_1, \lambda_2 \in \mathbb{C}$ are constants, to be derived hereafter.
The recovery of $Q(z)$ from $Q(z)\tilde{Q}(z)$ is identical to the univariate case, see e.g. \cite{beinert2015ambiguities, boche_fourier_2017}.
Denoting by $(\delta_i, \overline{\delta_i}^{-1})$ the pair roots of $Q(z)\tilde{Q}(z)$, $Q(z)$ can be written as
\begin{equation}
  Q(z) = \prod_{i=1}^D(z-\beta_{i}), \quad \beta_i \in (\delta_i, \overline{\delta_i}^{-1})
\end{equation}
As explained in \cref{theorem:PAFuniqueness}, the number of different solutions for $Q(z)$ dictates the number of solutions for the \ref{prob:PAF} problem.
Thus, if polynomials $(X_1'(z), X'_2(z))$ are solutions to \ref{prob:PAF} then they can be expressed as
\begin{align}
  X_1'(z) & = \lambda_1 \prod_{i=1}^D(z-\beta_{i})\prod_{i=1}^{N-D-1}(z-\alpha_{1i}) \\
  X_2'(z) & = \lambda_2 \prod_{i=1}^D(z-\beta_{i})\prod_{i=1}^{N-D-1}(z-\alpha_{2i})
\end{align}
where $\lambda_1, \lambda_2$ remain to be determined.

To this aim, one writes the expression of the measurements polynomials in terms of $X_1'(z)$ and $X_2'(z)$ above. For instance:
\begin{equation}
  \begin{split}
    \Gamma_{11}(z) &= X_1'(z)z^{N-1}\overline{X_1'(\overline{z}^{-1})} \\
    &= \vert \lambda_1\vert^2 \prod_{i=1}^D(z-\beta_{i})\prod_{i=1}^{N-D-1}(z-\alpha_{1i}) \prod_{i=1}^D(1-\overline{\beta_{i}}z)\prod_{i=1}^{N-D-1}(1-\overline{\alpha_{1i}}z)
  \end{split}
\end{equation}
Using that $\Gamma_{11}(z) := \sum_{n=0}^{2N-2} \gamma_{11}[n-N+1]z^n$, identifying leading order coefficients yields
\begin{equation}
  \gamma_{11}[N-1]= \vert \lambda_1\vert^2(-1)^{N-1}\prod_{i=1}^D\overline{\beta_{i}}\prod_{i=1}^{N-D-1}\overline{\alpha_{1i}}
\end{equation}
Similarly, one gets
\begin{align}
  \gamma_{22}[N-1]= \vert \lambda_2\vert^2(-1)^{N-1}\prod_{i=1}^D\overline{\beta_{i}}\prod_{i=1}^{N-D-1}\overline{\alpha_{2i}} \\
  \gamma_{12}[N-1] = \lambda_1\overline{\lambda}_2(-1)^{N-1}\prod_{i=1}^D\overline{\beta_{i}}\prod_{i=1}^{N-D-1}\overline{\alpha_{2i}}
\end{align}
These relations determine uniquely the amplitudes of $\lambda_1, \lambda_2$ as well as the phase difference between $\lambda_1$ and $\lambda_2$.
Thus $\lambda_1, \lambda_2$ are unique up to a global phase factor $\exp(\bmj \theta), \: \theta \in [-\pi, \pi)$.
One obtains eventually the following expressions
\begin{align}
  \lambda_1 = e^{\bmj \theta}\left(|\gamma_{11}[N-1]| \prod_{i=1}^{D} | \beta_{i}| ^{-1}\prod_{i=1}^{N-D-1} | \alpha_{1i}|^{-1}\right)^{1/2} \\
  \lambda_2 = e^{\bmj (\theta - \Delta)}\left(|\gamma_{22}[N-1]| \prod_{i=1}^{D} | \beta_{i}| ^{-1}\prod_{i=1}^{N-D-1} | \alpha_{2i}|^{-1}\right)^{1/2}
\end{align}
with
\begin{align}
  \Delta & = \arg (\lambda_1\overline{\lambda_2})                                                                \\
         & = \pi(N-1)+ \arg  \gamma_{12}[N-1] + \sum_{i=1}^D\arg \beta_i + \sum_{i=1}^{N-D-1}\arg \alpha_{2i}\:.
\end{align}

\section{Sylvester matrices and greatest common divisors}\label{app:sylvesterAGCD}

\begin{proof}[Proof of Proposition~\ref{prop:sylvesterLeftKernel}]
  We first note that the result on the rank of $\sylv{D}(A,B)$  is known (see, for example, \cite[Theorem 4.7]{usevich_variable_2017}).
  Thus, we are left prove the second part, which is somewhat related to  \cite[Remark 4.8]{usevich_variable_2017}.

  We write $A(z) = F(z)H(z)$, $B(z) = G(z) H(z)$, so that $\gcd(A,B) = 1$ and $F,G \in \bbC_{\le L-K}[z]$.
  Consider the following multiplication matrix
  \[
    \multmat{h}{2L-D-K} =
    \underbrace{\begin{bmatrix}
        \pco{h}{0} &        &            \\
        \vdots     & \ddots &            \\
        \pco{h}{K} &        & \pco{h}{0} \\
                   & \ddots & \vdots     \\
                   &        & \pco{h}{K}
      \end{bmatrix}}_{2L-D-K+1 \text{ columns}},
  \]
  and our first goal is to show that the range of $\sylv{D}  (A,B)$ is a subset of the range  of $\multmat{h}{2L-D-K}$.
  Indeed, the range of $\sylv{D}  (A,B)$ corresponds to all polynomials $R(z) \in \bbC_{\le 2L-D}[z]$ that can be represented as
  \begin{equation}\label{eq:polyRangeSyld}
    R(z) = U(z) A(z) + V(z) B(z) = H(z) (U(z) F(z) + V(z)G(z)),
  \end{equation}
  and therefore any element in the range of  $\sylv{D}  (A,B)$ belongs to the range  of $\multmat{h}{2L-D-K}$ (since the range of $\multmat{h}{2L-D-K}$ corresponds to all polynomials of the form $H(z)Q(z)$ with $Q \in \bbC_{\le 2L-D-K}[z]$).

  Next we note that  $\multmat{h}{2L-D-K}$ is full column rank and therefore the ranks of $\colspan(\sylv{D}  (A,B))$ and $\colspan (\multmat{h}{2L-D-K})$ are equal.
  Hence the ranges of the two matrices coincide, as well as the left kernels; in particular the following equivalence holds true
  \[
    \bfu^{\transp} \sylv{D}  (A,B) = 0 \iff \bfu^{\transp} \multmat{h}{2L-D-K} = 0.
  \]
  Finally, easy algebraic calculations (see also, for instance, \cite[eqn. (33)]{usevich_variable_2017}) show that
  \[
    \bfu^{\transp} \multmat{h}{2L-D-K}  = \bfh^{\transp} \hankel{K+1}(\bfu),
  \]
  which completes the proof.
\end{proof}

\section{Cramèr-Rao bound for PPR}
\label{app:cramerRaoBoundPPR}
Several authors have considered Cramèr-Rao bounds for the classical phase retrieval problem with additive white gaussian noise \cite{balan2016reconstruction,bandeira_saving_2014,qian2016phase}.
These results directly apply to the additive Gaussian noise PPR model \eqref{eq:noisyMeasuremntsPPR} since it can be equivalently rewritten as a particular one-dimensional noise model thanks to  \ref{prob:PPR-1D} model introduced in Section \ref{sub:1DequivalentModelphysicsModel}.
For completeness, we provide below an alternative derivation of the Cramèr-rao bound described in \cite{qian2016phase}, where we use a full complex-domain approach instead of considering separate Cramèr-Rao bounds on amplitude and phase.
Since measurement noise $n_{m,p}$ is i.i.d. Gaussian distributed with variance $\sigma^2$, the pdf of the vector of observations $\bfy$ is given by
\begin{align}
  p(\bfy \vert \vectx) & = \prod_{m=0}^{M-1} \prod_{p=0}^{P-1} p(y_{m,p}\vert \vectx)                                                                                                 \\
                       & = \prod_{m=0}^{M-1} \prod_{p=0}^{P-1} \frac{1}{\sqrt{2\pi}\sigma}\exp\left[-\frac{\left(y_{m,p}-\vectx^\herm \bfC_{m,p}\vectx\right)^2}{2\sigma^2}\right]\:.
\end{align}
where we recall that $\bfC_{m,p} \triangleq \bfc_{m,p}\bfc_{m,p}^\herm$ with $\bfc_{m,p}=\overline{\bfb}_p\otimes \bfa_m$ by definition.
One obtains the log-likelihood of observations as
\begin{equation}
  \log p(\bfy \vert \vect{ \bfx}) = -\frac{MP}{2}\log(2\pi\sigma^2) - \frac{1}{2\sigma^2}\sum_{m=0}^{M-1} \sum_{p=0}^{P-1} \left(y_{m,p}-\vectx^\herm \bfC_{m,p}\vectx\right)^2
\end{equation}
Since one wants to estimate the complex parameter vector $\vectx$, it is necessary to use the complex Fisher Information Matrix (FIM) \cite{van1994cramer,loesch2012cramer,ollila2008cramer}, which reads
\begin{equation}
  \calJ_{\vectx} = \begin{bmatrix} \calI_{\vectx}        & \calP_{\vectx}        \\
                \conj{\calP_{\vectx}} & \conj{\calI_{\vectx}}
  \end{bmatrix} \in \bbC^{4N\times 4N}\label{eq:complexFIMdefinition}
\end{equation}
where entries are defined using Wirtinger derivatives \cite{kreutz-delgado_complex_2009} since $\vectx$ is a complex vector:
\begin{align}
  \calI_{\vectx} & = \bfE\left[\left(\nabla_{\conj{\vectx}} \log p(\bfy \vert \vectx)\right)\left(\nabla_{\conj{\vectx}} \log p(\bfy \vert \vectx)\right)^\herm\right]   \\
  \calP_{\vectx} & = \bfE\left[\left(\nabla_{\conj{\vectx}} \log p(\bfy \vert \vectx)\right)\left(\nabla_{\conj{\vectx}} \log p(\bfy \vert \vectx)\right)^\transp\right]
\end{align}
Note that the FIM $\calJ_{\vectx}$ defined in \eqref{eq:complexFIMdefinition} is isomorphic to the real FIM which would have been obtained by stacking the real and imaginary parts of $\vectx$ in a single long vector \cite{loesch2012cramer}.
This explains why $\calJ_{\vectx}$ has dimension $4N \times 4N$.
Using properties of Wirtinger derivatives, we obtain
\begin{equation}
  \nabla_{\conj{\vectx}} \log p(\bfy \vert \vectx) = -\frac{1}{\sigma^2}\sum_{m=0}^{M-1} \sum_{p=0}^{P-1}(y_{m,p}-\vectx^\herm \bfC_{m,p}\vectx)\bfC_{m,p}\vectx\:.
\end{equation}
This allows to compute explicitly the block terms $\calI_{\vectx}$ and $\calP_{\vectx}$ that define $\calJ_{\vectx}$.
Using noise independence, one gets
\begin{align}
  \calI_{\vectx} & = \frac{1}{\sigma^4}\bfE\left[\left(\sum_{m,p} (y_{m,p}-\vectx^\herm \bfC_{m,p}\vectx)\bfC_{m,p}\vectx\right)\left(\sum_{m',p'} (y_{m',p'}-\vectx^\herm \bfC_{m', p'}\vectx)\vectx^\herm \bfC_{m', p'}\right)\right] \\
                 & = \frac{1}{\sigma^4}\sum_{m,p,m',p'}\bfE\left[n_{m,p}n_{m',p'}\right]\bfC_{m,p}\vectx \vectx^\herm \bfC_{m',p'}                                                                                                      \\
                 & = \frac{1}{\sigma^2}\sum_{m, p}\bfC_{m,p}\vectx \vectx^\herm \bfC_{m,p}                                                                                                                                              \\
                 & = \frac{1}{\sigma^2}\sum_{m, p}\vert\bfc_{m,p}^\herm \vectx\vert^2 \bfc_{m,p}\bfc_{m,p}^\herm
\end{align}
Similar calculations leads to:
\begin{equation}
  \calP_{\vectx} = \frac{1}{\sigma^2}\sum_{ij}\bfC_{m,p}\vectx (\vectx)^\transp \bfC_{m,p}^\transp = \frac{1}{\sigma^2}\sum_{m, p}\left(\bfc_{m,p}^\herm \vectx\right)^2 \bfc_{m,p}\bfc_{m,p}^\transp
\end{equation}

A key result \cite{ollila2008cramer} is that the inverse of the complex FIM \eqref{eq:complexFIMdefinition} provides a lower bound on the covariance and pseudo-covariance of any unbiased estimator $\hat{\boldsymbol{\xi}}$ of the complex parameter $\vectx$:
\begin{equation}
  \begin{bmatrix}
    \mathrm{cov}\:\hat{\boldsymbol{\xi}}     & \mathrm{pcov}\:\hat{\boldsymbol{\xi}}     \\
    \conj{\mathrm{pcov}\:\hat{\boldsymbol{\xi}}} & \conj{\mathrm{cov}\:\hat{\boldsymbol{\xi}}}
  \end{bmatrix} \succeq \calJ_{\vectx}^{-1}
\end{equation}
When the complex FIM is singular -- as in phase retrieval \cite{balan2016reconstruction,bandeira_saving_2014} --, one can show its pseudo-inverse remains a valid lower bound for the MSE; following the discussion in \cite{qian2016phase}, we still refer to the resultant bound as the CRB with little abuse.
In particular, we obtain the following bound on the MSE on any unbiased PPR estimator $\hat{\bfX}$ for the model \eqref{eq:noisyMeasuremntsPPR}:
\begin{equation}
  \bfE \Vert \hat{\bfX} - \bfX\Vert_F^2 = \bfE \Vert \hat{\boldsymbol{\xi}} - \vectx\Vert_2^2 = \trace  \mathrm{cov}\hat{\boldsymbol{\xi}} \geq \trace \left(\left[\calJ_{\vectx}^{\dagger}\right]_{[:2N, :2N]}\right)
\end{equation}
where the subscript $\phantom{0}_{[:2N, :2N]}$ denotes the restriction to the upper-left block of $\calJ_{\vectx}^{\dagger}$.

\bibliography{refs.bib}

\begin{thebibliography}{10}
\expandafter\ifx\csname url\endcsname\relax
  \def\url#1{\texttt{#1}}\fi
\expandafter\ifx\csname urlprefix\endcsname\relax\def\urlprefix{URL }\fi
\expandafter\ifx\csname href\endcsname\relax
  \def\href#1#2{#2} \def\path#1{#1}\fi

\bibitem{sayre_implications_1952}
D.~Sayre, \href{http://scripts.iucr.org/cgi-bin/paper?S0365110X52002276}{Some
  implications of a theorem due to {Shannon}}, Acta Crystallographica 5~(6)
  (1952) 843--843.
\newblock \href {http://dx.doi.org/10.1107/S0365110X52002276}
  {\path{doi:10.1107/S0365110X52002276}}.
\newline\urlprefix\url{http://scripts.iucr.org/cgi-bin/paper?S0365110X52002276}

\bibitem{elser2003phase}
V.~Elser, Phase retrieval by iterated projections, JOSA A 20~(1) (2003) 40--55.

\bibitem{elser2018benchmark}
V.~Elser, T.-Y. Lan, T.~Bendory, Benchmark problems for phase retrieval, SIAM
  Journal on Imaging Sciences 11~(4) (2018) 2429--2455.

\bibitem{millane1990phase}
R.~P. Millane, Phase retrieval in crystallography and optics, JOSA A 7~(3)
  (1990) 394--411.

\bibitem{fienup_reconstruction_1978}
J.~R. Fienup,
  \href{https://www.osapublishing.org/abstract.cfm?URI=ol-3-1-27}{Reconstruction
  of an object from the modulus of its {Fourier} transform}, Optics Letters
  3~(1) (1978) 27.
\newblock \href {http://dx.doi.org/10.1364/OL.3.000027}
  {\path{doi:10.1364/OL.3.000027}}.
\newline\urlprefix\url{https://www.osapublishing.org/abstract.cfm?URI=ol-3-1-27}

\bibitem{Fienup:93}
J.~R. Fienup, J.~C. Marron, T.~J. Schulz, J.~H. Seldin,
  \href{http://opg.optica.org/ao/abstract.cfm?URI=ao-32-10-1747}{Hubble space
  telescope characterized by using phase-retrieval algorithms}, Appl. Opt.
  32~(10) (1993) 1747--1767.
\newblock \href {http://dx.doi.org/10.1364/AO.32.001747}
  {\path{doi:10.1364/AO.32.001747}}.
\newline\urlprefix\url{http://opg.optica.org/ao/abstract.cfm?URI=ao-32-10-1747}

\bibitem{miao1999extending}
J.~Miao, P.~Charalambous, J.~Kirz, D.~Sayre, Extending the methodology of x-ray
  crystallography to allow imaging of micrometre-sized non-crystalline
  specimens, Nature 400~(6742) (1999) 342--344.

\bibitem{maiden2009improved}
A.~M. Maiden, J.~M. Rodenburg, An improved ptychographical phase retrieval
  algorithm for diffractive imaging, Ultramicroscopy 109~(10) (2009)
  1256--1262.

\bibitem{shechtman_phase_2015}
Y.~Shechtman, Y.~C. Eldar, O.~Cohen, H.~N. Chapman, J.~Miao, M.~Segev,
  \href{http://ieeexplore.ieee.org/lpdocs/epic03/wrapper.htm?arnumber=7078985}{Phase
  {Retrieval} with {Application} to {Optical} {Imaging}: {A} contemporary
  overview}, IEEE Signal Processing Magazine 32~(3) (2015) 87--109.
\newblock \href {http://dx.doi.org/10.1109/MSP.2014.2352673}
  {\path{doi:10.1109/MSP.2014.2352673}}.
\newline\urlprefix\url{http://ieeexplore.ieee.org/lpdocs/epic03/wrapper.htm?arnumber=7078985}

\bibitem{balan_signal_2006}
R.~Balan, P.~Casazza, D.~Edidin,
  \href{https://linkinghub.elsevier.com/retrieve/pii/S1063520305000667}{On
  signal reconstruction without phase}, Applied and Computational Harmonic
  Analysis 20~(3) (2006) 345--356.
\newblock \href {http://dx.doi.org/10.1016/j.acha.2005.07.001}
  {\path{doi:10.1016/j.acha.2005.07.001}}.
\newline\urlprefix\url{https://linkinghub.elsevier.com/retrieve/pii/S1063520305000667}

\bibitem{candes_phase_2011}
E.~J. Candès, Y.~C. Eldar, T.~Strohmer, V.~Voroninski,
  \href{https://doi.org/10.1137/110848074}{Phase retrieval via matrix
  completion}, SIAM Journal on Imaging Sciences 6~(1) (2013) 199--225.
\newblock \href {http://arxiv.org/abs/https://doi.org/10.1137/110848074}
  {\path{arXiv:https://doi.org/10.1137/110848074}}, \href
  {http://dx.doi.org/10.1137/110848074} {\path{doi:10.1137/110848074}}.
\newline\urlprefix\url{https://doi.org/10.1137/110848074}

\bibitem{candes_phase_2013}
E.~Candes, X.~Li, M.~Soltanolkotabi,
  \href{http://arxiv.org/abs/1310.3240}{Phase {Retrieval} from {Coded}
  {Diffraction} {Patterns}}, arXiv:1310.3240 [cs, math, stat]ArXiv: 1310.3240.
\newline\urlprefix\url{http://arxiv.org/abs/1310.3240}

\bibitem{bandeira2014phase}
A.~S. Bandeira, Y.~Chen, D.~G. Mixon, Phase retrieval from power spectra of
  masked signals, Information and Inference: a Journal of the IMA 3~(2) (2014)
  83--102.

\bibitem{beinert2015ambiguities}
R.~Beinert, G.~Plonka, Ambiguities in one-dimensional discrete phase retrieval
  from fourier magnitudes, Journal of Fourier Analysis and Applications 21~(6)
  (2015) 1169--1198.

\bibitem{boche_fourier_2017}
T.~Bendory, R.~Beinert, Y.~C. Eldar,
  \href{http://link.springer.com/10.1007/978-3-319-69802-1_2}{Fourier {Phase}
  {Retrieval}: {Uniqueness} and {Algorithms}}, in: H.~Boche, G.~Caire,
  R.~Calderbank, M.~März, G.~Kutyniok, R.~Mathar (Eds.), Compressed {Sensing}
  and its {Applications}, Springer International Publishing, Cham, 2017, pp.
  55--91, series Title: Applied and Numerical Harmonic Analysis.
\newblock \href {http://dx.doi.org/10.1007/978-3-319-69802-1_2}
  {\path{doi:10.1007/978-3-319-69802-1_2}}.
\newline\urlprefix\url{http://link.springer.com/10.1007/978-3-319-69802-1_2}

\bibitem{chipman_polarized_2018}
R.~A. Chipman, G.~Young, W.~S.~T. Lam, Polarized light and optical systems,
  Optical sciences and applications of light, Taylor \& Francis, CRC Press,
  Boca Raton, 2018.

\bibitem{perez_polarized_nodate}
J.~J.~G. Perez, R.~Ossikovski, Polarized {Light} and the {Mueller} {Matrix}
  {Approach}, Polarized Light (2018) 398.

\bibitem{Gordon2000}
J.~P. Gordon, H.~Kogelnik,
  \href{http://www.pnas.org/content/97/9/4541.abstract}{{PMD fundamentals:
  Polarization mode dispersion in optical fibers}}, Proceedings of the National
  Academy of Sciences 97~(9) (2000) 4541--4550.
\newblock \href {http://dx.doi.org/10.1073/pnas.97.9.4541}
  {\path{doi:10.1073/pnas.97.9.4541}}.
\newline\urlprefix\url{http://www.pnas.org/content/97/9/4541.abstract}

\bibitem{tyo2006review}
J.~S. Tyo, D.~L. Goldstein, D.~B. Chenault, J.~A. Shaw, Review of passive
  imaging polarimetry for remote sensing applications, Applied optics 45~(22)
  (2006) 5453--5469.

\bibitem{guo_revealing_2020}
S.-M. Guo, L.-H. Yeh, J.~Folkesson, I.~E. Ivanov, A.~P. Krishnan, M.~G. Keefe,
  E.~Hashemi, D.~Shin, B.~B. Chhun, N.~H. Cho, M.~D. Leonetti, M.~H. Han, T.~J.
  Nowakowski, S.~B. Mehta,
  \href{https://elifesciences.org/articles/55502}{Revealing architectural order
  with quantitative label-free imaging and deep learning}, eLife 9 (2020)
  e55502.
\newblock \href {http://dx.doi.org/10.7554/eLife.55502}
  {\path{doi:10.7554/eLife.55502}}.
\newline\urlprefix\url{https://elifesciences.org/articles/55502}

\bibitem{smirnova_attosecond_2009}
O.~Smirnova, S.~Patchkovskii, Y.~Mairesse, N.~Dudovich, D.~Villeneuve,
  P.~Corkum, M.~Y. Ivanov,
  \href{https://link.aps.org/doi/10.1103/PhysRevLett.102.063601}{Attosecond
  {Circular} {Dichroism} {Spectroscopy} of {Polyatomic} {Molecules}}, Physical
  Review Letters 102~(6) (2009) 063601.
\newblock \href {http://dx.doi.org/10.1103/PhysRevLett.102.063601}
  {\path{doi:10.1103/PhysRevLett.102.063601}}.
\newline\urlprefix\url{https://link.aps.org/doi/10.1103/PhysRevLett.102.063601}

\bibitem{raz_vectorial_2011}
O.~Raz, O.~Schwartz, D.~Austin, A.~S. Wyatt, A.~Schiavi, O.~Smirnova,
  B.~Nadler, I.~A. Walmsley, D.~Oron, N.~Dudovich,
  \href{http://arxiv.org/abs/1104.5086}{Vectorial {Phase} {Retrieval} for
  {Linear} {Characterization} of {Attosecond} {Pulses}}, Physical Review
  Letters 107~(13) (2011) 133902, arXiv: 1104.5086.
\newblock \href {http://dx.doi.org/10.1103/PhysRevLett.107.133902}
  {\path{doi:10.1103/PhysRevLett.107.133902}}.
\newline\urlprefix\url{http://arxiv.org/abs/1104.5086}

\bibitem{ferrand_ptychography_2015}
P.~Ferrand, M.~Allain, V.~Chamard,
  \href{https://www.osapublishing.org/abstract.cfm?URI=ol-40-22-5144}{Ptychography
  in anisotropic media}, Optics Letters 40~(22) (2015) 5144.
\newblock \href {http://dx.doi.org/10.1364/OL.40.005144}
  {\path{doi:10.1364/OL.40.005144}}.
\newline\urlprefix\url{https://www.osapublishing.org/abstract.cfm?URI=ol-40-22-5144}

\bibitem{ferrand_quantitative_2018}
P.~Ferrand, A.~Baroni, M.~Allain, V.~Chamard,
  \href{http://arxiv.org/abs/1712.00260}{Quantitative imaging of anisotropic
  material properties with vectorial ptychography}, Optics Letters 43~(4)
  (2018) 763, arXiv: 1712.00260.
\newblock \href {http://dx.doi.org/10.1364/OL.43.000763}
  {\path{doi:10.1364/OL.43.000763}}.
\newline\urlprefix\url{http://arxiv.org/abs/1712.00260}

\bibitem{baroni_extending_2020}
A.~Baroni, V.~Chamard, P.~Ferrand,
  \href{https://link.aps.org/doi/10.1103/PhysRevApplied.13.054028}{Extending
  {Quantitative} {Phase} {Imaging} to {Polarization}-{Sensitive} {Materials}},
  Physical Review Applied 13~(5) (2020) 054028.
\newblock \href {http://dx.doi.org/10.1103/PhysRevApplied.13.054028}
  {\path{doi:10.1103/PhysRevApplied.13.054028}}.
\newline\urlprefix\url{https://link.aps.org/doi/10.1103/PhysRevApplied.13.054028}

\bibitem{baroni_reference-free_2020}
A.~Baroni, P.~Ferrand,
  \href{https://www.osapublishing.org/abstract.cfm?URI=oe-28-23-35339}{Reference-free
  quantitative microscopic imaging of coherent arbitrary vectorial light
  beams}, Optics Express 28~(23) (2020) 35339.
\newblock \href {http://dx.doi.org/10.1364/OE.408665}
  {\path{doi:10.1364/OE.408665}}.
\newline\urlprefix\url{https://www.osapublishing.org/abstract.cfm?URI=oe-28-23-35339}

\bibitem{jaganathan_reconstruction_2019}
K.~Jaganathan, B.~Hassibi,
  \href{https://ieeexplore.ieee.org/document/8691612/}{Reconstruction of
  {Signals} {From} {Their} {Autocorrelation} and {Cross}-{Correlation}
  {Vectors}, {With} {Applications} to {Phase} {Retrieval} and {Blind} {Channel}
  {Estimation}}, IEEE Transactions on Signal Processing 67~(11) (2019)
  2937--2946.
\newblock \href {http://dx.doi.org/10.1109/TSP.2019.2911254}
  {\path{doi:10.1109/TSP.2019.2911254}}.
\newline\urlprefix\url{https://ieeexplore.ieee.org/document/8691612/}

\bibitem{raz_vectorial_2013}
O.~Raz, N.~Dudovich, B.~Nadler,
  \href{http://ieeexplore.ieee.org/document/6410442/}{Vectorial {Phase}
  {Retrieval} of 1-{D} {Signals}}, IEEE Transactions on Signal Processing
  61~(7) (2013) 1632--1643.
\newblock \href {http://dx.doi.org/10.1109/TSP.2013.2239994}
  {\path{doi:10.1109/TSP.2013.2239994}}.
\newline\urlprefix\url{http://ieeexplore.ieee.org/document/6410442/}

\bibitem{bandeira_saving_2014}
A.~S. Bandeira, J.~Cahill, D.~G. Mixon, A.~A. Nelson,
  \href{https://linkinghub.elsevier.com/retrieve/pii/S1063520313000936}{Saving
  phase: {Injectivity} and stability for phase retrieval}, Applied and
  Computational Harmonic Analysis 37~(1) (2014) 106--125.
\newblock \href {http://dx.doi.org/10.1016/j.acha.2013.10.002}
  {\path{doi:10.1016/j.acha.2013.10.002}}.
\newline\urlprefix\url{https://linkinghub.elsevier.com/retrieve/pii/S1063520313000936}

\bibitem{jaganathan_phase_20152}
K.~Jaganathan, Y.~Eldar, B.~Hassibi, Phase retrieval with masks using convex
  optimization, in: 2015 IEEE International Symposium on Information Theory
  (ISIT), IEEE, 2015, pp. 1655--1659.

\bibitem{heinig1984algebraic}
G.~Heinig, K.~Rost, Algebraic methods for Toeplitz-like matrices and operators,
  Birkh\"{a}user, Basel, 1984.

\bibitem{usevich_variable_2017}
K.~Usevich, I.~Markovsky,
  \href{https://linkinghub.elsevier.com/retrieve/pii/S0304397517302505}{Variable
  projection methods for approximate (greatest) common divisor computations},
  Theoretical Computer Science 681 (2017) 176--198.
\newblock \href {http://dx.doi.org/10.1016/j.tcs.2017.03.028}
  {\path{doi:10.1016/j.tcs.2017.03.028}}.
\newline\urlprefix\url{https://linkinghub.elsevier.com/retrieve/pii/S0304397517302505}

\bibitem{schaefer_measuring_2015}
B.~Schaefer, E.~Collett, R.~Smyth, D.~Barrett, B.~Fraher, Measuring the
  {Stokes} polarization parameters, Am. J. Phys. 75~(2) (2015) 6.

\bibitem{goldstein2017polarized}
D.~H. Goldstein, Polarized light, CRC press, 2017.

\bibitem{markovsky2008structured}
I.~Markovsky, Structured low-rank approximation and its applications,
  Automatica 44~(4) (2008) 891--909.

\bibitem{vandenberghe1996semidefinite}
L.~Vandenberghe, S.~Boyd, Semidefinite programming, SIAM review 38~(1) (1996)
  49--95.

\bibitem{monteiro2003first}
R.~Monteiro, et~al., First-and second-order methods for semidefinite
  programming, Mathematical Programming 97~(1) (2003) 209--244.

\bibitem{beck2017first}
A.~Beck, First-order methods in optimization, SIAM, 2017.

\bibitem{goldstein2014field}
T.~Goldstein, C.~Studer, R.~Baraniuk, A field guide to forward-backward
  splitting with a fasta implementation, arXiv preprint arXiv:1411.3406.

\bibitem{kreutz-delgado_complex_2009}
K.~Kreutz-Delgado, \href{http://arxiv.org/abs/0906.4835}{The {Complex}
  {Gradient} {Operator} and the {CR}-{Calculus}}, arXiv:0906.4835 [math]ArXiv:
  0906.4835.
\newline\urlprefix\url{http://arxiv.org/abs/0906.4835}

\bibitem{candes_phase_2015Wirt}
E.~Candes, X.~Li, M.~Soltanolkotabi,
  \href{http://arxiv.org/abs/1407.1065}{Phase {Retrieval} via {Wirtinger}
  {Flow}: {Theory} and {Algorithms}}, IEEE Transactions on Information Theory
  61~(4) (2015) 1985--2007, arXiv: 1407.1065.
\newblock \href {http://dx.doi.org/10.1109/TIT.2015.2399924}
  {\path{doi:10.1109/TIT.2015.2399924}}.
\newline\urlprefix\url{http://arxiv.org/abs/1407.1065}

\bibitem{jiang_wirtinger_2016}
X.~Jiang, S.~Rajan, X.~Liu,
  \href{http://ieeexplore.ieee.org/document/7572075/}{Wirtinger {Flow} {Method}
  {With} {Optimal} {Stepsize} for {Phase} {Retrieval}}, IEEE Signal Processing
  Letters 23~(11) (2016) 1627--1631.
\newblock \href {http://dx.doi.org/10.1109/LSP.2016.2611940}
  {\path{doi:10.1109/LSP.2016.2611940}}.
\newline\urlprefix\url{http://ieeexplore.ieee.org/document/7572075/}

\bibitem{xu_accelerated_2018}
R.~Xu, M.~Soltanolkotabi, J.~P. Haldar, W.~Unglaub, J.~Zusman, A.~F.~J. Levi,
  R.~M. Leahy, \href{http://arxiv.org/abs/1806.05546}{Accelerated {Wirtinger}
  {Flow}: {A} fast algorithm for ptychography}, arXiv:1806.05546 [eess,
  math]ArXiv: 1806.05546.
\newline\urlprefix\url{http://arxiv.org/abs/1806.05546}

\bibitem{walmsley_characterization_2009}
I.~A. Walmsley, C.~Dorrer,
  \href{https://www.osapublishing.org/aop/abstract.cfm?uri=aop-1-2-308}{Characterization
  of ultrashort electromagnetic pulses}, Advances in Optics and Photonics 1~(2)
  (2009) 308.
\newblock \href {http://dx.doi.org/10.1364/AOP.1.000308}
  {\path{doi:10.1364/AOP.1.000308}}.
\newline\urlprefix\url{https://www.osapublishing.org/aop/abstract.cfm?uri=aop-1-2-308}

\bibitem{gorski2005healpix}
K.~M. Gorski, E.~Hivon, A.~J. Banday, B.~D. Wandelt, F.~K. Hansen, M.~Reinecke,
  M.~Bartelmann, Healpix: A framework for high-resolution discretization and
  fast analysis of data distributed on the sphere, The Astrophysical Journal
  622~(2) (2005) 759.

\bibitem{cox_ideals_1997}
D.~Cox, J.~Little, D.~O'Shea, Ideals, Varieties and Algorithms: An Introduction
  to Computational Algebraic Geometry and Commutative Algebra, 2nd Edition,
  Springer, 1997.

\bibitem{balan2016reconstruction}
R.~Balan, Reconstruction of signals from magnitudes of redundant
  representations: The complex case, Foundations of Computational Mathematics
  16~(3) (2016) 677--721.

\bibitem{qian2016phase}
C.~Qian, N.~D. Sidiropoulos, K.~Huang, L.~Huang, H.~C. So, Phase retrieval
  using feasible point pursuit: Algorithms and cram{\'e}r--rao bound, IEEE
  Transactions on Signal Processing 64~(20) (2016) 5282--5296.

\bibitem{van1994cramer}
A.~Van~den Bos, A cram{\'e}r-rao lower bound for complex parameters, IEEE
  Transactions on Signal Processing [see also Acoustics, Speech, and Signal
  Processing, IEEE Transactions on], 42 (10).

\bibitem{loesch2012cramer}
B.~Loesch, B.~Yang, Cram{\'e}r-rao bound for circular and noncircular complex
  independent component analysis, IEEE transactions on signal processing 61~(2)
  (2012) 365--379.

\bibitem{ollila2008cramer}
E.~Ollila, V.~Koivunen, J.~Eriksson, On the cram{\'e}r-rao bound for the
  constrained and unconstrained complex parameters, in: 2008 5th IEEE Sensor
  Array and Multichannel Signal Processing Workshop, IEEE, 2008, pp. 414--418.

\end{thebibliography}

\end{document}